\documentclass[usenames,dvipsnames]{article}

\usepackage{etex}
 \usepackage{geometry} 
 \usepackage{appendix}
\usepackage{stackrel}
\usepackage{rotating}
 \usepackage[T1]{fontenc}
\usepackage{graphicx}
\usepackage{tabularx}
\usepackage{amsfonts}
\usepackage{cmll}
\usepackage{amsthm}
\usepackage{amsmath}
\usepackage{proof}
\usepackage{mathdots}
\usepackage{hyperref}
\usepackage{bbm}
\usepackage{xcolor}
\usepackage{amssymb}
\usepackage{color}
\usepackage{tikz}
\usetikzlibrary{shapes,snakes}
\usepackage{adjustbox}
\usepackage{tikz-cd}
\usepackage{mathrsfs}
\usepackage[all]{xy}
\usepackage{lscape}
\usepackage{stmaryrd}
\usepackage{mathdots}
\usepackage{bussproofs}
\usepackage{subcaption}

\EnableBpAbbreviations

\newtheorem{remark}{Remark}[section]

\newtheorem{definition}{Definition}[section]

\newtheorem{theorem}{Theorem}

\newtheorem{lemma}{Lemma}[section]

\newtheorem{proposition}{Proposition}[section]

\newcommand{\TT}[1]{\mathtt{#1}}
\newcommand{\D}[1]{\mathscr{#1}}

\newcommand{\CC}[1]{\mathcal{#1}}
\newcommand{\BB}[1]{\mathbb{#1}}

\newcommand{\OV}[1]{\overline{#1}}
\newcommand{\F}[1]{\mathfrak{#1}}

\newcommand{\To}{\Rightarrow}

\definecolor{color0}{HTML}{4682B4}

\begin{document}

\title{Proof nets and the instantiation overflow property}         


\author{Paolo Pistone \\ \texttt{paolo.pistone@uniroma3.it}}
\date{}


\maketitle

\begin{abstract}

Instantiation overflow is the property of those second order types for which all instances of full comprehension can be deduced from instances of atomic comprehension. In other words, a type has instantiation overflow when one can type, by atomic polymorphism, ``expansion terms'' which realize instances of the full extraction rule applied to that type. 
This property was investigated in the case of the types arising from the well-known Russell-Prawitz translation of logical connectives into System $F$, but is not restricted to such types. Moreover, it can be related to functorial polymorphism, a well-known categorial approach to parametricity in System $F$.

In this paper we investigate the instantiation overflow property by exploiting the representation of derivations by means of linear logic proof nets. We develop a geometric approach to instantiation overflow yielding a deeper understanding of the structure of expansion terms and Russell-Prawitz types. Our main result is a characterization of the class of types of the form $\forall XA$, where $A$ is a simple type, which enjoy the instantiation overflow property, by means of a generalization of Russell-Prawitz types.

\end{abstract}

\tableofcontents

\section{Introduction}

In his 1903, \emph{Principles of Mathematics}, Bertrand Russell showed that the connectives $\lnot, \land, \lor, \exists$ can be expressed in the $\To,\forall$-fragment of second order logic. Russell's translation was later extended by Prawitz (\cite{Prawitz1965}) to derivations, providing an embedding of full second order logic into its $\To,\forall$-fragment. 
The Russell-Prawitz translation ($RP$ translation for short) can be described as a method which allows to associate with any connective $\star$ defined by natural deduction introduction rules\footnote{The picture can be extended to the case in which the rule $\star$I$_{i}$ discharges a set of hypotheses, see \cite{StudiaLogica}. } 
$$
\AXC{$\Gamma_{i}$}
\RL{\scriptsize$\star$I$_{i}$}
\UIC{$\star$}
\DP
$$
a formula $RP(\star)=\forall X((\Gamma_{1}\To X)\To \dots \To (\Gamma_{n}\To X)\To X)$\footnote{Where $\Gamma_{i}\To X$ indicates the type $B_{1}\To \dots \To B_{p_{i}}\To X$, for $\Gamma_{i}$ the list $B_{1},\dots, B_{p_{i}}$.} belonging to this fragment. 

%
%

When restricting to intuitionistic logic, the $\To,\forall$-fragment of second order logic corresponds to the polymorphic $\lambda$-calculus or System $F$ (\cite{Girard72,Reynolds74}). The most characteristic rule of this system is the $\forall$-elimination rule
\begin{equation}\label{fulle}\AXC{$\forall XA$}\UIC{$ A[B/X]$}\DP\end{equation}
also called extraction rule, which allows to give type $A[B/X]$, for any type $B$, to a term having type $\forall XA$.
This rule expresses an impredicative comprehension principle and is responsible for the failure, in second order logic, of the subformula principle.


A salient feature of the types of the form $RP(\star)$ (let us call them $RP$ types) is the so-called \emph{instantiation overflow} property, first described in \cite{Ferreira2006}. A type of the form $\forall X A$ has this property when any instance of the full extraction rule \ref{fulle} can be deduced in System $F_{at}$ (\cite{Ferreira2013}), that is, the subsystem of $F$ in which rule \ref{fulle} is replaced by the atomic extraction rule below
 \begin{equation}\label{atome}\AXC{$\forall XA$}\UIC{$ A[Y/X]$}\DP\end{equation}More precisely, the type $\forall XA$ has instantiation overflow when for any second order type $B$, there exists an ``expansion term''  $IO_{A}(B)$ which can be given type $\forall XA\To A[B/X]$ in $F_{at}$. In other words, this property amounts to the possibility, for a given type, to deduce full comprehension from atomic, hence predicative, comprehension.

 The instantiation overflow property of $RP$ types was exploited in \cite{Ferreira2006} and \cite{Ferreira2013} to define a variant of the $RP$ translation based on atomic polymorphism. However, instantiation overflow is not restricted to $RP$ types: in \cite{Ferreira2016} is shown that it holds for all types $\forall XA_{n}$, where $A_{0}=Y\To X$ and $A_{n+1}=A_{n}\To X$.

%

In \cite{StudiaLogica1} instantiation overflow was related to functorial polymorphism (\cite{Bainbridge1990}), by exploiting a well-known connection between the $RP$ translation and dinaturality. We recall that functorial polymorphism is the semantics of System $F$ in which types are interpreted as functors (in a generalized, ``multivariant'', sense, see \cite{EKelly1966}) over a cartesian closed category, and well-typed terms as dinatural transformations between such functors. This semantics was proposed as a formalization of parametric polymorphism, one of the most investigated aspects of System $F$ (see \cite{Moggi1996} for an historical survey on parametricity). In particular, all parametric models of System $F$ are dinatural models, as parametricity implies dinaturality (\cite{Plotkin1993}).

The fact that the type $RP(\star)$ preserves all properties of the original connective $\star$ corresponds to the dinaturality condition for the type $RP(\star)$. In categorial terms, this means that the $RP$ translation preserves universal properties of connectives only in parametric models of System $F$ (\cite{Plotkin1993, Hasegawa2009}). In proof-theoretic terms, this means that the $RP$ translation, while preserving $\beta$-equivalence in all models, preserves $\eta$-equivalence and permuting conversions only up to the equational theory generated by dinaturality (\cite{StudiaLogica}).

When $\forall XA$ is a $RP$ type, the expansion terms $IO_{A}(B)$ realizing instantiation overflow can be described in a ``functorial'' way, by considering the fact that $A$ must be of the form $A_{1}\To \dots \To A_{n}\To X$, where the $A_{i}$ only contain positive occurrences of $X$. Such $A_{i}$ correspond then to covariant endofunctors over the category generated by derivations: given a derivation $u$ of hypothesis $B$ and conclusion $C$, one can construct a derivation $A_{i}(u)$, of hypothesis $A_{i}[B/X]$ and conclusion $A_{i}[C/X]$. Then, for any $B$, one can construct a $F_{at}$ derivation of $\forall XA\To A[B/X]$ as illustrated in figure \ref{IOintro}, by exploiting the functoriality of the $A_{i}$ over the derivation $Elim_{B}$ of hypothesis $B, B_{1},\dots, B_{n}$ and conclusion $Y$ (where $B$ is of the form $\forall \OV Y_{1}(B_{1}\To \forall \OV Y_{2}(B_{2}\To \dots\To \forall \OV Y_{n}(B_{n}\To \forall Y_{n+1}Y)\dots ))$), made only of elimination rules. When $\forall XA$ is the $RP$ translation of disjunction, conjunction or absurdity, such derivations correspond exactly to those described in \cite{Ferreira2013}. 
Moreover, in the equational theory generated by dinaturality, the expansion terms just described are equivalent to the derivations consisting only of one instance of the full extraction rule (\cite{StudiaLogica1}). This means in particular that expansion terms and instances of full extractions have the same denotations in all parametric models of System $F$. In other words, atomic polymorphism and full polymorphism for $RP$ types are indistiguishable modulo dinaturality/parametricity.

\begin{figure}
\begin{center}
\resizebox{0.8\textwidth}{!}{
$$
\AXC{$\forall X(A_{1}\To \dots \To A_{n}\To X)$}
\UIC{$A_{1}[Y/X]\To A_{2}[Y/X] \To \dots \To A_{n}[Y/X]\To Y$}
\AXC{$\stackrel{1}{A_{1}[B/X]}$}
\AXC{$\stackrel{p_{1}}{B_{1}},\dots,\stackrel{p_{n}}{B_{n}}$}
\noLine
\BIC{$A_{1}[Elim_{B}]$}
\noLine
\UIC{$A_{1}[Y/X]$}
\BIC{$A_{2}[Y/X]\To \dots \To A_{n}[Y/X]\To Y$}
\noLine
\UIC{$\ddots$}
\noLine
\UIC{$A_{n}[Y/X]\To Y$}
\AXC{$\stackrel{n}{A_{n}[B/X]}$}
\AXC{$\stackrel{p_{1}}{B_{1}},\dots,\stackrel{p_{n}}{B_{n}}$}
\noLine
\BIC{$A_{n}[Elim_{B}]$}
\noLine
\UIC{$A_{n}[Y/X]$}
\BIC{$Y$}
\doubleLine
\RL{$p_{1},\dots, p_{n}$}
\UIC{$B$}
\doubleLine
\RL{$1,\dots, n$}
\UIC{$A_{1}[B/X]\To \dots \To A_{n}[B/X]\To B$}
\DP
$$}
\end{center}
\caption{Instantiation Overflow for Russell-Prawitz types}
\label{IOintro}
\end{figure}

In this paper we investigate the instantiation overflow property by exploiting, in addition to the functorial intuition, the representation of derivations by means of linear logic proof nets. 
Proof nets can be considered as a unified framework for structural and categorial proof theory, as they provide a well-known bridge between the sequent calculus of linear logic and the language of symmetric monoidal closed as well as $^{*}$-autonomous categories (\cite{Seely1987, Blute1993, Blute1996}), refining a paradigm originating in Lambek's investigations on categories as deductive systems \cite{Lambek1969}. Proof nets for Intuitionistic Multiplicative Linear Logic (without units), $\mathit{IMLL^-}$, essentially correspond to Eilenberg-Kelly-MacLane graphs (see \cite{EKelly1966, Blute1993, Hughes2012}), a graphical formalism playing a central role in several coherence theorems (see \cite{MacLane1971, KellyLaplaza80}). Moreover, $\mathit{IMLL^{-}}$ types (that we call linear types) can be described as multivariant functors over the category generated by proof nets/allowable graphs.

We develop a geometric approach to instantiation overflow yielding a deeper understanding of the structure of expansion terms and Russell-Prawitz types.
Our main result is a characterization of instantiation overflow for the types of the form $\forall XA$, where $A$ is a simple type (theorem \ref{expoexpa}):
we define a class of types which generalize the $RP$ translation and we show that, when $A$ is a simple type, $\forall XA$ has instantiation overflow if and only if it is either derivable or logically equivalent to a product of types belonging to this class.

We use proof nets to investigate the \emph{expansion property} for the types of the linear simply typed $\lambda$-calculus $\lambda_{\multimap}$. A linear type is expansible when, for all $B$, there exists a variable $Y$ and a proof net of hypothesis $A[Y/X]$ and conclusion $A[B/X]$, for some variable $Y$. When considering $RP_{X}$\footnote{When $\forall XA$, is a $RP$ type, we say that $A$ is $RP$ \emph{in $X$} (in short, $RP_{X}$).
} types in $\lambda_{\multimap}$ (called linear $RP_{X}$ types), the  expansion terms, as the one in figure \ref{IOintro}, correspond to proof nets called ``simple expansion graphs''. In figure \ref{IOgraph} is shown the simple expansion graph for the type $(A\multimap B\multimap X)\multimap X$ (associated to the $RP$ translation of the multiplicative conjunction $A\otimes B$).  
Similarly to instantiation overflow, the expansion property is not limited to linear $RP_{X}$ types: for instance the types $C=((X\multimap X)\multimap X)\multimap X$ and $D=(X\multimap X)\multimap (X\multimap X)$ are expansible but are not linear Russell-Prawitz types.

Simple expansion graphs can be defined for any type having an equal number of positive and negative occurrences of a variable $X$. However, such graphs need not be proof nets, that is, satisfy the correction criterion. We show that the correctness of such graphs depends on the possibility of pairing the  occurrences of $X$ following a particular pattern (called an internal pairing). 
This property leads to introduce, for any variable $X$, the class of \emph{generalized Russell-Prawitz types in $X$} ($gRP_{X}$ types), which capture the geometrical properties of $RP_{X}$ types. We prove that a linear type is expansible if and only if it is logically equivalent to a $gRP_{X}$ type. For instance,  the type $C$ above and (as soon as intuitionsitic implication is replaced by linear implication) all types $A_{n}$ introduced in \cite{Ferreira2016}) are $gRP_{X}$; the type $D$ above is not $gRP_{X}$, but logically equivalent to the $gRP_{X}$ type $D'=X\multimap X$.  

The result just stated is actually a bit stronger, as it exploits a strict notion of logical equivalence, called \emph{collapse}, related to Craig interpolation: a type $A$ collapses into a type $B$ when $B$ is an interpolant of a  derivation of $A\multimap A$. For instance, the type $D$ above collapses into the type $D'$. 
Proof net interpolation algorithms are known from the literature (\cite{deGroote1996, Carbone1997}). As our results involve the implicational fragment of some intuitionistic systems, we had to consider the well-known fact that such fragments satisfy interpolation in a weaker form  (see \cite{Kanazawa}). To implement weak proof net interpolation in $\lambda_{\multimap}$, we adapted the algorithm in \cite{deGroote1996}.

The characterization of expansible linear types is extended to the simply typed $\lambda$-calculus $\lambda_{\To}$, by exploiting a folklore linearization argument relating simply typed $\lambda$-terms and proof nets. The characterization of expansible simple types is slightly different as one must consider that, if a type is derivable (that is, if there exists a closed term of that type), then, by weakening, it is also expansible, and that weak interpolation for $\lambda_{\To}$ is sensibly more complex than in the case of $\lambda_{\multimap}$. We prove that a simple type is expansible iff it is either derivable or logically equivalent to the product of a finite family of $gRP_{X}$ types.

We finally adapt these results to $F_{at}$: we show that a suitable extension of the expansion property yields a similar characterization of instantiation overflow for the types of the form $\forall XA$, where $A$ is a simple type: as mentioned above, such types are either derivable or logically equivalent (in $F_{at}$) to the product of a finite family of $gRP$ types (i.e. types of the form $\forall XB$, where $B$ is $gRP_{X}$).


There are many natural questions which are left open by the present investigations. In particular, we do not know whether the instantiation overflow property is decidable (as our characterization depends on the notions of derivability and logical equivalence, which are both undecidable in the case of $F_{at}$), nor how the ideas and techniques here presented can be extended to the case of an arbitrary second order type of the form $\forall XA$. Finally, the relation between $gRP$ types and the $RP$ translation should be investigated in more detail.
 We briefly discuss some of these questions at the end of the paper.
   
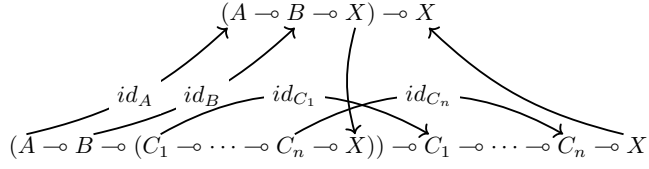
\begin{figure}
\begin{center}
\resizebox{0.6\textwidth}{!}{
\begin{tikzpicture}[every node/.style={fill=white}]

\node(a) at (0,0) {$(A\multimap B\multimap X)\multimap X$};

\node(b) at (0,-2) {$(A\multimap B\multimap (C_{1}\multimap \dots \multimap C_{n}\multimap X))\multimap C_{1}\multimap \dots \multimap C_{n}\multimap X$};

\draw[->, thick] (4.4,-1.8) to [bend left=15] (1.5,-0.2);

\draw[<-, thick] (0.4,-1.8) to [bend left=15] (0.4,-0.2);

\draw[->, thick] (-3.5,-1.8) to [bend right=15] node {$id_{B}$} (-0.5,-0.2);
\draw[->, thick] (-4.5,-1.8) to [bend right=15] node {$id_{A}$} (-1.5,-0.2);

\draw[->, thick] (-2.5,-1.8) to [bend left=30] node {$id_{C_{1}}$} (1.5,-1.8);
\draw[->, thick] (-0.5,-1.8) to  [bend left=30] node {$id_{C_{n}}$} (3.5,-1.8);

\end{tikzpicture}
}
\end{center}
\caption{Simple expansion graph for the $RP$ translation of $A\otimes B$}
\label{IOgraph}
\end{figure}

\paragraph{Related work}
The instantiation overflow phenomenon was first introduced in \cite{Ferreira2006} and later 
investigated in \cite{Ferreira2013} as a property of the Russell-Prawitz translation of disjunction. In particular, it is shown there that the usual Russell-Prawitz translation of logical connectives into $F$ can be transformed into a translation into $F_{at}$ by exploiting instantiation overflow.
Similar results were independently proved in \cite{Sandqvist}. 

The first investigation on the general class of formulas enjoying instantiation overflow is in \cite{Ferreira2016}, where  the ``Prawitz formulas of level $n$'' are introduced. The results are the following: (1) Prawitz formulas of level 2 have instantiation overflow, (2) there exist Prawitz formulas of arbitrary level (the formulas $A_{n}$ mentioned above) having instantiation overflow, (3) the formula $X\To Y$ does not have instantiation overflow. Such results can be deduced from our characterization, since (1) Prawitz formulas of level 2 correspond to Russell-Prawitz types, (2) the formulas $A_{n}$ correspond to generalized Russell Prawitz types and (3) $X\To Y$ is not logically equivalent to any product of generalized Russell-Prawitz types.

As already mentioned, the functorial formulation of instantiation overflow as well as the result that the instantiation overflow derivations are equivalent to instances of full extraction modulo dinaturality first appeared in \cite{StudiaLogica1} and will appear in a sequel paper to the journal version \cite{StudiaLogica}. These papers present the functorial interpretation of Russell-Prawitz types and their relation with dinaturality within a natural deduction frame.

Categories of allowable graphs are well-known in the literature since \cite{MacLane1971} and are used to establish coherence results (see \cite{KellyLaplaza80}). Several proof net formalisms for $\mathit{IMLL}$ have been used to establish coherence for symmetric monoidal categories, $^{*}$-autonomous categories and weakly distributive categories (\cite{Blute1996, Lamarche2004}). The main technical delicacy in such approaches involves the treatment of multiplicative unities $1$ and $\bot$. For this reason we limited ourselves to the system $\mathit{IMLL^-}$ and its $\multimap$-fragment $\lambda_{\multimap}$. \cite{Hughes2012} shows that this approach can be extended to treat $\bot$, yielding a representation of the free $^{*}$-autonomous category. Following \cite{Hughes2017}, allowable graphs for $\mathit{IMLL^{-}}$ should yield a representation of the free \emph{symmetric semi-monoidal closed category}.
Our category $\CC A$ of allowable graphs essentially follows \cite{Hughes2012, Hughes2017}. A major difference is that we define shapes as rooted DAGs, incorporating the correctness criterion for Lamarche essential nets (\cite{Lamarche, Ong}). 


Interpolation for linear logic and proof nets was investigated in \cite{Roorda}, \cite{deGroote1996} and \cite{Carbone1997}. Weak interpolation for the implicational fragment of intuitionistic logic was investigated in \cite{Wronski, Pentus, Kanazawa}. 
Our proof of weak interpolation for $\lambda_{\multimap}$  (in appendix \ref{appA}) is essentially a variant of the one in \cite{deGroote1996}.

\paragraph{Structure of the paper}

The paper can be subdivided in two parts. The first part, from section \ref{sec2} to section \ref{sec5}, is preliminary to the treatment of instantiation overflow: we first introduce type systems, proof-nets and their categorial and functorial interpretations, and then we discuss proof net interpolation and some useful applications. The second part, from section \ref{sec6} to section \ref{sec7}, is devoted to Russell-Prawitz types and the characterization of the expansion property and instantiation overflow.

More in detail, in section \ref{sec2} we recall the four type systems ($\lambda_{\multimap},\lambda_{\To}, F,F_{at}$) used in the paper and we describe the syntactic categories they generate as well as their functorial interpretations. Moreover, we introduce a graphical representation of linear terms through a category of allowable graphs (similarly to \cite{Hughes2012}), corresponding to essential nets (\cite{Lamarche, Ong}). 
In section \ref{sec4} we recall previous results on interpolation in $IMLL^{-}$ and we prove a weak interpolation result for the fragment $\lambda_{\multimap}$. Then we exploit this result to prove the positivity lemma \ref{functor}, a fundamental result which allows to extract, through interpolation, a type containing only positive occurrences of a variable from any type for which a ``functorial'' action on arrows is defined. 
In section \ref{sec5} we extend these results to $\lambda_{\To}$, by exploiting a linearization theorem. 

In section \ref{sec6} we describe instantiation overflow and Russell-Prawitz types and their relationship with functorial polymorphism. We also introduce generalized Russell-Prawitz types and the expansion property, which are investigated in the last two sections.
In section \ref{sec7} we investigate the expansion property for linear types. We prove our first ``density theorem'':  a linear type is expansible iff it collapses into a linear generalized Russell-Prawitz type. This section contains our geometrical investigation of instantiation overflow through simple expansion graphs.
In section \ref{sec8} we prove a similar ``density theorem'' for simple types and we apply it to characterize simple types enjoying instantiation overflow. 

 Finally, in section \ref{sec9} we discuss some open problems and further directions.

\section{$\lambda$-terms, proof nets and categories}\label{sec2}

We recall the type systems which will be used in the paper and we introduce proof nets for {Intuitionistic Multiplicative Linear Logic without units} $\mathit{IMLL^-}$, by defining a category $\CC A$ of \emph{allowable graphs} similarly to \cite{Hughes2012}. Then, we recall the syntactic categories generated by simply typed $\lambda$-terms and System $F$ typable $\lambda$-terms and their functorial interpretation, which will be exploited in section \ref{sec6} to describe the instantiation overflow property for Russell-Prawitz types.

\subsection{Type systems}

We introduce the four type systems which will be used throughout the text:
\begin{itemize}
\item the \emph{simply typed $\lambda$-calculus} $\lambda_{\To}$;
\item the \emph{linear simply typed $\lambda$-calculus} $\lambda_{\multimap}$;
\item the \emph{polymorphic $\lambda$-calculus} or \emph{System $F$} (\cite{Girard72,Reynolds74});
\item the \emph{atomically polymorphic $\lambda$-calculus} or \emph{System $F_{at}$} (\cite{Ferreira2013}).

\end{itemize}

Given a basic set of types $\CC T$, built over a set of variables $\CC V$, we will consider two notions of $\lambda$-terms:
\begin{enumerate}

\item \emph{$\lambda$-terms}, defined by the grammar below
$$
t,u \ := \ x\mid tu \mid \lambda x^{A}.t $$
where $A\in \CC T$; $t$ is linear in $x$ if $x$ occurs exactly once free in $t$;

\item \emph{$\lambda^{2}$-terms}, defined by the grammar below
$$
t,u \ := \ x\mid tu \mid \lambda x^{A}.t \mid \Lambda X.t \mid tA  $$
where $A\in \CC T$ and $X\in \CC V$.

\end{enumerate}
Observe that the definitions above depend on the choice of $\CC T$. This dependence will be often omitted, if it can be deduced from the context. 

$\lambda$-terms and $\lambda^{2}$-terms are considered up to renaming of bound variables, as usual. Given a $\lambda$-term (resp. $\lambda^{2}$-term) $u$, we let $FV(u)$ indicate the set of its free term (resp. term and type) variables, and $BV(u)$ indicate the set of its bound term (resp. term and type) variables.

For $\lambda$-terms and $\lambda^{2}$-terms we let $\simeq_{\beta\eta}$ indicate usual $\beta\eta$-equivalences, generated by the schemas in figure \ref{beta}. By a \emph{normal} $\lambda$-term (resp. $\lambda^{2}$-term) we indicate a term to which no $\beta$-reduction can be applied. Following \cite{Baren95}, by a \emph{$\lambda$-theory} (resp. a $\lambda^{2}$-theory) we indicate any set $T$ of equations over $\lambda$ (resp. $\lambda^{2}$) terms such that $T^{+}=T$, where $T^{+}$ is obtained by adding the $T$ $\beta$ and $\eta$ equivalence as well as the usual axioms and rules of the $\lambda$-calculus.

\begin{figure}[t]
\begin{center}
\resizebox{0.5\textwidth}{!}{
${\begin{matrix}
(\lambda x.t)u \ \simeq_{\beta} \ t[u/x]  &  \ & \ &
(\Lambda X.t)B \ \simeq_{\beta} \ t[B/X] \\ \ \\
\lambda x.tx \ \simeq_{\eta} \ t \ (x\notin FV(t))  &   \ & \ &
\Lambda X.tX \ \simeq_{\eta} \ t \ (X\notin FV(t)) 
\end{matrix}}$}\end{center}
\caption{$\beta$ and $\eta$ equivalences}
\label{beta}
\end{figure}



For any normal $\lambda$-term $u$, we define the set $Subt(u)$ of its subterms as follows: if $u=\lambda x_{1}.\dots.\lambda x_{n}.y$ has no application, then $Subt(u)=\{\lambda x_{i}.\lambda x_{i+1}.\dots.\lambda x_{n}.y\mid 1\leq i\leq n\}$; otherwise, $M=\lambda x_{1}.\dots.\lambda x_{n}. yu_{1}\dots u_{p}$, then $Subt(u)=\{\lambda x_{i}.\lambda x_{i+1}.\dots.\lambda x_{n}.yu_{1}\dots u_{p}\mid i\leq i\leq n\}\cup Subt(u_{1})\cup \dots \cup Subt(u_{p})$. We call a subterm $v\in Subt(u)$ \emph{proper} if $v\neq u$.

We introduce now the type systems. By a \emph{context} in $\CC T$ we indicate a list $\Gamma$ of type declarations $x_{1}:A_{1},\dots, x_{n}:A_{n}$, where the $A_{i}$ are types of  $\CC T$ and the $x_{i}$ are pairwise distinct term variables. We will indicate contexts as $\Gamma,\Delta,\dots$. Concatenation of contexts is indicated by comma $\Gamma,\Delta$.

All systems below include the \emph{exchange rule} $Ex$
$$\AXC{$\Gamma\vdash u:A$}
\RL{$Ex$}
\UIC{$\sigma\Gamma\vdash u:A$}
\DP$$
where $\Gamma$ is a context in $\CC T$, $A$ is a type in $\CC T$ and $\sigma\Gamma$ indicates a context obtained from $\Gamma$ by permuting the order of its elements.

By a \emph{partition} of a context $\Gamma$, we indicate a list $\Gamma_{1},\dots, \Gamma_{p}$ of contexts such that $\Gamma_{1},\dots,\Gamma_{p}=\sigma\Gamma$.

\begin{description}

\item[($\lambda_{\multimap}$)] the set of linear types $\CC L_{\multimap}$ is generated by the grammar $A,B:= X\mid  A\multimap B$; the typing rules for linear $\lambda$-terms are $Ex$ and those shown in figure \ref{stllambda};

\item[($\lambda_{\To}$)] the set of simple types $\CC L_{\To}$ is generated by the grammar $A,B:= X\mid  A\To B$; the typing rules for $\lambda$-terms are $Ex, W$ and those in figure \ref{stlambda};

\item[($F$)] the set of second order types $\CC L_{\To, \forall}$ is generated by the grammar $A,B:= X\mid A\To B\mid \forall XA$; the typing rules for $\lambda^{2}$-terms are those of $\lambda_{\To}$ plus those shown in figure \ref{F};

\item[($F_{at}$)] same types as $F$; the typing rules for $\lambda^{2}$-terms are those of $\lambda_{\To}$ plus those shown in figure \ref{Fat};
\end{description}


%


\begin{figure}
\begin{center}
\begin{subfigure}{0.30\textwidth}
\resizebox{\textwidth}{!}{
$\boxed{\begin{matrix}
\AXC{$x:A\vdash x:A$}\DP \\ \ \\
\AXC{$\Gamma, x:A\vdash u:B$}\AXC{$u$ linear in $x$}
\RL{\scriptsize$\multimap$I}
\BIC{$\Gamma\vdash \lambda x^{A}.u:A\multimap B$}\DP \\ \ \\
\AXC{$\Gamma\vdash u:A\multimap B$}
\AXC{$\Delta\vdash v:A$}\RL{\scriptsize$\multimap$E}
\BIC{$\Gamma, \Delta\vdash uv:B$}\DP 
\end{matrix}
}$}
\caption{System $\lambda_{\multimap}$}
\label{stllambda}
\end{subfigure}
\begin{subfigure}{0.30\textwidth}
\resizebox{\textwidth}{!}{
$
\boxed{\begin{matrix}
\AXC{$\Gamma, x:A\vdash x:A$}\DP \\ \ \\
\AXC{$\Gamma, x:A\vdash u:B$}
\RL{\scriptsize$\To$I}
\UIC{$\Gamma\vdash \lambda x^{A}.u:A\To B$}\DP \\ \ \\
\AXC{$\Gamma\vdash u:A\To B$}
\AXC{$\Gamma\vdash v:A$}
\RL{\scriptsize$\To$E}
\BIC{$\Gamma\vdash uv:B$}\DP 
\end{matrix}
}$}
\caption{System $\lambda_{\To}$}
\label{stlambda}
\end{subfigure}
\end{center}
\begin{center}
\begin{subfigure}{0.30\textwidth}
\begin{center}
\resizebox{0.9\textwidth}{!}{$
\boxed{\begin{matrix}
\AXC{$\Gamma\vdash u:A$} 
\AXC{$X\notin FV(\Gamma)$}
\RL{\scriptsize$\forall I$}
\BIC{$\Gamma\vdash \Lambda X.u:\forall XA$}
\DP 
\\ \ \\
\AXC{$\Gamma\vdash u:\forall XA$}
\RL{\scriptsize$\forall E$}
\UIC{$\Gamma\vdash uB:A[B/X]$}
\DP
\end{matrix}
}
$}\end{center}
\caption{System $F$}
\label{F}
\end{subfigure}
\begin{subfigure}{0.30\textwidth}
\begin{center}
\resizebox{0.9\textwidth}{!}{$
\boxed{
\begin{matrix}
\AXC{$\Gamma\vdash u:A$} 
\AXC{$X\notin FV(\Gamma)$}
\RL{\scriptsize$\forall I$}
\BIC{$\Gamma\vdash \Lambda X.u:\forall XA$}
\DP 
\\ \ \\
\AXC{$\Gamma\vdash M:\forall XA$}
\RL{\scriptsize$F_{at} E$}
\UIC{$\Gamma\vdash uY:A[Y/X]$}
\DP
\end{matrix}
}
$}\end{center}
\caption{System $F_{at}$}
\label{Fat}
\end{subfigure}
\end{center}
\caption{Type systems rules}
\end{figure}

Observe that the usual rules of contraction and weakening are derivable in $\lambda_{\To}, F, F_{at}$.
For any type $A$ in any of the systems above, we let $FV(A)$ (resp. $BV(A))$) indicate the set of its free (resp. bound) variables.
There exist obvious inverse translations from $^{\To}:\CC L_{\multimap}\to \CC L_{\To}$ and $^{\multimap}:\CC L_{\To}\to \CC L_{\multimap}$, given by $X^{\To}=X$, $X^{\multimap}=X$, $(A\multimap B)^{\To}=A^{\To}\To B^{\To}$ and $(A\To B)^{\multimap}=A^{\multimap}\multimap B^{\multimap}$. If $\Gamma\vdash u:A$ is derivable in $\lambda_{\multimap}$, then $\Gamma^{\To}\vdash u:A^{\To}$ is derivable in $\lambda_{\To}$, where $\Gamma^{\To}=x_{1}:A_{1}^{\To},\dots,x_{n}:A_{n}^{\To}$, for $\Gamma=x_{1}:A_{1},\dots, x_{n}:A_{n}$.

 A type $B\in \CC L_{\To, \forall}$ will be generally written $\forall \OV Y_{1}(B_{1}\To \forall \OV Y_{2}(B_{2}\To \dots \To \forall \OV Y_{n}Z))$, where $\forall \OV Y_{i}$ is shorthand for a finite, possibly empty, sequence of quantifications $\forall Y_{i_{1}}\dots \forall Y_{i_{k_{i}}}$. 

Given any of the systems above,
we say that a type $A$ is \emph{derivable} if there exists a closed term $u$ having type $A$. We say that two types $A,B$ are \emph{logically equivalent} if there exist closed terms $u,v$ having type $A\multimap B, B\multimap A$, respectively, in the case of $\lambda_{\multimap}$, and $A\To B, B\To A$, respectively, in all other cases. If, moreover, $\lambda x^{A}.v(ux)\simeq_{\beta\eta} \lambda x^{A}.x $ and $\lambda x^{B}.u(vx)\simeq_{\beta\eta}\lambda x^{B}.x$, then $A$ and $B$ are called \emph{isomorphic}. Finally, given types $A, B_{1},\dots, B_{n}$, we say that $A$ is \emph{logically equivalent to the product of $B_{1},\dots, B_{n}$} when there exist closed terms $u_{1},\dots, u_{n},u$ having types $A\multimap B_{1},\dots, A\multimap B_{n}, B_{1}\multimap \dots \multimap B_{n}\multimap A$, respectively (in the case of $\lambda_{\multimap}$) and types $A\To B_{1},\dots, A\To B_{n}, B_{1}\To \dots \To B_{n}\To A$, respectively, in all other cases.

In any of the systems above, given a normal $\lambda$-term $u$ such that $\Gamma\vdash u:A$, we introduce the following terminology:

\begin{itemize}
\item[$i.$] any variable $x$ occurring free or bound in $u$ is assigned a unique type $A$ that we indicate by $[x]$;
\item[$ii.$] any $v\in Subt(u)$ is assigned a unique type, that we indicate by $[v]$;
\item[$iii.$] $u$ is said in \emph{$\eta$-long normal form} when for any $v\in Subt(u)$, if $[v]=B\To C$ or $[v]=B\multimap C$, then $v=\lambda x^{B}.v'$, for some variable $x$ and term $v'\in Subt(u)$. 

\end{itemize}

Observe that, if $u$ is in $\eta$-long normal form, then for any type $B$ occurring positively (resp. negatively) in $A$ there exists $v\in Subt(u)$ (resp. $x\in BV(u)$) such that $[u]=B$ (resp. $[x]=B$).

%
%
%

\subsection{Proof nets and the category of allowable graphs}\label{sec22}

We introduce proof nets for Intuitionistic Multiplicative Linear Logic without units, $\mathit{IMLL^-}$, that is, the system obtained by adding to $\lambda_{\multimap}$ the $\otimes$ connective. Typed $\lambda$-calculi for $\mathit{IMLL}$ can be found in the literature (see \cite{Abramsky1993, Benton1993}). 

We recall that proof nets for $\mathit{IMLL}$ and its subsystems can be
 considered as a graphical representation of $\lambda$-terms or as a graphical formalism for arrows in free monoidal closed categories. Our definition merges the two viewpoints: on the one hand, our definition corresponds to the usual definition of \emph{essential nets} for $\mathit{IMLL^-}$ (\cite{Lamarche, Ong}), characterizing linearly typable $\lambda$-terms; on the other hand, we introduce proof-structures by means of a category of graphs following \cite{Hughes2012}; in particular, the correction criterion of essential nets generates the sub-category of \emph{allowable graphs}. 
 

We first define a category $\CC G$ of graphs, which are defined as certain morphisms between signed sets, i.e. sets whose elements are assigned a polarity $+,-$. Then we introduce shapes as certain rooted $DAG$s whose leaves form a signed set and we define a category $\CC A$ of allowable graphs, corresponding to graphs in $\CC G$ satisfying the correction criterion. \cite{Hughes2012} shows that this category can be extended to treat $\bot$, yielding a representation of the free $^{*}$-autonomous category. Following \cite{Hughes2017}, $\CC A$ might be seen as a representation of the free \emph{symmetric semi-monoidal closed category}.

%
%
%
%

The objects of $\CC G$ are \emph{signed sets}, i.e. finite sets whose elements are assigned a polarity $\epsilon\in \{+,-\}$ (i.e. edges $(s,\alpha)$, where $\alpha:s\to\{+,-\}$); 
Given $\epsilon\in \{+,-\}$, we let $\OV \epsilon$ be the opposite polarity; given a signed set $s$, we let $\OV s$ be the signed set whose underlying set is the same as $s$ and whose polarities are reversed. Given two signed sets $s,t$, we let $s+t$ denote their disjoint union.

Arrows $f:s\to t$ in $\CC G$, called \emph{graphs}, are bijections $s^{+}+t^{-} \to s^{-} + t^{+}$ (where $+$ indicates disjoint union).
Equivalently, a graph $f:s\to t$ is a set of disjoint \emph{edges}, i.e. disjoint edges of elements of $\OV s+ t$ which can be of three types:
\begin{description}
\item[Type I:] $e=(x^{\epsilon},y^{\OV\epsilon})$, where $x\in s$, $y\in t$, for $\epsilon\in \{+,-\}$;
\item[Type II:] $e=(x^{+},y^{-})$, where $x,y\in \OV s$;
\item[Type III:] $e=(x^{+},y^{-})$, where $x,y\in t$.

\end{description}
A graph can be illustrated as a directed acyclic graph (as in fig \ref{fig2a}) by orienting edges from positive to negative. We will call a graph \emph{pure} when it only consists of type I edges.

Composition of graphs is finite directed path composition (see \cite{Hughes2012}), as illustrated in figure \ref{fig2c}. More precisely, given $f:s\to t$ and $g:t\to u$, $g\circ f$ is the bijection $h:s^{+}+u^{-}\to s^{-}+ u^{+}$ where 
$$
h(x)= 
\begin{cases}
(f\circ g)^{n}f(x) & \text{ if }x\in s^{+} \text{ and } n \text{ minimum s.t.} (f\circ g)^{n}f(x)\in s^{-}+u^{+} \\
(g\circ f)^{n}g(x) & \text{ if }x\in u^{-} \text{ and } n \text{ minimum s.t.} (g\circ f)^{n}g(x)\in s^{-}+u^{+} 
\end{cases}
$$
The definition of $g\circ f$ relies on the following:
\begin{lemma}
If $f:s\to t$ and $g:t\to u$, then for each $x\in s^{+}$ (resp. $y\in u^{-}$) there exists an $n$ such that $(f\circ g)^{n}f(x)\in s^{-}+u^{+}$ (resp. $(g\circ f)^{n}g(y)\in s^{-}+u^{+}$). 

\end{lemma}
\begin{proof}
By a maximal chain in $f\circ g$ we indicate a sequence $e_{0},\dots, e_{2k-1}$ of even length obtained by alternating a type III edge $e_{2i}\in f$ and a type II edge $e_{2i+1}\in g$ such that, for $i\leq k-1$, if $e_{2i}$ is $(y^{\epsilon'},z^{\epsilon})$, then $e_{2i+1}$ is $(z^{\OV\epsilon}, w^{\epsilon''})$ and, moreover, $e_{2k-1}=(y^{\epsilon}, z^{\epsilon})$, where $z\in s^{-}+u^{+}$. If $N$ is the cardinality of $t^{+}+t^{-}$, then any chain in $f\circ g$ must have length $<N$. Now, if $x\in s^{+}$, then either $f(x)\in s^{-}$ (then put $n=0$), or $f(x)\in t^{+}$ is the start of a chain in $f\circ g$. Then for some $n< N$, the chain ends in some $z\in s^{-}+u^{+}$, and then $(f\circ g)^{n}f(x)=z$.

\end{proof}

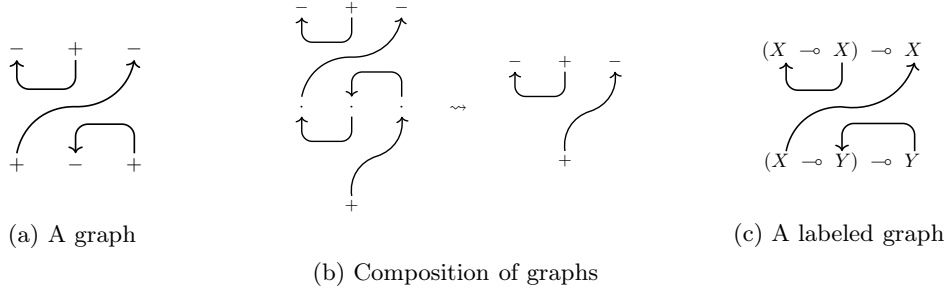
\begin{figure}
%
%
%
%
%
%
%
%
%
%
%
%
%
%
%
%
%
%
%
%
%
%
%
%
%
%
%
%
%
%
%
%
%
%
%
%
%
%
%
%
\begin{center}
\begin{subfigure}{0.23\textwidth}
\begin{center}
\resizebox{0.6\textwidth}{!}{
\begin{tikzpicture}
\node(b) at (-1,0) {$-$};
\node(c) at (0,0) {$+$};
\node(d) at (1,0) {$-$};

\node(b') at (-1,-2) {$+$};
\node(c') at (0,-2) {$-$};
\node(d') at (1,-2) {$+$};

\draw[<-, thick, rounded corners=6pt] (-1,-0.2) to (-1,-0.7) to (0,-0.7) to (0,-0.2);
\draw[<-, thick, rounded corners=6pt] (0,-1.8) to (0,-1.3) to (1,-1.3) to (1,-1.8);
\draw[<-, thick] (1,-0.2) to [bend left=40]  (0,-1) to [bend right=40] (-1,-1.8);

\end{tikzpicture}}
\end{center}
\caption{A graph}
\label{fig2a}
\end{subfigure}
\begin{subfigure}{0.44\textwidth}
\begin{center}
\resizebox{0.7\textwidth}{!}{
$\begin{matrix}
\begin{tikzpicture}[baseline=-14ex]
\node(b) at (-1,0) {$-$};
\node(c) at (0,0) {$+$};
\node(d) at (1,0) {$-$};

\node(b') at (-1,-2) {$\cdot$};
\node(c') at (0,-2) {$\cdot$};
\node(d') at (1,-2) {$\cdot$};

\node(b'') at (0,-4) {$+$};

\draw[<-, thick, rounded corners=6pt] (-1,-0.2) to (-1,-0.7) to (0,-0.7) to (0,-0.2);
\draw[<-, thick, rounded corners=6pt] (0,-1.8) to (0,-1.3) to (1,-1.3) to (1,-1.8);
\draw[<-, thick] (1,-0.2) to [bend left=40]  (0,-1) to [bend right=40] (-1,-1.8);
\draw[<-, thick, rounded corners=6pt] (-1,-2.2) to (-1,-2.7) to (0,-2.7) to (0,-2.2);
\draw[<-, thick] (1,-2.2) to [bend left=40]  (0.5,-3) to [bend right=40] (0,-3.8);

\end{tikzpicture} 
& \ & \leadsto & \ & 
\begin{tikzpicture}[baseline=-20ex]

\node(b') at (-1,-2) {$-$};
\node(c') at (0,-2) {$+$};
\node(d') at (1,-2) {$-$};

\node(b'') at (0,-4) {$+$};

\draw[<-, thick, rounded corners=6pt] (-1,-2.2) to (-1,-2.7) to (0,-2.7) to (0,-2.2);
\draw[<-, thick] (1,-2.2) to [bend left=40]  (0.5,-3) to [bend right=40] (0,-3.8);

\end{tikzpicture} 
\end{matrix}$}
\end{center}
\caption{Composition of graphs}
\label{fig2c}
\end{subfigure}
\begin{subfigure}{0.23\textwidth}
\begin{center}
\resizebox{0.7\textwidth}{!}{
\begin{tikzpicture}
\node(b) at (0,0) {$(X \ \multimap \ X) \ \multimap \ X$};

\node(c) at (0,-2) {$(X \ \multimap \ Y) \ \multimap \ Y$};

\draw[<-, thick, rounded corners=6pt] (-1,-0.2) to (-1,-0.7) to (0,-0.7) to (0,-0.2);
\draw[<-, thick, rounded corners=6pt] (0,-1.8) to (0,-1.3) to (1.3,-1.3) to (1.3,-1.8);
\draw[<-, thick] (1.3,-0.2) to [bend left=40]  (0,-1) to [bend right=40] (-1,-1.8);

\end{tikzpicture}}
\end{center}
\caption{A labeled graph }
\label{labgraph}
\end{subfigure}
\end{center}
\caption{Examples of graphs and labeled graphs}
\end{figure}

In order to introduce allowable graphs, we first define shapes. A shape corresponds to the switching of the syntactic tree of a $\mathit{IMLL}^{-}$ type. Hence, on the one hand the leaves of the shape form a signed set, so that an arrow between two shapes corresponds to a graph between the associated signed sets; on the other hand, by joining the graph with the shapes, we obtain a correction graph on which we can check the essential nets correction criterion (\cite{Ong}).

\begin{definition}[shape]
A shape $S$ is a rooted and labeled DAG whose leaves form a signed set $vS$, called the \emph{variable set} of $S$. The nodes of a shape are either leaves (hence labeled by $+$ or $-$) or labeled by $\multimap^{+}$ (resp. $\otimes^{+}$) or $\multimap^{-}$ (resp. $\otimes^{-}$). The root of $S$ is called the \emph{conclusion} $cS$ of $S$. If $S$ is a shape, by $\OV S$ we indicate the shape obtained from $S$ by reversing the sign of its leaves and replacing all labels $\multimap^{\epsilon}$ - resp. $\otimes^{\epsilon}$-  by $\multimap^{\OV \epsilon}$ - resp. $\otimes^{\OV\epsilon}$. Shapes are defined inductively as follows:

\begin{itemize}
\item $I$ is the shape $+$, $vI=\{+\}$, $cI=+$;

\item if $S,T$ are shapes, $S\multimap^{+} T$ is the shape in figure \ref{shapedef1}, 
$v(S\multimap^{+}T)= v\OV S+vT$ and $c(S\multimap^{+}T)=\multimap^{+}$; $c\OV S$ and $cT$ are called, respectively, \emph{left} and \emph{right premiss} of the node $c(S\multimap^{+}T)$.

\item if $S,T$ are shapes, $S\multimap^{-} T$ is the shape in figure \ref{shapedef2}, 
$v(S\multimap^{-}T)= vS+v\OV T$ and $c(S\multimap^{-}T)=\multimap^{-}$; $cS$ and $c\OV T$ are called, respectively, \emph{left} and \emph{right premiss} of the node $c(S\multimap^{-}T)$.

\item if $S,T$ are shapes, $S\otimes^{+} T$ is the shape in figure \ref{shapedef3}, 
$v(S\otimes^{+}T)= vS+vT$ and $c(S\otimes^{+}T)=\otimes^{+}$; $cS$ and $cT$ are called, respectively, \emph{left} and \emph{right premiss} of the node $c(S\otimes^{+}T)$.

\item if $S,T$ are shapes, $S\otimes^{-} T$ is the shape in figure \ref{shapedef4}, 
$v(S\otimes^{-}T)= v\OV S+ v\OV T$ and $c(S\otimes^{-}T)=\otimes^{-}$; $c\OV S$ and $c\OV T$ are called, respectively, \emph{left} and \emph{right premiss} of the node $c(S\otimes^{-}T)$.

%
%
%
%

\end{itemize}

\end{definition}

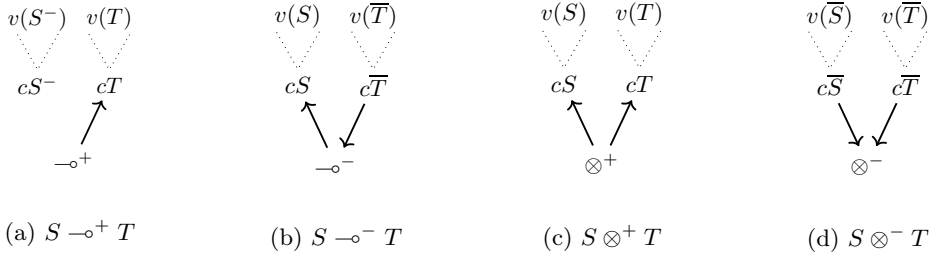
\begin{figure}
\begin{subfigure}{0.23\textwidth}
\begin{center}
\resizebox{0.6\textwidth}{!}{
\begin{tikzpicture}[baseline=2ex, scale=0.8]
\node(a) at (0,-0.5) {$\multimap^{+}$};
\node(b) at (-0.7,1) {$cS^{-}$};
\node(c) at (0.7,1) {$cT$};

\draw[dotted] (-0.3,2) to (-0.7,1.3) to (-1.1,2);

\node(v) at (-0.7,2.3) {$v(S^{-})$};
\node(v') at (0.7,2.3) {$v(T)$};
\draw[dotted] (0.3,2) to (0.7,1.3) to (1.1,2);

\draw[->, thick] (a) to (c);

\end{tikzpicture}
}
\end{center}
\caption{$S\multimap^{+}T$}
\label{shapedef1}
\end{subfigure}
\begin{subfigure}{0.23\textwidth}
\begin{center}
\resizebox{0.6\textwidth}{!}{
\begin{tikzpicture}[baseline=2ex, scale=0.8]
\node(a) at (0,-0.5) {$\multimap^{-}$};
\node(b) at (-0.7,1) {$cS$};
\node(c) at (0.7,1) {$c\OV T$};

\draw[dotted] (-0.3,2) to (-0.7,1.3) to (-1.1,2);

\node(v) at (-0.7,2.3) {$v(S)$};
\node(v') at (0.7,2.3) {$v(\OV T)$};
\draw[dotted] (0.3,2) to (0.7,1.3) to (1.1,2);

\draw[<-, thick] (a) to (c);
\draw[<-, thick] (b) to (a);

\end{tikzpicture}
}
\end{center}
\caption{$S\multimap^{-}T$}
\label{shapedef2}
\end{subfigure}
\begin{subfigure}{0.23\textwidth}
\begin{center}
\resizebox{0.6\textwidth}{!}{
\begin{tikzpicture}[baseline=2ex, scale=0.8]
\node(a) at (0,-0.5) {$\otimes^{+}$};
\node(b) at (-0.7,1) {$cS$};
\node(c) at (0.7,1) {$cT$};

\draw[dotted] (-0.3,2) to (-0.7,1.3) to (-1.1,2);

\node(v) at (-0.7,2.3) {$v(S)$};
\node(v') at (0.7,2.3) {$v(T)$};
\draw[dotted] (0.3,2) to (0.7,1.3) to (1.1,2);

\draw[->, thick] (a) to (c);
\draw[<-, thick] (b) to (a);

\end{tikzpicture}
}
\end{center}
\caption{$S\otimes^{+}T$}
\label{shapedef3}
\end{subfigure}
\begin{subfigure}{0.23\textwidth}
\begin{center}
\resizebox{0.6\textwidth}{!}{
\begin{tikzpicture}[baseline=2ex, scale=0.8]
\node(a) at (0,-0.5) {$\otimes^{-}$};
\node(b) at (-0.7,1) {$c\OV S$};
\node(c) at (0.7,1) {$c\OV T$};

\draw[dotted] (-0.3,2) to (-0.7,1.3) to (-1.1,2);

\node(v) at (-0.7,2.3) {$v(\OV S)$};
\node(v') at (0.7,2.3) {$v(\OV T)$};
\draw[dotted] (0.3,2) to (0.7,1.3) to (1.1,2);

\draw[<-, thick] (a) to (c);
\draw[->, thick] (b) to (a);

\end{tikzpicture}
}
\end{center}
\caption{$S\otimes^{-}T$}
\label{shapedef4}
\end{subfigure}
\caption{Definition of shapes}
\end{figure}

\begin{figure}
\begin{center}
\resizebox{0.4\textwidth}{!}{
\begin{tikzpicture}

\node(rdd) at (4,3) {$+$};
\node(rdsd) at (2,3) {$-$};
\node(rdss) at (0,3) {$-$};
\node(rsd) at (-1,3) {$-$};
\node(rssd) at (-2,3) {$+$};
\node(rsss) at (-4,3) {$+$};

\node(rss) at (-3,2) {$\otimes^{+}$};
\node(rs) at (-2,1) {$\multimap^{-}$};

\node(rds) at (1,2) {$\otimes^{-}$};
\node(rd) at (3,1) {$\multimap^{+}$};
\node(r) at (0,0) {$\multimap^{+}$};

\draw[->, thick] (r) to (rd);
\draw[->, thick] (rd) to (rdd);
\draw[->, thick] (rs) to (rss);
\draw[->, thick] (rdsd) to (rds);
\draw[<-, thick] (rds) to (rdss);
\draw[->, thick] (rsd) to (rs);
\draw[->, thick] (rs) to (rss);
\draw[->, thick] (rss) to (rssd);
\draw[->, thick] (rss) to (rsss);

\draw[->, thick] (rdd) to [bend right=55] (rsd);
\draw[<-, thick] (rdsd) to [bend right=55] (rssd);
\draw[<-, thick] (rdss) to [bend right=55] (rsss);

\end{tikzpicture}}
\end{center}
\caption{$id_{S}\leftrightharpoons ((I\otimes I)\multimap I)\to ((I\otimes I)\multimap I)$}
\label{graph}
\end{figure}
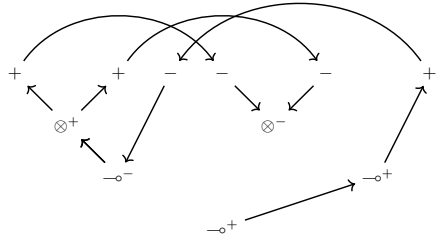

Given shapes $S,T$, by a graph $f:S\to T$ we indicate a graph $f: vS\to vT$. More generally, given a (non-empty) list $\Gamma=\{S_{1},\dots,S_{n}\}$ of shapes and a shape $T$, by a graph $f:\Gamma\to T$ we indicate a graph $f:S_{1}\otimes \dots \otimes S_{n}\to T$. Clearly, for any shape $T$ there exists a pure graph $id_{T}:T\to T$.
Given graphs $f_{1}:\Gamma_{1}\to A_{1},\dots, f_{n}:\Gamma_{n}\to A_{n}$, we let 
$f_{1}\otimes \dots \otimes f_{n}:\Gamma_{1},\dots,\Gamma_{n}\to A_{1}\otimes\dots \otimes A_{n}$ be the graph $f_{1}\cup \dots \cup f_{n}$.

 For any graph $f:S\to T$, the \emph{correction graph} of $f$, noted $f\leftrightharpoons S\multimap T$, is the rooted directed graph $f \cup S\multimap T$, with root $c(S\multimap T)$, called the \emph{conclusion} of $f\leftrightharpoons S\multimap T$.

\begin{definition}[allowable graph]
Let $S,T$ be shapes and $f:S\to T$ be a graph. $f$ is \emph{allowable} (or \emph{correct}) for $S\to T$ if $f\leftrightharpoons S\multimap T$ satisfies:
\begin{description}
\item[(acyclicity)] $f\leftrightharpoons S\multimap T$ is a connected $DAG$;

\item[(functionality)] for every $\multimap^{+}$ node of $f\leftrightharpoons S\multimap T$, every path going from the conclusion to the left premiss of the node passes through the node.

\end{description}

\end{definition}

In figure \ref{graph} the correction graph $id_{S}\leftrightharpoons S \to S$ is illustrated, where $S$ is the shape $ (I\otimes I)\multimap I$. 

%
%

The category $\CC A$ of allowable graphs has shapes $S,T$ as objects and allowable graphs $f:S\to T$ as morphisms (with composition defined as in $\CC G$). 
$\CC A$ can be presented also as a \emph{symmetric multicategory} (\cite{Leinster2004}) $m\CC A$ whose objects are shapes and whose multiarrows are graphs $f:\Gamma\to S$, where $\Gamma$ is a (possibly empty) list of shapes. Multicomposition is defined as follows: given (multi)arrows $f:B_{1},\dots, B_{n}\to C$, $g_{1}:\Delta_{1}\to B_{1},\dots ,g_{n}:\Delta_{n}\to B_{n}$, where the $\Delta_{i}$ indicate finite lists of shapes, one can define $f\circ_{m} (g_{1},\dots, g_{n}):\Delta_{1},\dots, \Delta_{n}\to C$ as
$f\circ (g_{1}\otimes \dots \otimes g_{n})$. 
Observe that, in $m\CC A$ one can consider arrows $f:\emptyset \to S$, corresponding to closed proofs. Due to the absence of the tensor unit in $\mathit{IMLL}^{-}$, such arrows do not exist in $\CC A$. 
In the following we will often confuse $\CC A$ and $m\CC A$.


We let $\CC A^{\multimap}$ (resp. $m\CC A^{\multimap}$) indicate the subcategory of $\CC A$ (resp. the sub-multicategory of $m \CC A$) whose shapes do not contain $\otimes^{+}$ and $\otimes^{-}$. 



We let $\CC L_{\multimap,\otimes}$ be the language given by the grammar $A,B:= X\mid A\multimap B\mid A\otimes B$, where $X\in \CC V$.
Any linear type $A\in \CC L_{\multimap, \otimes}$ is obviously assigned a shape $S_{A}$ and a \emph{labeling}, i.e. a map $l_{A}:vS_{A}\to \CC V$ associating the leafs of $S_{A}$ with a variable. $f:S\to T$ is a \emph{labeled graph} (or, simply, a graph when no ambiguity occurs) $f:A\to B$ if $S_{A}=S, S_{B}=T$ and, by letting $l= l_{A}\cup l_{B}$, $(x,y)\in f \ \To \ l(x)=l(y)$. A labeled graph (illustrated in figure \ref{labgraph}) can be though as a graph over signed multisets of variables. Given $X\in \CC V$, if $e=(x,y)\in f$ and $l(x)=l(y)=X$, then we say that the $e$ is \emph{over $X$}. Clearly, for any $A\in \CC L_{\multimap}$, if $S_{A}=T$, then $id_{T}$ (which we will note simply $id_{A}$) is a correct pure labeled graph from $A$ to $A$. In the following we will often confuse between a linear type and its associated shape. We will also often confuse the context $\Gamma=\{x_{1}:A_{1},\dots, x_{n}:A_{n}\}$ of linear types with the list of shapes $\Gamma=S_{A_{1}},\dots, S_{A_{n}}$.

If $f:A\to B$ is any (non necessarily correct) graph, then for any $X\in \CC V$, $f$ induces a \emph{$X$-pairing} of $A\multimap B$, i.e. a partition of all occurrences of $X$ in $A\multimap B$ in pairs whose elements have opposite polarity.



We say that two types $A,B\in \CC L_{\multimap,\otimes}$ are \emph{isomorphic} when there exist correct graph $f:A\to B$ and $g:B\to A$ such that $g\circ f=id_{A}$, $f\circ g= id_{B}$.

We conclude the presentation of allowable graphs by showing how to associate to any normal linear $\lambda$-term $u$ such that $\Gamma\vdash u:A$ is derivable in $\lambda_{\multimap}$, an allowable graph $\CC G(u):\Gamma\to A$\footnote{Observe that we are here confusing the list $\Gamma=A_{1},\dots, A_{n}$ with the context $\Gamma=x_{1}:A_{1},\dots, x_{n}:A_{n}$. We will often confuse them, if it creates no ambiguity.}. The definition of $\CC G(u)$ actually depends on $\Gamma$ and $A$, so it should be written more pedantically as $\CC G_{\Gamma}^{A}(u)$, as different typings of the same $\lambda$-term give rise to different labeled graphs\footnote{Indeed all such graphs can be obtained by suitable expansions from the graph $\CC G_{\Gamma_{0}}^{A_{0}}(u)$, where $\Gamma_{0}\vdash u:A_{0}$ is a \emph{principal typing} of $u$.}). 


For any linear type $A\in \CC L_{\multimap}$ we let $\pi A\subseteq \{\TT l, \TT r\}^{*}$ (where $\TT l$ stands for ``left'' and $\TT r$ stands for ``right'') be the set of all paths, i.e. all finite sequences of elements of $\{\TT l, \TT r\}$, leading to variables in the syntactic tree of $A$. 
Given a context $\Gamma$ and a linear type $A$, any element of $\OV{v\Gamma}  +v( A)$, where $\OV{v\Gamma}=\OV{vA_{0}}+\dots + \OV{vA_{n-1}}$ and $\Gamma=\{A_{0},\dots, A_{n-1}\}$, is  uniquely determined by a pair $(i,\pi)$ made of an index $i\in \{0,\dots,n\} $ (by letting $A_{n}=A$) and a path $\pi\in\{\TT l, \TT r\}^{*}$ such that $\pi\in \pi A_{i}$. Let $p(\Gamma,A)$ be the set of such pairs. There exists then a bijection $\alpha_{\Gamma,A}:p(\Gamma,A)\to \OV{v\Gamma}  +v( A)$ where $\alpha(i,\pi)$ is the node corresponding to the path $\pi$ in the syntactic tree of $A_{i}$. In case $\Gamma=\emptyset$, then there is a canonical bijection $\alpha_{A}:\pi A\to vA$ such that $\alpha_{A}(\pi)$ is the node corresponding to $\pi$ in the syntactic tree of $A$. The translation can then be defined inductively as follows:

\begin{enumerate}
\item if $u=x$, then $n=1$, $\Gamma=\{A_{0}\}$, $A_{0}=A_{1}=A$ and we have
$x:A\vdash x:A$. Then $\CC G(u)=\{ (\alpha_{A,A}(0,\pi), \alpha_{A,A}(1,\pi )  ) \mid \pi\in \pi A \}$;

\item if $u=\lambda x^{B}.u'$, then we have
$\Gamma, x:B\vdash u':C$, so by induction hypothesis, the graph $\CC G(u')$ is defined.
Observe that $\alpha_{\Gamma,B\multimap C}:p(\Gamma, B\multimap C)\to \OV{v\Gamma}+ \OV{vB}+ vC$ is defined by
$\alpha_{\Gamma, B\multimap C}(i,\pi)=\alpha_{\Gamma\cup \{B\}, C}(i,\pi)$ when $i< n-1$ and $\pi\in \pi A_{i}$, $\alpha_{\Gamma,B\multimap C}(n,\TT l\cdot\pi)= \alpha_{\Gamma\cup\{B\},C}(n,\pi)$, when $\pi\in \pi B$ and $\alpha_{\Gamma,B\multimap C}(n,\TT r\cdot \pi)=\alpha_{\Gamma\cup\{B\},C}(n+1, \pi)$, when $\pi\in \pi B$. 
We put then $\CC G(u):=\alpha_{\Gamma,B\multimap C}(\alpha_{\Gamma\cup\{B\}, C}^{-1}(\CC G(u')))$.

\item if $u=yu_{1}\dots u_{p}$, where 
$y:C=B_{1}\multimap \dots \multimap B_{p}\multimap A$, and $\Gamma_{i}\vdash u_{i}:B_{i}$, where the $\Gamma_{i}$ form a partition of $\Gamma-\{C\}$, then by induction hypotheses the graphs $\CC G(u_{i})$ are defined and we have bijections $\alpha_{\Gamma_{i},B_{i}}:p(\Gamma_{i},B_{i})\to \OV{v\Gamma_{i}}+ vB_{i}$ and $\alpha_{C}: \pi C\to vC, \alpha_{A}:\pi A\to vA$.
 $\alpha_{\Gamma,A}:p(\Gamma,A)\to \OV{v\Gamma}+vA$ is defined by
$\alpha_{\Gamma,A}( \sum_{j<i}k_{j}+l, \pi)= \alpha_{\Gamma_{i},B_{i}}(l, \pi)$, where $i<n-1$ $\pi\in \pi A_{i}$, $k_{j}$ is the cardinality of $\Gamma_{j}$ and $l\leq k_{i}$, $\alpha_{\Gamma,A}(n-1, \pi)=\alpha_{C}(\pi)$, for $\pi\in \pi C$ and $\alpha_{\Gamma,A}(n,\pi)=\alpha_{A}(\pi)$, for $\pi\in \pi A$. 
Then $\CC G(u)= \alpha_{\Gamma,A}(\alpha_{\Gamma_{1},A_{1}}^{-1}( \CC G(u_{1}))) \cup \dots \cup \alpha_{\Gamma,A}(\alpha_{\Gamma_{p},B_{p}}^{-1}(\CC G(u_{p}))) \cup \{(\alpha_{C}(\TT l^{p}\cdot\pi),\alpha_{A}(\pi)) \mid \pi\in \pi A\}$, where $\TT r^{p}\cdot \pi$ indicates $\underbrace{\TT r \cdot \dots \cdot \TT r }_{p \text{ times}}\cdot \pi$.



\end{enumerate}

We recall some standard results relating the $\lambda$-term $u$ and its graph $\CC G(u)$:

\begin{theorem}[adequacy, sequentialization and normalization, \cite{Lamarche, Ong}]\label{seque}
Let $u,v$ be normal $\lambda$-terms.

\begin{itemize}
\item If $\Gamma\vdash u:A$ is derivable in $\lambda_{\multimap}$, then $\CC G(u):\Gamma\to A$ is allowable;
\item if $f:\Gamma\to A$ is allowable, then for some normal $u$ such that $\Gamma\vdash u:A$, $f=\CC G(u)$;

\item if $\Gamma\vdash u:A$ and $x:A\vdash v:B$, then $\CC G(u)\circ \CC G(v)= \CC G(w)$, where $w$ is the $\beta$-normal form of $u[v/x]$.

\end{itemize}
\end{theorem}



\subsection{Categories of typed $\lambda$-terms and their functorial interpretation}\label{sec23}

Since Lambek's pioneering work (\cite{Lambek1969}), it has become standard to treat deductive systems as syntactic (multi)categories whose objects are formulas or types and whose arrows are equivalence classes of derivations.
The categorial treatment of deductive systems allows to introduce their functorial interpretation in an internal way, i.e. by considering types as endofunctors over the syntactic categories of types and terms, and well-typed terms as dinatural transformations between such functors.

In the previous subsection we described the category $\CC A$ generated by proof nets and we observed that, by suitably extending proof nets to unities, the syntactic category obtained corresponds either to the free $^{*}$-autonomous category (\cite{Lamarche2004, Hughes2012}) or to the free symmetric monoidal closed category (\cite{Blute1996}). 
We now briefly recall the categorial description of the simply typed $\lambda$-calculus and System $F$. 

We let $\CC T$ be the category generated by $\lambda_{\To}$: the types of $\CC T$ are the simple types and the arrows $u:A\to B$ are $\lambda$-terms $u$ with exactly one free variable, such that $x:A\vdash u:B$ is derivable in $\lambda_{\To}$, considered 
up to $\beta\eta$-equivalence and up to renaming of its unique free variable. We will often note an arrow $u[z]:A\to B$, where $z$ indicates its unique free variable. When no variable is indicated, we let $x$ denote by convention the unique free variable of an arrow $u:A\to B$. The composition of arrows $u:A\to B$ and $v:B\to C$ is the arrow $v[u/x]:A\to C$, where $x$ is the unique free variable in $v$.
The symmetric multicategory $m\CC T$ can be defined in a similar way. As in the case of $\CC A$, when indicating a  a multiarrow $u:\Gamma\to A$, we will confuse the list $A_{1},\dots, A_{n}$ with
the context $\Gamma=x_{1}:A_{1},\dots, x_{n}:A_{n}$. As already mentioned, we will often confuse a category and its associated multicategory.

As is well-known, if $\lambda_{\To}$ is extended with finite products constructions, then the syntactic category obtained is the free cartesian closed category (\cite{LambekScott}). 
Finally, we consider the syntactic categories $\CC F$ and $\CC F_{at}$ (and the respective multi-categories), defined similarly to $\CC T$: $\CC F$ (resp. $\CC F_{at}$) is the category having as objects the types in $\CC L_{\To, \forall}$ and such that an arrow $u:A\to B$ is a $\lambda$-term (up to $\beta\eta$-equivalence and renaming of its unique free variable) with exactly one free variable such that $x:A\vdash u:B$ is derivable in $F$ (resp. in $F_{at}$).

Functorial polymorphism (\cite{Bainbridge1990}) is the interpretation of types as multivariant functors over a category (either symmetric monoidal closed or $^{*}$-autonomous in the linear case - see \cite{Blute1993} - and cartesian closed in the non linear case -\cite{Girard1992}) and of typed terms as dinatural transformations between such types. We recall that, given a category $\CC C$ and multivariant functors $F,G: \CC C^{op}\times \CC C \to \CC C$, a \emph{dinatural transformation} between $F$ and $G$ is a family of arrows $\theta_{A}$ indexed by the objects of $\CC C$ such that, given objects $A,B$ in $\CC C$ and an arrow $f: A\to  B$, the following diagram commutes:

\begin{center}
\resizebox{0.4\textwidth}{!}{
$$
\xymatrix{
 &  F A A \ar[r]^{\theta_{A}} & G AA \ar[rd]^{G A f}&  \\
F B A \ar[rd]_{F B f} \ar[ru]^{F f A}  &   &  &   G A B\\
  & F B B \ar[r]_{\theta_{B}} & G B B  \ar[ru]_{G f B}&
}
$$}
\end{center}

We can now define functorial interpretation by considering the categories $\CC A$, $\CC T$ and $\CC F$.
By considering the category $\CC A$ (resp. $\CC T$, $\CC F$), any linear type (resp. simple type) $A$, depending on a variable $X$, can be interpreted as a multivariant endofunctor $A:\CC A^{op}\otimes \CC A\to \CC A$ (resp. $A:\CC T^{op}\times \CC T\to \CC T$, $A:\CC F^{op}\otimes \CC F\to \CC F$), which is covariant when $A$ is p-$X$ and contravariant when $A$ is n-$X$. 

Given types $A,B,C$ in any of the above categories, we indicate by $A[B,C]$ the type obtained by substituting $B$ for the negative occurrences of $X$ in $A$ and $C$ for the positive occurrences of $X$ in $A$. 
Given types $B,B',C,C'$ and arrows $u:B\to B'$, $v:C\to C'$ we can define the arrow $A(u,v):A[B',C]\to A[B,C']$ in $\lambda$-calculus notation by induction on $A$ as follows:
\begin{equation*}
\begin{split}
Y(u, v) \ & = \ \begin{cases} v & \text{ if } X=Y \\ Y & \text{ otherwise}\end{cases} \\
(A_{1}\To A_{2})(u,  v) \ & = \ A_{1}(v, u)\To A_{2}( u,  v) \\
\forall YA( u,  v) \ & = \ \Lambda Y.' A[Y'/Y]( u,  v)(xY')
\end{split}
\end{equation*}
where, given arrows $u[z]:B\to B'$ and $v[z']:C\to C'$, $(u\To v)[x]:(B'\To C  ) \to (B\To  C')$ is the arrow $\lambda y^{B}. v[x(u[y/z])/z']$.


It is a standard result (\cite{Blute1993, Girard1992}) that any arrow $f:A\to B$ in $\CC A$ (resp. $u:A\to B$ in $\CC T$) corresponds to a dinatural transformation between the functors associated to $A$ and $B$, respectively. 
This means that any arrow $u:A\to B$ in $\CC T$ yields a dinatural transformation between $A$, and $B$, by letting $u_{C}:A(C , C)\to B(C, C)$ be $u[C/X]$, that is, given any arrow $v: C\to C'$, the equation below 
\begin{equation}\label{dina}
B(C,  v) \circ u_{C} \circ A(v, C) \ = \ B( v, C' )\circ  u_{C'} \circ A(C',  v) 
\end{equation}
where $u\circ v=\lambda x.u(vx)$, holds modulo $\beta\eta$-equivalence.

This fact does not extend to $\CC F$ (\cite{Delatail2009}): there exist well-typed System $F$ terms which are not dinatural. 
This means that dinaturality generates an equational theory over System $F$ terms which strictly extends $\beta\eta$-equivalence (see \cite{StudiaLogica} for a proof theoretic discussion).
We let $ T_{\varepsilon} $ indicate the $\lambda^{2}$-theory generated by all equations \ref{dina} between arrows in $\CC F$. The equivalence $\simeq_{\varepsilon}$ induced by $T_{\varepsilon}$ captures then all equations between System $F$ terms which hold when interpreting such terms in an arbitrary dinatural model of System $F$.

The functorial interpretation can be extended to proof nets and $\lambda$-terms with undischarged assumptions or free variables, respectively, by considering the \emph{polynomial categories} $\CC A[x_{1},\dots, x_{n}]$, $\CC T[x_{1},\dots, x_{n}], \CC F[x_{1},\dots, x_{n}]$ (see \cite{LambekScott}) obtained by adding to $\CC A$, $\CC T$ and $\CC F$ new arrows $x_{i}: \emptyset\to A_{i}$, for some types $A_{1},\dots, A_{n}$. This extension allows to consider the functorial action of types over terms with undischarged assumptions: for instance, let $x_{1}:\emptyset\to A_{1},\dots, x_{n}:\emptyset\to A_{n}$ be ``variable arrows'' and $u: A_{1},\dots, A_{n}, C\to C', b:A_{1},\dots, A_{n}:D\to D'$ in $\CC T$; then the functorial action of the type $A=X\To X$ in $\CC T[x_{1},\dots, x_{n}]$ corresponds to an arrow
$A(u,v): A_{1},\dots, A_{n}, A[C',D]\to A[C,D']$ in $\CC T$.
When considering polynomial categories, the theory $T_{\varepsilon}$ can be defined in a more general and uniform way as the $\lambda^{2}$-theory generated by all equations 
\begin{equation}\label{dinas}
B(C,  yx) \circ u_{C} \circ A(yx, C) \ = \ B( yx, C' )\circ  u_{C'} \circ A(C',  yx) 
\end{equation}
where $u:A\to B$ in $\CC F$ and $yx :C\to C'$ in the polynomial category $\CC F[y]$ generated by the variable arrow $y:\emptyset\to C\To C'$. Observe that all instances of \ref{dina} can be deduced from \ref{dinas} by using standard identity axioms (which hold in any $\lambda^{2}$-theory).

 As it will be clear in section \ref{sec6}, the functorial action of types over open terms is a basic tool in the functorial formulation of instantiation overflow for Russell-Prawitz types.


\section{Proof net interpolation and the positivity lemma}\label{sec4}

%
 
In this section we discuss interpolation from the viewpoint of proof nets and we use it to prove the positivity lemma, which shows that, whenever a type $A$ has a covariant (resp. contravariant) action over $\CC A$, then $A$ is equivalent to a type containing only positive (resp. negative) occurrences of $X$. This result is in some sense the converse of the remark that a type containing only positive (resp. negative) occurrences of a variable is a covariant (resp. contravariant) endofunctor over $\CC A$ (see section \ref{sec6}).


First, we recall proof net interpolation for $\mathit{IMLL^{-}}$ and we define the interpolation order $\prec_{I}$ over linear types; then, we consider weak proof net interpolation for the fragment $\lambda_{\multimap}$ (proved in appendix \ref{appA} by adapting the argument in \cite{deGroote1996} for $\mathit{IMLL^{-}}$ interpolation). Weak interpolation will be exploited to. Finally, we deduce the positivity lemma from proof net interpolation.



\subsection{Proof net interpolation in $\mathit{IMLL^-}$}\label{3}

Craig Interpolation for sequent calculus is usually formulated as follows (see \cite{Troelstra}): given a cut-free natural deduction derivation of $\Gamma, \Delta\vdash A$, there exists a type $I$, called the \emph{interpolant} of the derivation, such that the variables of $I$ occur in both $\Gamma$ and $\Delta,A$ and there exist two derivations of conclusions respectively $\Gamma\vdash I$ and $\Delta,I\vdash A$. 

Sequent calculus interpolation for linear logic was first investigated in \cite{Roorda}, when it was realized that the statement above can be strengthened by considering proof nets. 
Indeed, by taking $\Delta=\emptyset$, interpolation yields a procedure to ``split'' the type II and type III parts of a proof net, yielding two graphs $f_{1}:\Gamma\to I$, $f_{2}:I\to A$, where $f_{1}$ contains the type $II$ part of $f$, $f_{2}$ contains the type $III$ part of $f$ and both $f_{1}$ and $f_{2}$ contain the type $I$ part of $f$.

This idea appears in two different approaches to proof net interpolation, the one in \cite{deGroote1996} (inspired from \cite{Roorda}) for $\mathit{IMLL^-}$, which we recall here, and the one in \cite{Carbone1997}, based on flow graphs for the classical sequent calculus $LK$.  

%

Let $A,B\in \CC L_{\multimap}$ and $\Gamma=\{B_{1},\dots, B_{n}\}$ be a finite multiset of types. We say that $A$ \emph{injects into $B$} (resp. \emph{$A$ injects into $\Gamma$}), noted $A\hookrightarrow B$ (resp. $A\hookrightarrow \Gamma$) if there exists an injective function $h:vA\to vB$ ($h:vA\to \sum_{i}vB_{i}$) preserving polarities and labels.

The proof net interpolation problem for $\mathit{IMLL^-}$ can be described as follows
\begin{definition}[$\mathit{IMLL^-}$ interpolation problem]\label{interpo1}
Given $f: \Gamma\to A$ in $\CC A$, where $f_{I},f_{II},f_{III}\subseteq f$ denote its type I, type II and type III parts, respectively, find a linear type $I$, called the \emph{interpolant of $f$}, such that $f_{II}\cup f_{I}: \Gamma\to I$ and $f_{I}\cup f_{III}: I\to A$ in $\CC A$.
%
%
\end{definition}

Observe that, following definition \ref{interpo1}, interpolation forces $I\hookrightarrow A$ and $I\hookrightarrow \Gamma$, as there exist type $I$ edges connecting any variable occurrence of $I$ with variable occurrences in both $\Gamma$ and $A$. 

The interpolation problem can be reformulated by considering graphs with cuts (as in \cite{deGroote1996}). 
Let us add to the class of shapes the shapes $cut^{+}$ and $cut^{-}$ in figure \ref{shapedef11} and \ref{shapedef21}, respectively. 
By $cut$ we will generically indicate either $cut^{+}$ or $cut^{-}$.

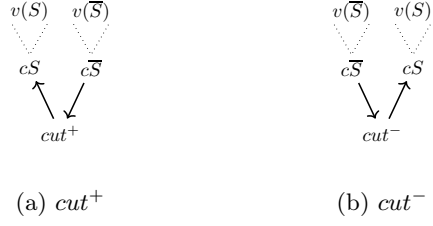
\begin{figure}
\begin{center}
\begin{subfigure}{0.28\textwidth}
\begin{center}
\resizebox{0.4\textwidth}{!}{
\begin{tikzpicture}[baseline=2ex, scale=0.8]
\node(a) at (0,-0.5) {$cut^{+}$};
\node(b) at (-0.7,1) {$cS$};
\node(c) at (0.7,1) {$c\OV S$};

\draw[dotted] (-0.3,2) to (-0.7,1.3) to (-1.1,2);

\node(v) at (-0.7,2.3) {$v(S)$};
\node(v') at (0.7,2.3) {$v(\OV S)$};
\draw[dotted] (0.3,2) to (0.7,1.3) to (1.1,2);

\draw[<-, thick] (a) to (c);
\draw[<-, thick] (b) to (a);

\end{tikzpicture}
}
\end{center}
\caption{$cut^{+}$}
\label{shapedef11}
\end{subfigure}
\begin{subfigure}{0.28\textwidth}
\begin{center}
\resizebox{0.4\textwidth}{!}{
\begin{tikzpicture}[baseline=2ex, scale=0.8]
\node(a) at (0,-0.5) {$cut^{-}$};
\node(b) at (-0.7,1) {$c\OV S$};
\node(c) at (0.7,1) {$cS$};

\draw[dotted] (-0.3,2) to (-0.7,1.3) to (-1.1,2);

\node(v) at (-0.7,2.3) {$v(\OV S)$};
\node(v') at (0.7,2.3) {$v( S)$};
\draw[dotted] (0.3,2) to (0.7,1.3) to (1.1,2);

\draw[->, thick] (a) to (c);
\draw[->, thick] (b) to (a);

\end{tikzpicture}
}
\end{center}
\caption{$cut^{-}$}
\label{shapedef21}
\end{subfigure}
\end{center}
\caption{$cut$ links}
\label{cuts}
\end{figure}

By a \emph{graph with cuts} we indicate a correct graph $f:\Gamma, cut,\dots, cut \to A$, where $\Gamma$ and $A$ have no occurrence of $cut$.
By a \emph{graph with $n$ splitting cuts} we indicate a graph with cuts $f:\Gamma, \underbrace{cut,\dots, cut}_{n\text { times}}\to A$ such that no type I edge of $f$ connects $\Gamma$ and $A$ and no type II edge of $f$ connects any cut. 
A graph with splitting cuts can be 
described as a graph $f:\Gamma_{1},\dots,\Gamma_{p}\to A$, where $\Gamma_{i}=\Delta_{i},cut,\dots, cut$, where the $\Delta_{i}$ form a partition of $\Gamma$, and the correction graph is as in figure \ref{fig5}.

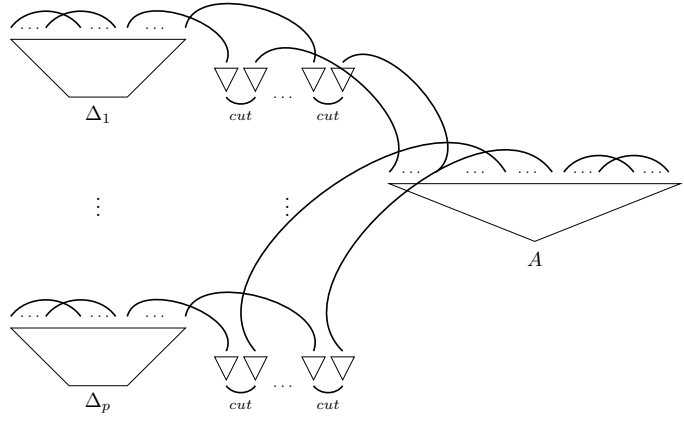
\begin{figure}[h]\begin{center}\resizebox{0.6\textwidth}{!}{\begin{tikzpicture}

\draw[] (0,0) to (1,0) to (2,1) to (-1,1) -- cycle;
\draw[] (8,-2.5) to (10.5,-1.5) to (5.5,-1.5) -- cycle;

\node(g) at (0.5,-0.3) {$\Delta_{1}$};
\node(a) at (8,-2.8) {$A$};

\draw[thick] (-1,1.2) to [bend left=55] (0.2,1.2);
\draw[thick] (-0.4,1.2) to [bend left=55] (0.8,1.2);
\node(d1) at (-0.65,1.2) {\scriptsize$\dots$};
\node(d2) at (0.55,1.2) {\scriptsize$\dots$};

\draw[thick] (8.5,-1.3) to [bend left=55] (9.7,-1.3);
\draw[thick] (9.1,-1.3) to [bend left=55] (10.3,-1.3);
\node(d1) at (8.85,-1.3) {\scriptsize$\dots$};
\node(d2) at (10.05,-1.3) {\scriptsize$\dots$};

\draw (2.5,0.5) to (2.9,0.5) to (2.7,0.1) -- cycle;
\draw (3,0.5) to (3.4,0.5) to (3.2,0.1) -- cycle;
\draw[thick] (2.7,0) to [bend right=55] node[below] {\scriptsize$cut$} (3.2,0);

\node(dd) at (3.7,0) {\scriptsize$\dots$};

\draw (4,0.5) to (4.4,0.5) to (4.2,0.1) -- cycle;
\draw (4.5,0.5) to (4.9,0.5) to (4.7,0.1) -- cycle;
\draw[thick] (4.2,0) to [bend right=55] node[below] {\scriptsize$cut$} (4.7,0);

\draw[thick] (1,1.2) to [bend left=95] (2.7,0.6);
\draw[thick] (2,1.2) to [bend left=95] (4.2,0.6);
\node(ddd) at (1.5,1.2) {\scriptsize$\dots$};

\draw[thick] (3.2,0.6) to [bend left=95] (5.5,-1.3);
\draw[thick] (4.7,0.6) to [bend left=95] (6.3,-1.3);
\node(ddd) at (5.9,-1.3) {\scriptsize$\dots$};

\draw[] (0,-5) to (1,-5) to (2,-4) to (-1,-4) -- cycle;

\node(g1) at (0.5,-5.3) {$\Delta_{p}$};

\draw[thick] (-1,-3.8) to [bend left=55] (0.2,-3.8);
\draw[thick] (-0.4,-3.8) to [bend left=55] (0.8,-3.8);
\node(d1) at (-0.65,-3.8) {\scriptsize$\dots$};
\node(d2) at (0.55,-3.8) {\scriptsize$\dots$};

\draw (2.5,-4.5) to (2.9,-4.5) to (2.7,-4.9) -- cycle;
\draw (3,-4.5) to (3.4,-4.5) to (3.2,-4.9) -- cycle;
\draw[thick] (2.7,-5) to [bend right=55] node[below] {\scriptsize$cut$} (3.2,-5);

\draw (4,-4.5) to (4.4,-4.5) to (4.2,-4.9) -- cycle;
\draw (4.5,-4.5) to (4.9,-4.5) to (4.7,-4.9) -- cycle;
\draw[thick] (4.2,-5) to [bend right=55] node[below] {\scriptsize$cut$} (4.7,-5);

\draw[thick] (1,-3.8) to [bend left=95] (2.7,-4.4);
\draw[thick] (2,-3.8) to [bend left=95] (4.2,-4.4);
\node(ddd) at (1.5,-3.8) {\scriptsize$\dots$};

\draw[thick] (3.2,-4.4) to [bend left=95] (7.5,-1.3);
\draw[thick] (4.7,-4.4) to [bend left=95] (8.3,-1.3);
\node(ddd) at (7.9,-1.3) {\scriptsize$\dots$};

\node(dd) at (3.7,-5) {\scriptsize$\dots$};

\node(dd) at (7,-1.3) {\scriptsize$\dots$};
\node(vd) at (0.5,-1.8) {$\vdots$};
\node(vd) at (3.75,-1.8) {$\vdots$};

\end{tikzpicture}}
\end{center}
\caption{Graph with splitting cuts}
\label{fig5}
\end{figure}


The interpolation problem can then be formulated as follows: given a graph $f:\Gamma\to A$ in $\CC A$, find a type $I\in \CC L_{\multimap,\otimes}$ such that the correction graph of $f$ (illustrated in figure \ref{fig1}) can be split into a graph with one splitting cut as in figure \ref{fig2}, preserving correctness. A simple example of interpolation is illustrated in figure \ref{isom2}.

\begin{figure}
\begin{subfigure}{0.48\textwidth}
\begin{center}
\resizebox{0.9\textwidth}{!}{
\begin{tikzpicture}

\node(t2) at (-0.1,2.3) {Type II};
\node(t1) at (2.7,2.3) {Type I};
\node(t3) at (5.4,2.3) {Type III};

\draw[] (0,0) to (1,0) to (2,1) to (-1,1) -- cycle;
\draw[] (5,0) to (6.5,1) to (3.5,1) -- cycle;

\node(g) at (0.5,-0.3) {$\Gamma$};
\node(a) at (5,-0.3) {$A$};

\draw[thick] (-1,1.2) to [bend left=55] (0.2,1.2);
\draw[thick] (-0.4,1.2) to [bend left=55] (0.8,1.2);
\node(d1) at (-0.65,1.2) {\scriptsize$\dots$};
\node(d2) at (0.55,1.2) {\scriptsize$\dots$};

\draw[thick] (4.5,1.2) to [bend left=55] (5.7,1.2);
\draw[thick] (5.1,1.2) to [bend left=55] (6.3,1.2);
\node(d1) at (4.85,1.2) {\scriptsize$\dots$};
\node(d2) at (6.05,1.2) {\scriptsize$\dots$};

\draw[thick] (1.2,1.2) to [bend left=55] (3.7,1.2);
\draw[thick] (1.8,1.2) to [bend left=55] (4.3,1.2);
\node(d1) at (1.55,1.2) {\scriptsize$\dots$};
\node(d2) at (4.05,1.2) {\scriptsize$\dots$};

\end{tikzpicture}}
\end{center}
\caption{Before interpolation}
\label{fig1}
\end{subfigure}
\begin{subfigure}{0.5\textwidth}
\begin{center}\resizebox{0.9\textwidth}{!}{
\begin{tikzpicture}

\node(t2) at (-0.1,2.3) {Type II};
\node(t3) at (7.4,2.3) {Type III};

\draw[] (0,0) to (1,0) to (2,1) to (-1,1) -- cycle;
\draw[] (7,0) to (8.5,1) to (5.5,1) -- cycle;

\node(g) at (0.5,-0.3) {$\Gamma$};
\node(a) at (7,-0.3) {$A$};

\draw[thick] (-1,1.2) to [bend left=55] (0.2,1.2);
\draw[thick] (-0.4,1.2) to [bend left=55] (0.8,1.2);
\node(d1) at (-0.65,1.2) {\scriptsize$\dots$};
\node(d2) at (0.55,1.2) {\scriptsize$\dots$};

\draw[thick] (6.5,1.2) to [bend left=55] (7.7,1.2);
\draw[thick] (7.1,1.2) to [bend left=55] (8.3,1.2);
\node(d1) at (6.85,1.2) {\scriptsize$\dots$};
\node(d2) at (8.05,1.2) {\scriptsize$\dots$};

\draw[] (2.5,1) to (3.5,1) to (3,0) -- cycle;
\draw[] (4,1) to (5,1) to (4.5,0) -- cycle;
\node(i1) at (3,-0.3) {$I^{+}$};
\node(i2) at (4.5,-0.3) {$I^{-}$};

\draw[thick] (i1) to [bend right=55] node[below] {$cut$} (i2);

\draw[thick] (1.2,1.2) to [bend left=55] (2.5,1.2);
\draw[thick] (1.8,1.2) to [bend left=55] (3.1,1.2);
\node(d1) at (1.55,1.2) {\scriptsize$\dots$};
\node(d2) at (2.85,1.2) {\scriptsize$\dots$};

\draw[thick] (4.2,1.2) to [bend left=55] (5.7,1.2);
\draw[thick] (4.8,1.2) to [bend left=55] (6.3,1.2);
\node(d1) at (4.55,1.2) {\scriptsize$\dots$};
\node(d2) at (6.05,1.2) {\scriptsize$\dots$};

\end{tikzpicture}}
\end{center}
\caption{After interpolation}
\label{fig2}
\end{subfigure}
\caption{Interpolation and the correction graph}
\end{figure}
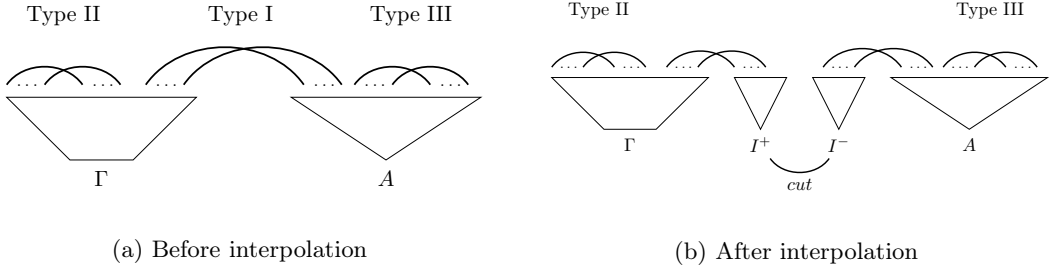

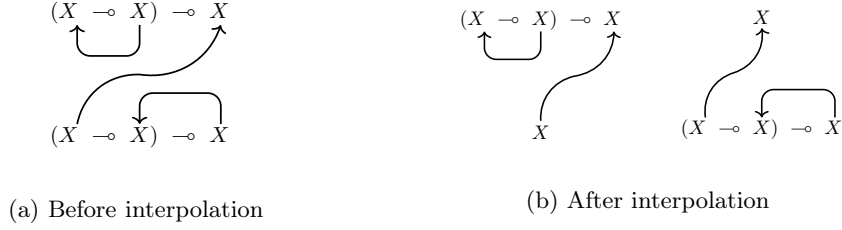
\begin{figure}
\begin{subfigure}{0.45\textwidth}
\begin{center}
\resizebox{0.4\textwidth}{!}{
\begin{tikzpicture}
\node(b) at (0,0) {$(X \ \multimap \ X) \ \multimap \ X$};

\node(c) at (0,-2) {$(X \ \multimap \ X) \ \multimap \ X$};

\draw[<-, thick, rounded corners=6pt] (-1,-0.2) to (-1,-0.7) to (0,-0.7) to (0,-0.2);
\draw[<-, thick, rounded corners=6pt] (0,-1.8) to (0,-1.3) to (1.3,-1.3) to (1.3,-1.8);
\draw[<-, thick] (1.3,-0.2) to [bend left=40]  (0,-1) to [bend right=40] (-1,-1.8);

\end{tikzpicture}}
\end{center}
\caption{Before interpolation}
\label{isom}
\end{subfigure}
\begin{subfigure}{0.45\textwidth}
\begin{center}
\resizebox{0.8\textwidth}{!}{
\begin{tikzpicture}
\node(b) at (0,0) {$(X \ \multimap \ X) \ \multimap \ X$};

\node(c) at (0,-2) {$X$};

\draw[<-, thick, rounded corners=6pt] (-1,-0.2) to (-1,-0.7) to (0,-0.7) to (0,-0.2);
\draw[<-, thick] (1.3,-0.2) to [bend left=40]  (0.6,-1) to [bend right=40] (0,-1.8);

\end{tikzpicture}
\ \ \ \ \ \ 
\begin{tikzpicture}
\node(b) at (0,0) {$X$};

\node(c) at (0,-2) {$(X \ \multimap \ X) \ \multimap \ X$};

\draw[<-, thick, rounded corners=6pt] (0,-1.8) to (0,-1.3) to (1.3,-1.3) to (1.3,-1.8);
\draw[<-, thick] (0,-0.2) to [bend left=40]  (-0.5,-1) to [bend right=40] (-1,-1.8);

\end{tikzpicture}}
\end{center}
\caption{After interpolation}
\label{isom2}
\end{subfigure}
\caption{Example of interpolation}
\end{figure}

\begin{theorem}[\cite{deGroote1996}, \cite{Carbone1997}]\label{interpolation}
Any $\mathit{IMLL^-}$ interpolation problem has a solution.
\end{theorem}

Interpolation allows to introduce an order over types, the \emph{interpolation order} $\prec_{I}$, where $A\prec_{I} B$ holds if $A$ is the interpolant of an allowable graph $f:B\to B$. Observe that, if $A\prec_{I} B$, then $A$ and $B$ are interderivable and $A\hookrightarrow B$. When $A\prec_{I}B$, we will say that $B$ \emph{linearly collapses} into $A$. If for no $A$, $A\prec_{I}B$, then $B$ will be called \emph{$\prec_{I}$-minimal}.

For example, the type $(X\multimap X)\multimap (X\multimap X)$ is not $\prec_{I}$-minimal, as it linearly collapses into the $\prec_{I}$-minimal type $X\multimap X$. 

Observe that, if $A$ linearly collapses into $B$ and $B$ linearly collapses into $C$, then $A$ linearly collapses into $C$.

%
%
%
%
%
%
%

Let $\ell(A)$ indicate the number of variable occurrences in $A$ (equivalently, the length of $vA$). If $A\prec_{I} B$, then $\ell(A)< \ell (B)$. Types with different lengths cannot be isomorphic:

\begin{proposition}\label{iso}
If $\ell(A)< \ell(B)$, then $A$ is not isomorphic to $B$.
\end{proposition}
\begin{proof}
Suppose $f:A\to B$ and $g:B\to A$, then $f\circ g:B\to B$ contains both type II and type III edges, hence it cannot be the identity.

\end{proof}

We deduce that, if $A$ linearly collapses into $B$, then $A$ cannot be isomorphic to $B$.

The following is an immediate consequence of interpolation:

\begin{lemma}[linear collapse lemma for $\mathit{IMLL}^{-}$]\label{colla}
For any $A\in \CC L_{\multimap,\otimes}$, if $f:A\to A$ in $\CC A$ contains at least one type III (or type II) edge, then $A$ is not $\prec_{I}$-minimal.
\end{lemma}
%
%
%

The linear collapse lemma states that if a type $A$ is $\prec_{I}$-minimal, then any correct graph $f:A\to A$ must be pure (hence a permutation of the identity graph). As an example we mention the graph $f:D\to D$, where $D=(X\multimap X)\multimap X$ shown in figure \ref{isom}: as the graph $f$ is not pure, the type $D$ is not minimal with respect to linear collapse. Indeed $D$ linearly collapses onto $X$ and the two graphs $g:D\to X, h:X\to D$, shown in figure \ref{isom2}, are obtained by decomposing $f$.

\subsection{Weak proof net interpolation in $\lambda_{\multimap}$}\label{sec32}

We now consider proof net interpolation for $\lambda_{\multimap}$. Interpolation \ref{interpo1} fails for this fragment: the interpolant of the graph $f: Y,Z\to (Y\multimap Z\multimap X)\multimap X$ is the type $Y\otimes Z\notin \CC L_{\multimap}$ and there is no interpolant for $f$ in $\CC L_{\multimap}$. 

It is well-known that interpolation similarly fails in the $\To$-fragment of propositional intuitionistic logic (see \cite{Kanazawa}). One usually considers a weak form of interpolation, which consists in admitting multiple interpolants. In the case of $\lambda_{\multimap}$, we can formulate weak interpolation as follows:

\begin{definition}[weak $\lambda_{\multimap}$ interpolation problem]\label{interpo2}

Given $f: \Gamma\to A$ in $\CC A$, where $f_{I},f_{II},f_{III}\subseteq f$ denote its type I, type II and type III parts, respectively, find a partition $\Gamma_{1},\dots, \Gamma_{p}$ of $\Gamma$, a partition $f_{I}^{1},\dots, f_{I}^{p}$ of $f_{I}$, a partition $f_{II}^{1},\dots, f_{II}^{p}$ of $f_{II}$ and types $I_{1},\dots, I_{p}$ such that $f_{II}^{i}\cup f_{I}^{i}: \Gamma_{i}\to I_{i}$ and $f_{I}\cup f_{III}: I_{1},\dots, I_{p}\to A$ in $\CC A$.
%
\end{definition}

Weak interpolation forces $I_{i}\hookrightarrow A$ and $I_{i}\hookrightarrow \Gamma_{i}$ for all $1\leq i\leq p$. 

 The weak interpolation problem can be formulated by means of graphs with cuts as follows: given a graph $f:\Gamma\to A$ in $\CC A$ find a partition $\Gamma_{1},\dots, \Gamma_{p}$ of $\Gamma$ and $p$ types $I_{1},\dots, I_{p}$ such that the correction graph (figure \ref{fig1}) can be split into a graph with $p$ splitting cuts as in figure \ref{fig3}, preserving correctness.

\begin{figure}
\begin{center}
\resizebox{0.6\textwidth}{!}{
\begin{tikzpicture}

\draw[] (0,0) to (1,0) to (2,1) to (-1,1) -- cycle;
\draw[] (8,-2.5) to (10.5,-1.5) to (5.5,-1.5) -- cycle;

\node(g) at (0.5,-0.3) {$\Gamma_{1}$};
\node(a) at (8,-2.8) {$A$};

\draw[thick] (-1,1.2) to [bend left=55] (0.2,1.2);
\draw[thick] (-0.4,1.2) to [bend left=55] (0.8,1.2);
\node(d1) at (-0.65,1.2) {\scriptsize$\dots$};
\node(d2) at (0.55,1.2) {\scriptsize$\dots$};

\draw[thick] (8.5,-1.3) to [bend left=55] (9.7,-1.3);
\draw[thick] (9.1,-1.3) to [bend left=55] (10.3,-1.3);
\node(d1) at (8.85,-1.3) {\scriptsize$\dots$};
\node(d2) at (10.05,-1.3) {\scriptsize$\dots$};

\draw[] (2.5,1) to (3.5,1) to (3,0) -- cycle;
\draw[] (4,1) to (5,1) to (4.5,0) -- cycle;
\node(i1) at (3,-0.3) {$I^{+}_{1}$};
\node(i2) at (4.5,-0.3) {$I^{-}_{1}$};

\draw[thick] (i1) to [bend right=55] node[below] {$cut$} (i2);

\draw[thick] (1.2,1.2) to [bend left=55] (2.5,1.2);
\draw[thick] (1.8,1.2) to [bend left=55] (3.1,1.2);
\node(d1) at (1.55,1.2) {\scriptsize$\dots$};
\node(d2) at (2.85,1.2) {\scriptsize$\dots$};

\draw[thick] (4.2,1.2) to [bend left=85] (7.7,-1.3);
\draw[thick] (4.8,1.2) to [bend left=85] (8.3,-1.3);
\node(d1) at (4.55,1.2) {\scriptsize$\dots$};
\node(d2) at (8.05,-1.3) {\scriptsize$\dots$};

\draw[] (0,-5) to (1,-5) to (2,-4) to (-1,-4) -- cycle;

\node(g1) at (0.5,-5.3) {$\Gamma_{p}$};

\draw[thick] (-1,-3.8) to [bend left=55] (0.2,-3.8);
\draw[thick] (-0.4,-3.8) to [bend left=55] (0.8,-3.8);
\node(d1) at (-0.65,-3.8) {\scriptsize$\dots$};
\node(d2) at (0.55,-3.8) {\scriptsize$\dots$};

\draw[] (2.5,-4) to (3.5,-4) to (3,-5) -- cycle;
\draw[] (4,-4) to (5,-4) to (4.5,-5) -- cycle;
\node(i1) at (3,-5.3) {$I^{+}_{p}$};
\node(i2) at (4.5,-5.3) {$I^{-}_{p}$};

\draw[thick] (i1) to [bend right=55] node[below] {$cut$} (i2);

\draw[thick] (1.2,-3.8) to [bend left=55] (2.5,-3.8);
\draw[thick] (1.8,-3.8) to [bend left=55] (3.1,-3.8);
\node(d1) at (1.55,-3.8) {\scriptsize$\dots$};
\node(d2) at (2.85,-3.8) {\scriptsize$\dots$};

\draw[thick] (4.2,-3.8) to [bend left=95] (5.7,-1.3);
\draw[thick] (4.8,-3.8) to [bend left=95] (6.3,-1.3);
\node(d1) at (4.55,-3.8) {\scriptsize$\dots$};
\node(d2) at (6.05,-1.3) {\scriptsize$\dots$};

\node(dd) at (7,-1.3) {\scriptsize$\dots$};
\node(vd) at (0.5,-1.8) {$\vdots$};
\node(vd) at (3.75,-1.8) {$\vdots$};

\end{tikzpicture}}
\end{center}
\caption{Correction graph after weak interpolation}
\label{fig3}
\end{figure}
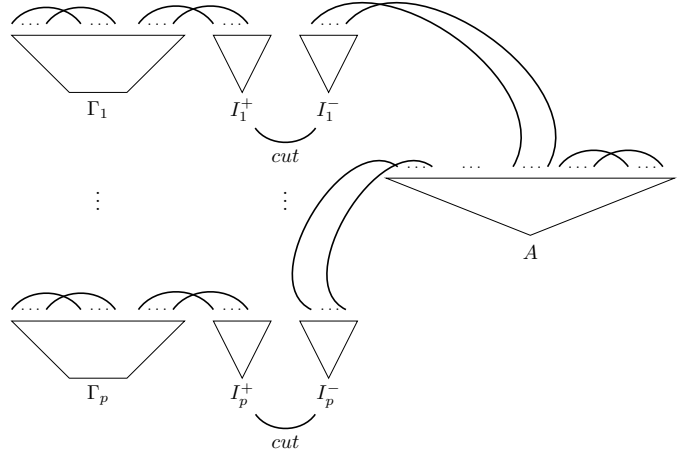
 


%
%
%
%
%
%
%
%
%

By an argument similar to the one in \cite{deGroote1996} we prove that weak interpolation problems admit solutions in $\lambda_{\multimap}$. 

\begin{theorem}[weak $\lambda_{\multimap}$-interpolation]\label{interpolation2}
Any weak interpolation problem for $\lambda_{\multimap}$ admits a solution.
\end{theorem}
\begin{proof}
Proof in appendix \ref{appA}.
\end{proof}

\begin{remark}\label{unique}
For a graph $f:A\to B$, weak interpolation coincides with interpolation. Indeed, by definition \ref{interpo2}, if the weak interpolation problem for $f$ has a solution, then there exist types $I_{1},\dots, I_{p}$ and partitions $f_{I}^{i},f_{II}^{i}$ such that, for one $i$, say $i=c$, $f_{II}^{i}\cup f_{I}^{i}:A\to I_{c}$ and for all $i\neq c$, $f_{II}^{i}\cup f_{I}^{i}: \emptyset \to I_{i}$. However, as there can be no type III free graph $g:\emptyset \to I_{i}$, it must be $p=1$ and $I_{c}$ is the unique interpolant of $f$. 

\end{remark}

Remark \ref{unique} allows to extend the interpolation order to $\lambda_{\multimap}$: $A\prec_{I} B$ if $A$ is the unique interpolant of an arrow $f:B\to B$ in $\CC A$. The linear collapse lemma can then be immediately extended to $\lambda_{\multimap}$:
\begin{lemma}[linear collapse lemma for $\lambda_{\multimap}$]\label{wcolla}
For any type $A$, if $f:A\to A$ is allowable and contains at least one type III (or type II) edge, then $A$ is not $\prec_{I}$-minimal.
\end{lemma}
%
%
%

The linear collapse lemma states that if a type $A$ is $\prec_{I}$-minimal, then any correct graph $f:A\to A$ must be pure a permutation of the identity graph. 

\subsection{The positivity lemma for $\mathit{IMLL^{-}}$}\label{4}

By exploiting interpolation, we prove a fundamental lemma which characterizes the types which have a positive action over allowable graphs: if $A$ is such a type (relatively to a variable $X$) then $A$ collapses into a type containing only positive occurrences of $X$. 

For any variable $X$, we say that a type $A\in\CC L_{\multimap}$ is \emph{p-$X$} (resp. \emph{n-$X$}) if it has no negative (resp. positive) occurrence of $X$.
%
%
%
%
%
%
%
%
%

\begin{definition}[definable $\CC A$-morphism]\label{defifunc}

A \emph{covariant (resp. contravariant) $\CC A$-morphism} is given by a map $F$ over shapes along with a map over allowable graphs in any polynomial extension $\CC A[x_{1},\dots, x_{n}]$ of $\CC A$ such that:
\begin{itemize}
\item if $f:S\to T$ in $\CC A[x_{1},\dots, x_{n}]$, then $F(f): F(S)\to F(T)$ (resp. $F(f):F(T)\to F(S)$ in $\CC A[x_{1},\dots, x_{n}]$). 
\item $F(f)$ preserves labelings: if $f:A\to B$, then $F(f):F(A)\to F(B)$.
\end{itemize}

Given $X\in \CC V$, a type $A\in \CC L_{\multimap}$ is a \emph{definable covariant (resp. contravariant) $\CC A$-morphism in $X$} if there is a covariant (resp. contravariant) $\CC A$-morphism $F$ such that, for any shape $S$, $F(S)=A[S/X]$.

%
%

\end{definition}

%
%
%
%
%

Since p-$X$ (resp- n-$X$) types are correspond to endofunctors $A:\CC A[x_{1},\dots, x_{n}]\to \CC A[x_{1},\dots, x_{n}]$ (resp. $A: \CC A^{op}[x_{1},\dots, x_{n}]\to \CC A[x_{1},\dots, x_{n}]$), they are in particular definable $\CC A$-morphisms.

\begin{proposition}\label{sopra}
For all $X\in \CC V$, if $A\in \CC L_{\multimap,\otimes}$ is p-$X$ (resp. n-$X$), then it is a definable covariant (resp. contravariant) $\CC A$-morphism in $X$.
\end{proposition}
\begin{proof}
By induction on $A$. 
\end{proof}

Conversely, a $\CC A$-morphism need not be an endofunctor over $\CC A$. If a type $A$ is logically equivalent to a p-$X$ type $B$, then one can define a positive action over allowable graphs for $A$ starting from the covariant functorial action of $B$. However, this action need not be functorial. For instance, since the type $A=(X\multimap X)\multimap X$ collapses into $X$, one can define a covariant action of $A$ over allowable graphs by pre-composing and post-composing the functorial action of $X$ with the arrows $g:A\to X$ and $h:X\to A$. However, with such definitions, $A$ is not a functor, as the graph $A(id_{X})= h\circ id_{X} \circ g$ has a type III edge, so $A(id_{X})\neq id_{X}$.

We now show that a definable $\CC A$-morphism always collapses into a p-$X$ type. By reasoning as above, this fact shows that p-$X$ (resp. n-$X$) types are essentially the unique definable covariant endofunctors $A:\CC A\to \CC A$ (resp. $A:\CC A^{op}\to \CC A$).

\begin{proposition}[positivity lemma]\label{functor}
For all $X\in \CC V$, if $A\in \CC L_{\multimap,\otimes}$ is a definable covariant (resp. contravariant) $\CC A$ morphism in $X$, then either $A$ is p-$X$, or $A$ linearly collapses into a p-$X$ (resp. n-$X$) type.
\end{proposition}
\begin{proof}
%
%
%

%
Let $A\in \CC L_{\multimap,\otimes}$ be a definable $\CC A$-morphism in $X$. 
Let $Y$ be a variable not occurring free in $A$ and let $f$ be the only labeled graph such that $f:X\multimap Y, X\to Y$ in $m\CC A$. Since $f$ is an arrow $f:X\to Y$ in $\CC A[x]$, where $x:\emptyset\to X\multimap Y$ is a variable arrow, by hypothesis, there exists an arrow $A(f):A\to A[Y/X]$ in $\CC A$, corresponding to a graph $A(f): X\multimap Y,\dots,X\multimap Y, A \to A[Y/X]$ in $m\CC A$. Observe that all type I edges over $Y$ of $A(f)$ are directed upwards (as $Y$ only occurs positively in the hypotheses of $f$) and all edges over $X$ of $A(f)$ are type $II$, as $X$ does not occur in the conclusion of $f$.

\begin{figure}
\begin{center}
\resizebox{0.4\textwidth}{!}{
\begin{tikzpicture}[baseline=-8ex]

\node(a) at (0,0) {$X \ \multimap \ Y$};

\node(d) at (-1.5,0) {$\dots$};
\node(b) at (-3,0) {$X^{+}$};

\node(e) at (0,-2) {$Y^{+}$};

\draw[<-, thick] (0.4,-0.3) to [bend right=15] (e);

\draw[->, thick] (-0.6,-0.3) to [bend left=55] (-3,-0.3);


\end{tikzpicture}
$\qquad\leadsto\qquad$
\begin{tikzpicture}[baseline=-8ex]
%
\node(b) at (-0.4,0) {$X^{+}$};

\node(e) at (0,-2) {$X^{+}$};

\draw[<-, thick] (b) to [bend left=15] (e);


\end{tikzpicture}}
\end{center}
\caption{Contraction of edges}
\label{contra1}
\end{figure}
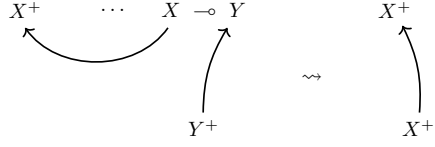
Let $f^{*}$ be the graph obtained from $A(f)$ by contracting all such edges as illustrated in figure \ref{contra1}.
We have then $f^{*}:A\to A^{*}$, where $A^{*}$ is the result of replacing, in $A[Y/X]$ all (positive) occurrences of $Y$ corresponding to type $I$ edges of $A(f)$ by $X$. Observe that $A^{*}$ has no negative occurrence of $X$. Moreover, $f^{*}$ is correct since all the paths in the correction graph of $f^{*}$ come from paths in the correction graph of $A(f)$ (as the latter contains a path from $Y^{+}$ to $X^{-}$).

Now two cases arise: if $f^{*}$ has no type III edge over $Y$, then $A^{*}[Y/X]=A[Y/X]$ and $A$ is p-$X$. Otherwise, if $f^{*}$ has some type III edge over $Y$, we can apply theorem \ref{interpolation} : we find then a type $A'\hookrightarrow A^{*}$ with no occurrence of $Y$ and only positive occurrences of $X$ and two correct graphs $g:A\to A'$ and $h:A'\to A^{*}$, with $h\circ g=f^{*}$. Since all edges over $Y$ in $h$ are type III, by a renaming $Y\mapsto X$ we see that $h:A'\to A$, since $A^{*}[X/Y]=A$. We conclude then that $A'\hookrightarrow A$ and $A,A'$ are interderivable.


\end{proof}

%
%

By proposition \ref{functor} and proposition \ref{iso} we deduce that the property of being p-$X$ (resp. n-$X$) is stable under isomorphism:
\begin{proposition}\label{iso2}
For any $X\in \CC V$ and $A\in \CC L_{\multimap}$ and $B\in \CC L_{\multimap} $ p-$X$ (resp. n-$X$), $A$ is isomorphic to $B$ only if $A$ is p-$X$ (resp. n-$X$).
\end{proposition}
%

\section{Interpolation and collapse in $\lambda_{\To}$}\label{sec5}

In this section we extend the positivity lemma to $\lambda_{\To}$, by proving a similar result for definable $\CC T$-morphisms. We exploit the linearization theorem, a folklore result in linear logic, which shows that arrows in $\CC T$ factor through arrows in $\CC A$.
Linearization allows to apply proof net (weak) interpolation to simply typable $\lambda$-terms, yielding a new weak interpolation result for $\lambda_{\To}$, from which we deduce the positivity lemma similarly to the previous section.

\subsection{The linearization theorem}\label{6}

The basic ingredient to generalize theorem \ref{interpolation2} to $\CC L_{\To}$ is theorem \ref{linear} below, which establishes a precise connection between $\CC T$ and $\CC A$: any arrow in $\CC T$ can be factorized through an arrow in $\CC A$.

We first define linear expansions of types $A\in \CC L_{\To}$.
\begin{definition}[linear expansions]
For any type $A\in \CC L_{\To}$, we define the sets $\CC E^{+}(A),\CC E^{-}(A)\subseteq \CC L_{\multimap}$ of its \emph{positive and negative linear expansions}, respectively, by induction on $A$:
\begin{itemize}
\item if $A=X$, then $\CC E^{+}(A)=\CC E^{-}(A)=\{A\}$;
\item if $A=B_{1}\To \dots \To B_{n}\To Z$, then 
\begin{equation*}
\begin{split}
\CC E^{+}(A) &=\{  C_{11}  \multimap \dots \multimap C_{1k_{1}} \multimap\dots \multimap  
  \multimap C_{n1}\multimap \dots \multimap C_{nk_{n}}\multimap  Z \mid  \\ &\qquad  k_{1},\dots, k_{n}\in \BB N, C_{ij}\in \CC E^{-}(B_{i}), \forall 1\leq i\leq n, 1\leq j\leq k_{i}\}\\
\CC E^{-}(A)&  =\{ C_{1}\multimap \dots \multimap C_{n}\multimap Z \mid C_{i}\in \CC E^{+}(B_{i})\}
\end{split}
\end{equation*}

\end{itemize}
\end{definition}

The following are easily established.

\begin{proposition}
\begin{itemize}
\item if $B\in \CC E^{+}(A)$, then there exists $u:B^{\To}\to A$ in $\CC T$;
\item if $B\in \CC E^{-}(A)$, then there exists $u:A\to B^{\To}$ in $\CC T$.

\end{itemize}
\end{proposition}

We can now prove the linearization theorem.

\begin{theorem}[Linearization]\label{linear}
Let $u:\Gamma\to A$ be an arrow in $\CC T$, where $\Gamma=\{A_{1},\dots, A_{p}\}$. Then there exist integers $d_{1},\dots, d_{p}$, linear types $A_{ij}^{-}\in \CC E^{-}(A_{i}), A^{+}\in \CC E^{+}(A)$, arrows $LIN_{ij}:A_{i}\to (A_{ij}^{-})^{\To}$, $coLIN:(A^{+})^{\To}\to A$ in $\CC T$ and an arrow $\partial u: \Gamma^{-}\to A^{+}$ in $\CC A$, where $\Gamma^{-}=\{A_{ij}^{-}\mid 1\leq i\leq p, 1\leq j\leq d_{i}\}$, such that the following diagram commutes in $\CC T$
(by seeing $\partial u: (\Gamma^{-})^{\To} \to (A^{+})^{\To}$ as an arrow in $\CC T$)
 \begin{equation}
\xymatrix{
\Gamma \ar[r]^{u} \ar[d]_{LIN} & A \\
(\Gamma^{-})^{\To} \ar[r]_{\partial u} & (A^{+})^{\To} \ar[u]_{coLIN}
}
\end{equation}
where $LIN=(LIN_{1},\dots, LIN_{p})$ and $LIN_{i}=(LIN_{i1},\dots, LIN_{id_{i}})$, that is,
$$  coLIN\big [\partial u [LIN_{ij}/x_{ij}] /x\big ]\ \simeq_{\beta\eta} \ u $$

\end{theorem}
%
%
\begin{proof}

Induction on the typing derivation of $u$: 
\begin{itemize}
\item if $u=x$, then we have $x:A\vdash x:A$, so we put $ A^{-}=A^{+}=A^{\multimap}$, $d_{1}=1$ and $\partial u=u$;

\item if $u=\lambda z^{B_{1}}.u'$, then we have $\AXC{$\Gamma, z:B_{1}\vdash u':B_{2}$}\UIC{$\Gamma\vdash u:B_{1}\To B_{2}$}\DP$, where $\Gamma=\{A_{1},\dots, A_{p}\}$, then by induction hypothesis there exist integers $e_{1},\dots e_{p+1}$, types $C_{11},\dots, C_{pe_{p}}$, $D_{1},\dots, D_{e_{p+1}},D$, arrows $v_{ij}:A_{j}\to C_{ij}$, $v'_{l}[y]: B_{1}\to D_{l}$, $w:D\to B_{2}$ and a linear term $\partial u'$ such that 
$x_{11}:C_{11},\dots, x_{pe_{p}}:C_{pe_{p}},y_{1}:D_{1},\dots, y_{e_{p+1}}:D_{e_{p+1}}\vdash\partial u: D$ and 
$w \big[ \partial u'[v_{ij}/x_{ij}, v'_{l}/y_{l}]/x\big] \simeq_{\beta\eta} u'$.

We let then $d_{1}=e_{1},\dots, d_{p}=e_{p}$, $ A_{ij}^{-}=C_{ij}$, $A^{+}= D_{1}\To \dots \To D_{e_{p+1}}\To D$, and moreover $LIN_{ij}:=u_{ij}$, $coLIN = \lambda z^{B_{1}}. w \big [x (v'_{1}[z/y])\dots (v'_{e_{p+1}}[z/y])/x\big ]$ and $\partial u:=\lambda y_{1}^{D_{1}}.\dots \lambda y_{e_{p+1}}^{D_{e_{p+1}}}.\partial u'$. We can compute then 
$coLIN(\partial u[LIN_{ij}/x_{ij}]) \simeq_{\beta\eta} \lambda z^{B_{1}}.u' \simeq_{\beta\eta} u$

\item if $u=xu_{1}\dots u_{q}$, then $[x]=B_{1}\To \dots \To B_{q}\To A$ and $\Gamma\vdash u_{i}:B_{i}$, where $\Gamma=\{A_{1},\dots, A_{p}\}$ and $A_{k}=B_{1}\To \dots \To B_{q}\To A$, for some $1\leq k\leq p$, hence by induction hypothesis for all $1\leq l\leq q$
%
there exist integers $e_{1}^{l},\dots, e_{p}^{l}$, types $C_{11}^{l},\dots, C_{pe_{p}}^{l},C^{l}$, arrows $v_{ij}^{l}:B_{l}\to P_{ij}^{l}$, $w^{l}[y]:C^{l}\to B_{l}$ and a linear term $\partial u_{l}$ such that $x^l_{11}:C^{l}_{11},\dots, x^{l}_{pe_{p}}:C^{l}_{pe_{p}}\vdash \partial u_{l}: C^{l}$ and 
$w^{l}\big [\partial u_{l}[v^{l}_{ij}/x^{l}_{ij}]/x\big ]\simeq_{\beta\eta} u_{l}$.

We let then, for $1\leq i\leq q$, $d_{i}:= \sum_{l}^{q}e_{i}^{l}$, for $i\neq k$, and $d_{k}=(\sum_{l}^{q}e_{k}^{l} )+1$, we let, for $i\neq k$ the $A_{ij}^{-}$ be all $C_{ij}^{l}$, for all $1\leq l\leq q$; for $j\leq d_{k}-1$, we let the $A_{jk}^{-}$ be all $C_{kj}^{l}$, for all $1\leq l\leq q$, and we finally let $A_{d_{k}k}^{-}$ be $C^{1}\To \dots \To C^{q}\To A$. 
Also, for $i\neq k$, we let the $LIN_{ij}$ be all $v_{ij}^{l}$ and for $j\leq d_{k}-1$, the $LIN_{kj}$ be all $v_{kj}^{l}$, and we let $LIN_{kd_{k}}[x]=\lambda z_{1}^{C_{1}}.\dots.\lambda z_{q}^{C_{q}}.x (w^{1}[z_{1}/y])\dots (w^{q}[z^{q}/y])$ and $coLIN=x$. We finally let 
$\partial u:= x_{kd_{k}}\partial u_{1}\dots \partial u_{q}$. 
We can compute then $coLIN\big [\partial u[LIN_{ij}/x_{ij}]/x\big ]\simeq_{\beta} LIN_{kd_{k}} \big (w^{1}\big [\partial u_{1}[v_{ij}^{1}/x_{ij}^{1}]/y\big ]\big ) \dots \big (w^{q}\big [\partial u_{q}[v_{ij}^{q}/x_{ij}^{q}]/y\big ])\big ) \simeq_{\beta\eta} u$.

\end{itemize}
\end{proof}

Let us call a $\lambda$-term \emph{clean}  if $FV(u)\cap BV(u)=\emptyset$ and for any $x\in BV(u)$ there exists in $u$ exactly one abstraction of the form $\lambda x$. Clearly, any $\lambda$-term $u$ is $\alpha$-equivalent to a clean term, that we indicate $u^{*}$.

For any $\lambda$-term $u$,  
and any variable $x\in FV(u^{*})\cup BV(u^{*})$, let $s_{x}(u)$ indicate the number of occurrences of $x$ in $u^{*}$. More precisely, we let 
$$
\begin{matrix}
s_{x}(y) \  = \ \begin{cases} 1 & \text{if } x=y \\ 0 & \text{otherwise} \end{cases} \\
s_{x}(tu)\  = \ s_{x}t+s_{x}u \\
s_{x}(\lambda y.t) \  = \ \begin{cases} s_{x}(t)+1 & \text{if } x=y \\ s_{x}(t) & \text{otherwise} \end{cases} 
\end{matrix}
$$
 Let then the \emph{size of $u$}, noted $s(u)$, be the sum of all $s_{x}(u)$, for $x\in FV(u^{*})\cup BV(u^{*})$. 
Linearization preserves the size, in the following sense:
\begin{lemma}\label{size}
Let $u:\Gamma\to A$ be an arrow in $\CC T$ and $\partial u:\Gamma^{-}\to A^{+}$ be defined as in theorem \ref{linear}. Then
$s(u)=s(\partial u)$.
\end{lemma}
\begin{proof}
Simple verification by inspecting the proof of theorem \ref{linear}.

\end{proof}

\subsection{Weak interpolation in $\lambda_{\To}$}

Similarly to the case of $\lambda_{\multimap}$, interpolation fails for $\lambda_{\To}$ (i.e. for the $\To$-fragment of intuitionistic propositional logic): the interpolant of $ \lambda z^{A\To B\To X}.zxy:  A,B\to (A\To B\To X)\To X$ is the product type $A\land B$ which is not a type in $\CC L_{\To}$. 

Weak interpolation for $\lambda_{\To}$ was investigated in \cite{Wronski, Pentus, Kanazawa}. In particular, \cite{Kanazawa} shows that, by suitably modifying the algorithm in \cite{Prawitz1965}, one can always find strongest interpolants. By exploiting linearization and weak interpolation for $\lambda_{\multimap}$ we immediately obtain a new weak interpolation theorem for $\lambda_{\To}$, theorem \ref{interpolationl} below.

Due to the loss of linearity, the requirement that interpolant types inject into the conclusion and the context is replaced by the weaker requirement that the free type variables of the interpolant types be included into the free variables shared by the conclusion and the context. 

\begin{theorem}[weak interpolation for $\CC T$]\label{interpolationl}
Let $u:\Gamma\to A$ in $\CC T$ and suppose $\partial u$ is not pure. Then there exist simple types $I_{1},\dots, I_{p}$ and arrows $v_{j}:\Gamma\to I_{j}$, $w:I_{1},\dots, I_{p}\to A$ such that
%
%


\begin{enumerate}

\item $FV(I_{j})\subseteq FV(A)\cap FV(\Gamma)$;
\item $(\sum_{j}^{p}s(v_{j}))+ s(v)< s(u)$;

\item $v[v_{1}/x_{1},\dots, v_{p}/x_{p}] \ \simeq_{\beta\eta}\  u  $.
\end{enumerate}

\end{theorem}
\begin{proof}
By theorem \ref{interpolation2} there exists an integer $p$, a partition $(\Gamma_{1}^{-},\dots,\Gamma_{p}^{-})$ of $\Gamma^{-}$, types $I_{1},\dots, I_{p}$ and a splitting of $\partial u$ in $\CC A$ into $v_{j}:\Gamma_{j}^{-}\to I_{j}$ and $v:I_{1},\dots, I_{p}\to A^{+}$.

\end{proof}

Observe that, while the logical complexity of weak interpolants for an arrow $f:\Gamma\to A$ in $\CC A$ is bounded by $A$, this need not be the case in $\CC T$. For instance, let $A=(X\To X)\To (X\To X)$ and $u_{k}:A\to A$ be the term $\lambda y_{1}^{X\To X}.\lambda y_{2}^{X}. y_{1}^{k}(xy_{1}y_{2})$, for $k\in \BB N$. Let, for all $k\in \BB N$, $A_{k}= \underbrace{(X\multimap X)\multimap \dots \multimap (X\multimap X)}_{k \text{ times}}\multimap A^{\multimap}$. Then $\partial u_{k+1}: A^{\multimap} \to A_{k+2}$ is not pure (see figure \ref{impure}) and can be split into 
$u_{1}:A^{\multimap}\to A_{k+1}$ and $u_{2}:A_{k+1}\to A_{k+2}$. We deduce that $u_{k}$ is split into $v_{1}: A\to A_{k+1}^{\To}$ and $v_{2}:A_{k+1}^{\To }\to A$. The problem of investigating the growth in complexity of interpolants is well investigated in the literature (see \cite{Pruim1998}).
Observe that, while the size of interpolants might grow, the size of the terms decreases (condition 2. in theorem \ref{interpolationl}), as $s(v_{1})< s(u)$ (one can easily compute $v_{1}= u_{k}$).

\begin{figure}
\begin{center}
\resizebox{0.6\textwidth}{!}{
\begin{tikzpicture}
\node(a) at (0,1) {$(X\multimap X)\multimap (X\multimap X)$};

\node(b) at (0,-2) {$\underbrace{(X\multimap X)\multimap (X\multimap X)\multimap \dots \multimap (X\multimap X)}_{k+1 \text{ times}}\multimap (X\multimap X)\multimap (X\multimap X)$};

\draw[->, thick] (-3.2,-1.5) to [bend right=75] (-4.2,-1.5);
\draw[->, thick] (-1.8,-1.2) to [bend right=55] (-2.2,-1.5);

\draw[thick] (0.2,-1.5) to [bend right=55] (-0.2,-1.2);
\node(d) at (-1,-1.2) {$\dots$};

\draw[->, thick] (5.2,-1.5) to [bend right=55] (1,-1.5);
\draw[->, thick] (-5.2,-1.5) to [bend left=25] (-1.8,-0.6) to [bend right=25] (1.5, 0.8);

\draw[->, thick] (2.2,-1.5) to [bend right=15] (0.45,-0.6)   to  [bend left=25] (-1.5,0.8);
\draw[<-, thick] (3.1,-1.5) to [bend right=15] (1.45,-0.6)   to  [bend left=25] (-0.5,0.8);
\draw[<-, thick] (4.2,-1.5) to [bend right=15] (2.45,-0.6)   to  [bend left=25] (0.5,0.8);

\end{tikzpicture}
}
\end{center}
\caption{Graph of $\partial u_{k+1}$}
\label{impure}
\end{figure}
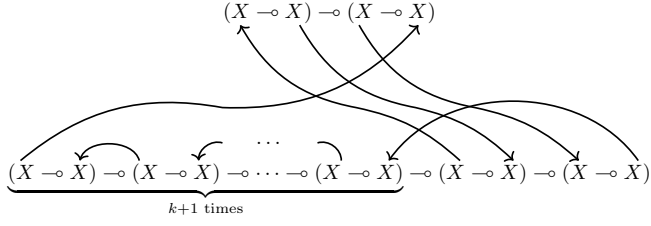

Remark \ref{unique} cannot be extended to $\lambda_{\To}$. In particular, an arrow $u:A\to B$ in $\CC T$ might well have more than one interpolant. Take for instance $u: C\to C$, where
$C=(X\To X)\To( ((A\To X)\To X)\To ( (B\To X)\To X)\To X)\To X$ and 
$u[x]= \lambda x_{1}^{C_{1}}.\lambda x_{2}^{C_{2}}.x_{1} ( x_{2} u_{1}u_{2})$, where $C_{1}=X\To X$, $C_{2}=((A\To X)\To X)\To ((B\To X)\To X)\To X$ and $u_{1}=\lambda z^{A\To X}.x (\lambda y^{X}.y 
(\lambda y_{1}^{(A\To X)\To X}.\lambda y_{2}^{(B\To X)\To X}.y_{1}z))$, 
$u_{2}=
\lambda z^{B\To X}.x (\lambda y^{X}.y (
\lambda y_{1}^{(A\To X)\To X}.\lambda y_{2}^{(B\To X)\To X}.y_{1}z))$. One can verify that $u$ has two interpolants $I_{1}=(A\To X)\To X$ and $I_{2}=(B\To X)\To X$.

We define then the \emph{weak interpolation} order $\prec_{wI}$ as follows: $A\prec_{wI} B$ if $A$ is among the interpolants of some arrow $u:B\to B$ in $\CC T$. When $\Gamma=\{I_{1},\dots, I_{p}\}$ are the interpolants of some arrow $u:B\to B$, we say that $B$ \emph{collapses} into $\Gamma$. If for no $\Gamma$, $B$ collapses into $\Gamma$, then $B$ is called \emph{$\prec_{wI}$-minimal}.

%
%

The linear collapse lemma \ref{colla} can be extended to $\CC L_{\To}$:

\begin{lemma}[collapse lemma for $\CC T$]\label{lambdacolla}
Let $u:A\to A$ be an arrow in $\CC T$ such that $\partial u$ contains a type III edge. Then $A$ is not $\prec_{wI}$-minimal.
\end{lemma}

For instance, the arrow $u:C\to C$ discussed above shows that $C$ is not $\prec_{wI}$-minimal and collapses into two minimal types $I_{1},I_{2}$.


\subsection{The positivity lemma for $\lambda_{\To}$}\label{7}

By exploiting weak interpolation, we now extend the positivity lemma \ref{functor} to $\lambda_{\To}$.


The notions of p-$X$ and n-$X$ types for $\CC L_{\To}$ are defined as for $\CC L_{\multimap}$.  Observe that, if $B\in \CC E^{\epsilon}(A)$, then $B$ is p-$X$ (resp n-$X$) iff $A$ is.
We define definable $\CC T$-morphisms similarly to the case of $\CC A$.

\begin{definition}[definable $\CC T$-morphism]\label{defifuncT}

A \emph{covariant (resp. contravariant) $\CC T$-morphism} is given by a map $F$ over simple types along with a map over arrows in any polynomial extension $\CC T[x_{1},\dots, x_{n}]$ of $\CC T$, such that, if $u:A\to B$ is an arrow in $\CC T[x_{1},\dots, x_{n}]$, then $F(u): F(A)\to F(B)$ (resp. $F(u):F(B)\to F(A)$) is an arrow in $\CC T[x_{1},\dots, x_{n}]$.

Given $X\in \CC V$, a simple type $A\in \CC L_{\To}$ is a \emph{definable covariant (resp. contravariant) $\CC T$-morphism} in $X$ if there is a covariant (resp. contravariant) $\CC T$ morphism $F$ such that, for any simple type $B$, $F(B)=A[B/X]$.

%
%

\end{definition}

As for $\CC A$, for one direction, the fact that p-$X$ (resp n-$X$) types correspond to covariant (resp. contravariant) functors can be deduced from the fact that all types are multivariant functors.

\begin{proposition}
For all $X\in \CC V$, if $A$ is p-$X$ (resp. n-$X$), then it is a covariant (resp. contravariant) definable $\CC T$-morphism in $X$.
\end{proposition}
\begin{proof}
By induction on $A$.
\end{proof}

The following proposition extends proposition \ref{functor} to $\CC T$.

\begin{proposition}\label{functor1}
For all $X\in \CC V$, if $A\in \CC L_{\To}$ is a covariant (resp. contravariant) definable $\CC T$-morphism in $X$, then either $A$ is p-$X$ (resp. n-$X$) or $A$ collapses into a finite set of p-$X$ (resp. n-$X$) types. 
\end{proposition}
\begin{proof}
Let $A\in \CC L_{\To}$ be a covariant definable $\CC T$-morphism in $X$. Let $Y$ be a variable not occurring free in $A$ and let $u=xy:X\To Y, X\to Y$ in $m\CC T$. Since $u$ is an arrow $u:X\to Y$ in $\CC T[x]$, where $x:\emptyset\to X\To Y$ is a variable arrow, by hypothesis, there exists an arrow $A(u):A\to A[Y/X]$, i.e. a $\lambda$-term $v=A(u):X\To Y, A\to A[Y/X]$. 

By theorem \ref{linear}, $v$ is $\beta\eta$-equivalent to 
$LIN(\partial v[LIN_{j}/x_{j}])$, where $\partial v:X\multimap Y, \dots, X\multimap Y, A^{-}_{1},\dots,A^{-}_{r}\to A^{+}[Y/X]$, where $A^{-}_{i}\in \CC E^{-}(A)$ and $A^{+}\in \CC E^{+}(A)$.

Suppose now $A$ is not p-$X$; we deduce that $A^{+}$ is not p-$X$ either. By reasoning similarly to the proof of proposition \ref{functor} we deduce that, if $\partial v$ has no type III edge over $X$, then $A^{+}$ (and a fortiori $A$) is p-$X$, while if 
$\partial v$ has a type III edge over $X$, then there exist p-$X$ types $A_{1},\dots, A_{p}\in \CC L_{\multimap}$ such that $A_{i}\hookrightarrow A^{+}$ and linear terms $v_{1},\dots, v_{p},w$ such that 
$v_{i}:A^{-}_{1},\dots, A^{-}_{r_{i}}\to A_{i}$ and $w:A_{1},\dots, A_{p}\to A^{+}$. We conclude that $A$ collapses into the p-$X$ types $A_{1}^{\To},\dots, A_{p}^{\To}$. 


\end{proof}

By proposition \ref{functor1} and the fact (proved in \cite{Bruce91}) that type isomorphism $\CC L_{\To}$ is trivial, i.e. $A\simeq B$ iff $A=B$, we have:

\begin{proposition}
For any $A\in \CC L_{\To}$ and $B\in \CC L_{\To}$ p-$X$ (resp. n-$X$), $A$ is isomorphic to $B$ only if $A$ is p-$X$ (resp. n-$X$).
\end{proposition}

%
%

\section{Instantiation overflow and Russell-Prawitz types}\label{sec6}

In this section we introduce the instantiation overflow property and its relation to Russell-Prawitz types. Moreover, we introduce some classes of types which generalize Russell-Prawitz types, though preserving the instantiation overflow property, and which will play a central role in the next sections.
Our formulation of instantiation overflow essentially follows the functorial approach developed in \cite{StudiaLogica1},
which yields expansion terms equivalent to those in \cite{Ferreira2013}.

Functorial polymorphism is a useful tool to investigate the Russell-Prawitz translation, as the requirement that the translated connectives satisfy all properties of the original connectives corresponds to a dinaturality condition (\cite{Hasegawa2009}). In categorial terms, this requirement corresponds to asking that the $RP$ translation preserves universal properties of connectives: for instance, that the $RP$ translation of conjunction and disjunction preserve the universal properties of products and coproducts, respectively. In proof-theoretic terms, this means asking that the $RP$ translation preserves the equational theory over derivations generated by $\beta,\eta$-equivalences and permuting conversions (see \cite{StudiaLogica}).

We will now describe instantiation overflow within the functorial framework introduced in subsection \ref{sec23} and show that, when $A$ is a $RP$ type, the expansion term $IO_{A}(B):\forall XA\to A[B/X]$, is equivalent, modulo dinaturality, to the term $xB$, corresponding to an instance of full extraction.


\begin{definition}[instantiation overflow]
A type of the form $\forall XA$ has the \emph{instantiation overflow} property ($IO$ for short) if, for any $B\in \CC L_{\To,\forall}$, there exists an arrow $IO_{A}(B):\forall XA\to A[B/X]$ in $\CC F_{at}$.

\end{definition}

As it was discussed in the introduction, the types figuring as the Russell-Prawitz translation of logical connectives 
are examples of types having $IO$. These
are types of the form
$\forall X(A_{1}\To \dots \To A_{p}\To X)$\footnote{Here, for simplicity, we only consider types translating propositional connectives. A more general treatment of Russell-Prawitz types requires to define them as types of the form $\forall X\forall \OV Y_{1}(A_{1}\To \dots \To \forall \OV Y_{p}(A_{p}\To X))$, in order to account also for the translation of second order existential quantification $\exists YA$ as $\forall X(\forall Y(A\To X)\To X)$.}, where the $A_{i}$ are in turn of the form
$B_{i}^{1}\To \dots \To B_{i}^{n_{i}}\To X$, with $X\notin FV(B_{i}^{j})$.

Observe that the types $A_{i}$ are p-$X$. This remark allows to define $RP$ types formally as follows:

\begin{definition}[$RP$ type]\label{rpx}
For any $X\in\CC V$, a type $A\in \CC L_{\To, \forall}$ is called a \emph{Russell-Prawitz type in $X$} ($RP_{X}$ for short) if $A= A_{1}\To \dots \To A_{p}\To X$, where the $A_{i}$ are p-$X$.

A type $A\in \CC L_{\To, \forall}$ is called a \emph{Russell-Prawitz type} ($RP$ for short) if $A=\forall XA'$, where $A'$ is $RP_{X}$. 

\end{definition}

We now show how to construct expansion terms for $RP$ types by exploiting the functoriality of p-$X$ types.

\begin{proposition}
Any $RP$ type has $IO$.
\end{proposition}
\begin{proof}
Let $B$ be the type $\forall \OV Y_{1}(B_{1}\To \forall \OV Y_{2}(B_{2}\To \dots \To \forall \OV Y_{n}(B_{n}\To \forall \OV Y_{n+1}Z))$. We let $Elim_{B}: B,B_{1},\dots, B_{n}\to Z$ be the term below
$$Elim_{B}:= x \OV Y_{1} x_{1}\OV Y_{2}x_{2} \dots \OV Y_{n} x_{n}\OV Y_{n+1}$$
where $\OV Y_{i}$ denotes a sequence of type variables $Y_{i1}\dots Y_{ik_{i}}$. We have that $x:B, \Delta\vdash Elim_{B}: Z$, where $\Delta=\{x_{i}:B_{i}\mid 1\leq i\leq n\}$. 
For any arrow $u:\Gamma,\Delta\to Z$, where $\Delta$ is as before, we let $Intro_{B}(u)$ be the term below
$$Intro_{B}(u)=\Lambda\OV Y_{1}.\lambda x_{1}^{B_{1}}.\Lambda \OV Y_{2}.\lambda x_{2}^{B_{2}}.\dots.\Lambda\OV Y_{n}.u$$
where $\Lambda \OV Y_{i}$ indicates a finite sequence of abstractions $\Lambda Y_{i1}.\dots.\Lambda Y_{ik_{i}}$. We have that $\Gamma\vdash Intro_{B}(u): B$

Let now $A=A_{1}\To \dots \To A_{p}\To X$ be $RP_{X}$. We can define $IO_{A}(B)$ as follows:
$$
IO_{A}(B)=
\lambda x_{1}^{A_{1}[B/X]}.\dots.\lambda x_{p}^{A_{p}[B/X]}.Intro_{B} \big (
xZ (A_{1}(Elim_{B})\dots (A_{p}(Elim_{B}))\big )$$
where $A_{i}(Elim_{B})$ denotes the application of the definable $\CC T$-morphism $A_{i}$ to the arrow $Elim_{B}:B\to X$ in $\CC T$ and $x_{i}$ is the unique variable of $A_{i}(Elim_{B})$ of type $A_{i}[B/X]$. 

The term $IO_{A}(B)$ corresponds to the derivation illustrated in figure \ref{IOintro} in the introduction.

\end{proof}

We now show that the arrows $IO_{A}(B):\forall XA\to A[B/X]$ in $\CC F_{at}$ are equivalent, modulo dinaturality, to the arrows $xB:\forall XA\to A[B/X]$ in $\CC F$, obtained by one instance of the full extraction rule. This fact says that atomic polymorphism and full polymorphism for $RP$ types are indistiguishable modulo dinaturality/parametricity.

We recall that $\simeq_{\varepsilon}$, introduced in subsection \ref{sec23}, indicates the $\lambda^{2}$-theory extending $\beta\eta$-equivalence and generated by dinaturality.

\begin{proposition}\label{dinatural}

Let $\forall XA\in \CC L_{\To}$ be $RP$. Then, for all $B\in \CC L_{\To, \forall}$, $IO_{A}(B)\simeq_{\varepsilon} xB$.



\end{proposition}
\begin{proof}
Let $A=A_{1}\To \dots \To A_{p}\To X$, where the $A_{i}$ are p-$X$. For all $B=
\forall \OV Y_{1}(B_{1}\To \dots \forall \OV Y_{n}(B_{n}\To \forall \OV Y_{n+1} Z)\dots )$, let us consider the polynomial category $\CC F[x_{1},\dots, x_{n}]$, with variable arrows $x_{i}:\emptyset \to B_{i}$.

Let, for any type $C$, $ext_{C}[x]$ indicate the term $xC$; we have that 
$A(Elim_{B},Z)\circ ext_{Z}= ext_{Z} (A_{1}(Elim_{B}))\dots (A_{p}(Elim_{B})): \forall XA \to A[Z,Z]$, 
hence $IO_{A}(B)$ can be written as
$$IO_{A}(B)=\lambda z_{1}^{A_{1}[B/X]}.\dots.\lambda z_{p}^{A_{p}[B/X]}.Intro_{B}( A(Elim_{B},Z)\circ ext_{Z})$$
Observe that $A(Elim_{B},Z): A[Z,Z]\to A[B,Z]$ and $A(B,Elim_{B}): A[B,B]\to A[B,Z]$. Hence, dinaturality yields 
\begin{center}
\resizebox{0.4\textwidth}{!}{
$$
\xymatrixcolsep{5pc}\xymatrix{
\forall XA \ar[r]^{ext_{Z}[x]}\ar[d]_{ext_{B}[x]}  &  A[Z,Z] \ar[d]^{A(Elim_{B},Z)} \\
A[B,B] \ar[r]_{A(B, Elim_{B})} & A[B,Z]
}
$$}
\end{center}
i.e. 
$ A(Elim_{B},Z)\circ ext_{Z}  \ \simeq_{\varepsilon} \  A(B,Elim_{B})\circ ext_{B}$, and we can compute
$A(B,Elim_{B})\circ ext_{B} \simeq_{\beta\eta}
\lambda z_{1}^{A_{1}[B/X]}.\dots.\lambda z_{p}^{A_{p}[B/X]}.Elim_{B}[ext_{B}z_{1}\dots z_{n}/x]$. We finally deduce 
$$IO_{ A}(B)= \lambda z_{1}^{A_{1}[B/X]}.\dots.\lambda z_{p}^{A_{p}[B/X]}.Intro_{B} (Elim_{B}[ext_{B}[x]z_{1}\dots z_{p}/x]) \  \simeq_{\beta\eta} \  ext_{B}[x]$$


\end{proof}

We introduce now two classes of types, \emph{quasi Russell-Prawitz types} and \emph{generalized Russell-Prawitz types}, which will play a central role in the next sections.
These classes generalize definition \ref{rpx} though preserving the instantiation overflow property. 

\begin{definition}[$qRP$ and $gRP$ types]\label{grpx}

A \emph{quasi Russell-Prawitz type in $X$} ($qRP_{X}$) is a type of the form $A_{0}[R/X]$, where $A_{0}$ is p-$X$ and has a unique occurrence of $X$ and $R$ is $RP_{X}$. If $\CC X\subset \CC V$, a \emph{quasi Russell-Prawitz type in $\CC X$} ($qRP_{\CC X}$) is a type which is $qRP_{X}$ for all $X\in \CC X$;

A \emph{generalized Russell-Prawitz type in $X$} ($gRP_{X}$) is a type of the form $A[X/X_{1},\dots, X/X_{p}]$, where $A$ is $qRP_{\{X_{1},\dots, X_{p}\}}$.

For all $A\in \CC L_{\To}$, the type $\forall XA$ is called \emph{quasi Russell-Prawitz}, $qRP$ for short (resp. \emph{generalized Russell-Prawitz}, $gRP$ for short), when $A$ is $qRP_{X}$ (resp. $A$ is $gRP_{X}$). 

\end{definition}

An example of a $qRP$ type is the type $\forall X((X\To X)\To C)\To D)$, where $X\notin FV(C), FV(D)$. Examples of $gRP$ types are $\forall X((X\To X)\To X)\To X)$, $\forall X(X\To  (((A\To X)\To X)\To B)\To X)$, where $x\notin FV(A),FV(B)$.

In section \ref{sec9} we briefly discuss the relation between $gRP$ types and the Russell-Prawitz translation. 
As it will be clear from the next section, $qRP$ and $gRP$ types generalize $RP$ types in the sense that we can construct expansion terms $IO_{A}(B)$ for them in a way similar to $RP$ types.
Indeed, one can easily extend instantiation overflow to quasi Russell-Prawitz types by exploiting the functoriality of $A_{0}$. The extension to generalized Russell-Prawitz types is a bit more involved and will be treated in the following sections, through an equivalent inductive definition of $gRP_{X}$ types. 

%
%
%
%
%
%
%
%
%
%
%

Generalized Russell-Prawitz types can be seen as $qRP_{\CC X}$ types from which we deleted informations about how to localize its Russell-Prawitz subtypes. However, we will show (proposition \ref{renaming2}) that, given a closed term of type a $gRP_{X}$ type, one can always ``separate variables'', i.e. rename variables so to transform the term into a term of type a $qRP_{\CC X}$ type. 


%
The definition of $RP_{X}, qRP_{X}$ and $gRP_{X}$ types can be adapted to $\lambda_{\multimap}$.
Linear Russell-Prawitz types are defined as follows:
\begin{definition}[linear $RP_{X}$ type]\label{lrpx}
For any variable $X$, a type $A\in \CC L_{\multimap}$ is called a \emph{linear Russell-Prawitz type in $X$} if $A=A_{1}\multimap \dots \multimap A_{p}\multimap X$, where all $A_{i}$ have no occurrence of $X$ but one, which is p-$X$ and has exactly one positive occurrence of $X$.

\end{definition}

For instance, the type $C=(A\multimap B\multimap X)\multimap X$ (we suppose $X\notin FV(A), FV(AB)$), translating the linear connective $A\otimes B$, is linear $RP_{X}$. Clearly, if $A$ is a linear $RP_{X}$ type, then $A^{\To}$ is $RP_{X}$. However the converse need not hold: for instance, while $(A\To X)\To (B\To X)\To X$ is $RP_{X}$ (again, we suppose $X\notin FV(A), FV(AB)$) the type $(A\multimap X)\multimap (B\multimap X)\To X$ is not a linear Russell-Prawitz type in $X$. 

Linear $qRP_{X}$ and $gRP_{X}$ types are defined as in definition \ref{grpx}. For instance, let $D=(A\multimap B\multimap X)\multimap X$, where $X\notin FV(A),FV(B)$ be the linear $RP_{X}$ type discussed above. The type $(D\multimap Y)\multimap (Y\multimap Y)\multimap Z$ is linear $qRP_{X}$. The types $((X\multimap X)\multimap X')\multimap X'$ and  $(D\multimap Y) \multimap (D[X'/X]\multimap Y)\multimap Z$ are linear $qRP_{\{X,X'\}}$. Finally, the types $((X\multimap X)\multimap X)\multimap X$ and $(D\multimap Y) \multimap (D\multimap Y)\multimap Z$ are linear $gRP_{X}$.


\section{The linear expansion property}\label{sec7}

In this section we characterize the types $A\in \CC L_{\multimap}$ which have the \emph{linear expansion property}: for any $B\in \CC L_{\multimap}$, there exists an arrow $EXP_{A}(B): A[Z/X]\to A[B/X]$ in $\CC A$, where $Z$ is the rightmost variable in $B$. This property is obviously related to instantiation overflow: if $A$ is linearly expansible, then $\forall XA^{\To}$ has instantiation overflow. 

We characterize the class of linearly expansible types by showing that linear $gRP_X$ types are ``dense'' in that class. More precisely, we show that, if a type $A$ is linearly expansible, then either $A$ is $gRP_{X}$ or $A$ linearly collapses into a $gRP_X$ type.
Our argument is based on two graphical characterizations of linear $gRP_{X}$ types as (1) those types admitting an internal pairing (proposition \ref{pairing}) and (2) those paired types for which all simple expansion graphs are correct (proposition \ref{$gRP_X$2}). 
%
%
%
%

\subsection{Linearly expansible types and simple expansion graphs}

We fix for all this section a variable $X$. For any type $A$, we indicate by $n(A)$ the number of occurrences of $X$ in $A$ and by $n^{+}(A)$ (resp. $n^{-}(A)$) the number of positive (resp. negative) occurrences of $X$ in $A$, respectively. We will call a type $A$ \emph{paired} in $n^{+}(A)=n^{-}(A)$.

%
%

\begin{definition}\label{wexpa}
A type $A\in \CC L_{\multimap}$ is \emph{weakly linearly expansible in $X$} (resp. \emph{linearly expansible in $X$}) if for every linear types $C_{1},\dots, C_{p}\in \CC L_{\multimap}$, there exists a graph (resp. an allowable graph) $EXP_{A}(B): A \to A[B/X]$, where $B$ is 
$C_{1}\multimap \dots \multimap C_{p}\multimap X$. 

Dually, a type $A\in \CC L_{\multimap}$ is \emph{weakly co-linearly expansible in $X$} (resp. \emph{co-linearly expansible in $X$}) if for every linear types $C_{1},\dots, C_{p}\in \CC L_{\multimap}$, there exists a graph (resp. an allowable graph) $coEXP_{A}(B): A[B/X] \to A$, where $B$ is 
$C_{1}\multimap \dots \multimap C_{p}\multimap X$.

\end{definition}

If $A$ is linearly expansible then, for any linear type $B=B_{1}\multimap \dots \multimap B_{n}\multimap Y$, there exists a correct graph $f:A[Y/X]\to A[B/X]$: indeed, by definition \ref{wexpa} there exists a correct graph $f:A\to A[C/X]$,   where $C=B'_{1}\multimap \dots \multimap B'_{n}\multimap X$, and the $B'_{i}$ are such that $B'_{i}[Y/X]=B_{i}$, hence in particular, $f$ is a correct graph $f:A[Y/X]\to A[B/X]$. 
Conversely, if for any linear type $B=B_{1}\multimap \dots \multimap B_{n}\multimap Y$, there exists a correct graph $f:A[Y/X]\to A[B/X]$, then for any types $C_{1},\dots, C_{n}$, $f$ is a correct graph $f:A\to A[D/X]$, where $D=C_{1}\multimap \dots \multimap C_{n}\multimap X$.

Let $A,B\in \CC L_{\multimap}$, $B=B_{1}\multimap \dots \multimap B_{n}\multimap X$ and $f:A\to A[B/X]$, not necessarily correct. $f$ is called a \emph{simple $B$-expansion of $A$} if $id_{A}\subseteq f$.
Dually, a graph $f:A[B/X]\to A$ is called a \emph{simple $B$-coexpansion of $A$} if $id_{A}\subseteq f$. 

Examples of simple $B$-expansions and $B$-coexpansions are shown in figure \ref{expa}, for the paired types $X\multimap X$ and $(Y\multimap X)\multimap (X\multimap Z)$, respectively.
Simple expansions might fail to be correct. For instance the simple $Y\multimap X$-expansion in figure \ref{expa2} is not correct.

%
%
%
%
%


\begin{figure}
\begin{subfigure}{0.48\textwidth}
\begin{center}
\resizebox{0.5\textwidth}{!}{
\begin{tikzpicture}
\node(a) at (0,0) {$X\multimap X$};

\draw[<-, thick] (0.4,-0.2) to [bend right=18] (1.4,-1.8);

\draw[<-, thick] (-0.6,-1.8) to [bend right=18] (-0.4,-0.2);

\draw[<-, thick] (0.6,-1.8) to [bend right=45] (-1.4,-1.8);

\node(d) at (0,-2) {$(Y\multimap X)\multimap (Y\multimap X)$};

\end{tikzpicture}
}
\end{center}
\caption{Simple expansion of $X\multimap X$}
\label{expa1}
\end{subfigure}
\begin{subfigure}{0.48\textwidth}
\begin{center}
\resizebox{0.7\textwidth}{!}{
\begin{tikzpicture}
\node(a) at (0,0) {$(X\multimap Y)\multimap (X\multimap Z)$};

\draw[<-, thick] (1.4,-0.2) to [bend right=18] (2.6,-1.8);

\draw[->, thick] (-1.4,-1.8) to [bend right=12] (-1.4,-0.2);
\draw[->, thick] (0.6,-0.2) to [bend right=9] (1.6,-1.8);
\draw[->, thick] (-0.6,-0.2) to [bend left=9] (-0.6,-1.8);

\draw[<-, thick] (-2.3,-1.8) to [bend left=45] (0.6,-1.8);

\node(d) at (0,-2) {$( (Y\multimap X)\multimap Y)\multimap (Y\multimap X)\multimap Z$};

\end{tikzpicture}
}
\end{center}
\caption{Simple expansion of $(X\multimap Y)\multimap (X\multimap Z)$}
\label{expa2}
\end{subfigure}

\

\
\

\begin{subfigure}{0.48\textwidth}
\begin{center}
\resizebox{0.5\textwidth}{!}{
\begin{tikzpicture}
\node(a) at (0,-2) {$X\multimap X$};

\draw[->, thick] (0.4,-1.8) to [bend left=18] (1.4,-0.2);

\draw[->, thick] (-0.6,-0.2) to [bend left=18] (-0.4,-1.8);

\draw[->, thick] (0.6,-0.2) to [bend left=45] (-1.4,-0.2);

\node(d) at (0,0) {$(Y\multimap X)\multimap (Y\multimap X)$};

\end{tikzpicture}
}
\end{center}
\caption{Simple coexpansion of $X\multimap X$}
\label{expa3}
\end{subfigure}
\begin{subfigure}{0.48\textwidth}
\begin{center}
\resizebox{0.7\textwidth}{!}{
\begin{tikzpicture}
\node(a) at (0,-2) {$(X\multimap Y)\multimap (X\multimap Z)$};

\draw[->, thick] (1.4,-1.8) to [bend left=18] (2.6,-0.2);

\draw[<-, thick] (-1.4,-0.2) to [bend left=12] (-1.4,-1.8);
\draw[<-, thick] (0.6,-1.8) to [bend left=9] (1.6,-0.2);
\draw[<-, thick] (-0.6,-1.8) to [bend right=9] (-0.6,-0.2);

\draw[->, thick] (-2.3,-0.2) to [bend right=45] (0.6,-0.2);

\node(d) at (0,0) {$( (Y\multimap X)\multimap Y)\multimap (Y\multimap X)\multimap Z$};

\end{tikzpicture}
}
\end{center}
\caption{Simple coexpansion of $(X\multimap Y)\multimap (X\multimap Z)$}
\label{expa4}
\end{subfigure}
\caption{}
\label{expa}
\end{figure}

The following proposition characterizes weakly linearly expansible types:
\begin{proposition}\label{paired}
$A$ is weakly linearly expansible (resp. weakly co-linearly expansible) iff $A$ is paired.
\end{proposition}
\begin{proof}
If a type $A$ is paired, then let $p(A)$ be the set of its $X$-pairings, i.e. the set of all partitions of the occurrences of $X$ in $A$ into edges of occurrences of opposite polarity. Hence, for any $p\in p(A)$ and $B=B_{1}\multimap \dots \multimap B_{n}\multimap X$, we can define a simple $B$-expansion $f_{p}:A\to A[B/X]$ by joining $id_{A}$ with 
edges connecting occurrences of $B_{i}$ of opposite polarity corresponding to edges in $p$. One can similarly define a simple $B$-coexpansion $f_{p}:A[B/X]\to A$.

If $A$ is not paired then, by letting $Y$ be a variable not appearing in $A$ and $B=Y\multimap X$, there can be no graph $f: A\to A[B/X]$, since the number of occurrences of $Y$ in $A\multimap A[B/X]$ is odd.

\end{proof}

To investigate linearly expansible types we must take into consideration the correction criterion for simple expansion graphs. In the following subsections we will show that the paired types whose simple expansion graphs are always corrects are exactly the linear $gRP_{X}$ types.

\subsection{Linear generalized Russell-Prawitz types}

In the rest of this section, by a $RP_{X}$ (resp. $qRP_{X}$, $gRP_{X}$) we will indicate a linear Russell-Prawitz type in $X$ (resp. quasi Russell-Prawitz in $X$, generalized Russell-Prawitz in $X$).

If the linear types $A_{1},\dots, A_{n}$ are $qRP_X$ (resp. $qRP_{\CC X}$, $gRP_{X}$) and $Z\neq X$ (resp. $Z\notin \CC X$, $Z\neq X$), then the type $A_{1}\multimap \dots \multimap A_{n}\multimap Z$ is called \emph{$co-qRP_{X}$} (resp. \emph{$co-qRP_{\CC X}$}, \emph{$co-gRP_{X}$}).

We given now a different, inductive, definition of $gRP_X$ and $co-gRP_{X}$ types. The equivalence of definition \ref{grpx} and the one below is proved at the end of this section (proposition \ref{equivalences}). The inductive definition will allow us to obtain two different graphical characterizations of $gRP_X$ types.

\begin{definition}[linear $gRP_X$ type, inductive definition]\label{$gRP_X$}
We define by mutual induction the classes $ gRP_X^n, co-gRP_X^n\subseteq \CC L_{\multimap}$, for $n\in \BB N$, as follows:
\begin{enumerate}

\item if $A$ has no occurrence of $X$, then $A\in gRP_{X}^{0}, co-gRP_X^0$;

\item $X\in co-gRP^1_{X}$; 
\item if $B\in  gRP_X^p$ and $C\in co-gRP_X^{q}$, then $B\multimap C\in co-gRP_X^{p+q}$;
\item if $B\in co-gRP_X^p$ and $C\in  gRP_X^{q}$, then $B\multimap C\in  gRP_X^{p+q}$;

\item if $A_{1}\in co-gRP_X^{n_1},\dots, A_{p}\in co-gRP_X^{n_p}$ and for at least one $1\leq i\leq p$, $i\neq 0$, then 
$
A_{1}\multimap \dots  \multimap A_{p}\multimap X
$
is in $ gRP^X_{(\sum_{i}^{p}{}n_{i})-1}$. 

\end{enumerate}

We let $gRP_{X}:= gRP_X^0$ (resp. $co-gRP_{X}=co-gRP_X^0$).
\end{definition}
It can be verified by induction that
\begin{itemize}
\item if $A\in gRP_X^0$ or $A\in co-gRP_X^0$, then $n^{+}(A)=n^{-}(A)$;
\item if $A\in co-gRP_X^n$, then $n^{+}(A)=n^{-}(A)+n$;
\item if $A\in gRP_X^n$, then $n^{-}(A)=n^{+}(A)+n$.

\end{itemize}

Given a linear type $A\in \CC L_{\multimap}$, we indicate each occurrence of $X$ in $A$ by a distinct label $\alpha\in \BB N$. We introduce the notion of $(n,\epsilon)$-pairing:
\begin{definition}[pairing]
Let $A\in \CC L_{\multimap}$ be such that $n^{+}(A)=n^{-}(A)+q$ (resp. $n^{-}(A)=n^{+}(A)+q$). A 
\emph{$(q,+)$-pairing} (resp. \emph{$(q,-)$-pairing}) of $A$ is a pair
$(\F p, \F a)$ where $\F a=\{\alpha_{1},\dots,\alpha_{q}\}$ is a set containing labels of distinct positive (resp. negative) occurrences of $X$ in $A$ and $\F p$ is a pairing of the remaining $2\cdot n^{+}(A)$ occurrences of $X$ in $A$.
\end{definition}

For instance, $((X_{\alpha}^{+},X_{\beta}^{-}), \{X_{\gamma}^{+}\})$ is a $(1,+)$ pairing of the linear type $(X_{\alpha}\multimap X_{\beta})\multimap X_{\gamma}$ and $((X_{\gamma}^{+},X_{\beta}^{-}), \{X_{\alpha}^{-}\})$ is a $(1,-)$ pairing of the linear type $X_{\alpha}\multimap X_{\beta}\multimap X_{\gamma}$.

%
%
%
%

The definition below associates with any $gRP_X$ or $co-gRP_{X}$ type $A$ a set of {pairings} $P(A)$. 

\begin{definition}
To any $A\in  gRP_X^n$ (resp. $A\in co-gRP_X^n$) we associate a set of $(n,-)$-pairings (resp. $(n,+)$-pairings) $P_{n}(A)$ as follows;
\begin{enumerate}

\item[1.] if $A$ has no occurrence of $X$, then $P_{0}(A)=\{(\emptyset, \emptyset)\}$;

\item[2.] if $A=X_{\alpha}$, then $P_{1}(A)= \{ (\emptyset, \{\alpha\})\}$;

\item[3.] if $A=B\multimap C$, where $B\in  gRP_X^{n_1}, C\in co-gRP_X^{n_2}$, then $P_{d}(A)$, where $d=n_{1}+n_{2}$, contains all $(d,+)$-pairings of the form
$(\F p\cup \F q, \F a\cup \F b)$, where $(\F p,\F a)\in P_{n_{1}}(B)$, $(\F q, \F b)\in P_{n_{2}}(C)$;


\item[4.] if $A=B\multimap C$, where $B\in co-gRP_X^{n_1}, C\in  gRP_X^{n_2}$, then $P_{d}(A)$, where $d=n_{1}+n_{2}$, contains all $(d,-)$-pairings of the form
$(\F p\cup\F  q, \F a\cup \F b)$, where $(\F p,\F a)\in P_{n_{1}}(B)$, $(\F q,\F b)\in P_{n_{2}}(C)$;

%
%

\item[5.] if $A=A_{1}\multimap \dots  \multimap A_{k}\multimap X_{\alpha}$, with $A_{i}\in co-gRP_X^{n_i}$ and for some $1\leq k\leq i$, $k\neq 0$, then $P_{d}(A)$, where $d=(\sum_{i}^{p}{n_{i}})-1$ contains all $(d,-)$-pairings of the form
$(\F p_{1}\cup \dots \cup \F p_{k}\cup \F p \cup \{(\beta ,\alpha)\}, \F c)$, where $(\F p_{i}, \F a_{i})\in P_{n_{i}}(A_{i})$, $\beta\in \bigcup_{i}\F a_{i}$ and $\F c= \bigcup \F a_{i}-\{\beta\}$.

%

\end{enumerate}

If $A\in gRP_{X}$ or $A\in co-gRP_{X}$, then we let $P(A)= \{\F p\mid (\F p, \emptyset)\in P_{0}(A)\}$.
\end{definition}

We will now show that $gRP_X$ and $co-gRP_{X}$ types can be characterized by the properties of their pairings.

%
%

We introduce some terminology about correction graphs:
with abuse of notation, we will make no distinction between the type $A$ and the shape $S_{A}$. In particular, by a \emph{$\multimap^{+}$-node} (resp. \emph{$\multimap^{-}$-node}) in $A$ we indicate a positive (resp. negative) occurrence of a type $B\multimap C$ in $A$. By the \emph{positive branch} (illustrated in figure \ref{pos}) of a $\multimap^{+}$-node $l$ we indicate the list of all $\multimap^{+}$-node which are reachable in $S_{A}$ going upwards from $l$ plus the positive occurrence of a variable appearing at the end of this branch. 
Similarly, by a \emph{negative branch} (illustrated in figure \ref{neg}) of a $\multimap^{-}$-node we indicate the list of all $\multimap^{-}$-node from which $l$ is reachable in $S_{A}$ plus the negative occurrence of a variable from which all such nodes are reachable going downwards.

\begin{figure}
\begin{subfigure}{0.48\textwidth}
\begin{center}
\resizebox{0.4\textwidth}{!}{
\begin{tikzpicture}
\node(a) at (0,0) {$\multimap^{+}$};

\node(b) at (2,2) {$\multimap^{+}$};
\node(c) at (3,3) {$Y^{+}$};

\draw[thick,->] (a) to (b);
\draw[thick,->] (b) to (c);

\end{tikzpicture}}
\end{center}
\caption{Positive branch}
\label{pos}
\end{subfigure}
\begin{subfigure}{0.48\textwidth}
\begin{center}
\resizebox{0.4\textwidth}{!}{
\begin{tikzpicture}
\node(a) at (3,3) {$Y^{-}$};

\node(b) at (2,2) {$\multimap^{-}$};
\node(c) at (0,0) {$\multimap^{-}$};

\draw[thick,->] (a) to (b);
\draw[thick,->] (b) to (c);
\draw[thick,->] (b) to (1,3);
\draw[thick,->] (c) to (-1,1);

\end{tikzpicture}}
\end{center}
\caption{Negative branch}
\label{neg}
\end{subfigure}
\caption{Positive and negative branches}
\end{figure}
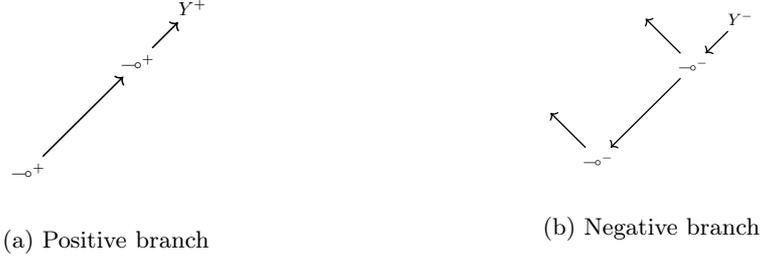

Also, by the \emph{tree of a node $\multimap^{\epsilon}$} we indicate the subtree having the node as root (and also the corresponding occurrence of a subtype of $A$). By the \emph{maximal tree of a $\multimap^{+}$ (resp. $\multimap^{-}$) node $l$ } we indicate the subtree having as root the first (resp. last) $\multimap^{+}$ (resp. $\multimap^{-}$) node $l'$ of the positive (resp. negative) branch of $l$.


Let $A$ be a type and $e=(X_{\alpha}^{+},X_{\beta}^{-})$ indicate an edge made, respectively, of a positive and a negative occurrence of $X$ in $A$. We will say that $e$ is a \emph{jump out} edge in $A$ if either  $X_{\alpha}$ is the conclusion of the positive branch starting in a $\multimap^{+}$-node and $X^{-}$ occurs outside of the maximal tree of this node, or $X_{\alpha}$ is lefthand premiss of a $\multimap^{-}$-node. 

The dual notion is that of a \emph{jump in} edge in $A$, that is an edge $e=(X^{+}_{\alpha},X^{-}_{\beta})$ such that either  $X^{-}_{\beta}$ is the start of the negative branch ending in a $\multimap^{-}$-node and $X^{+}$ occurs outside of the maximal tree of this node, or $X^{-}_{\beta}$ is lefthand premiss of a $\multimap^{+}$-node. 


An edge $e=(X^{+}_{\alpha},X^{-}_{\beta})$ which is not a jump out in $A$, i.e. such that $X^{+} _{\alpha}$ is the conclusion of the positive branch of a $\multimap^{+}$ node $l$ and $X^{-}_{\beta}$ occurs inside the maximal tree of $l$, is called an \emph{internal edge}.

For instance, in the type $((X^{-}_{\alpha}\multimap X^{+}_{\alpha'})\multimap X^{-}_{\beta})\multimap X^{+}_{\beta'}$ the edge $(X^{+}_{\beta'}, X^{-}_{\alpha})$ is internal, the edge $(X^{+}_{\alpha'},X^{-}_{\beta})$ is a jump out edge and the edge $(X^{+}_{\alpha'},X^{-}_{\alpha})$ is both internal and a jump in edge.

%
%
%
%
%
%

The proposition below characterizes $gRP_n^X$ and $co-gRP_n^X$ types by properties of their $(n,\epsilon)$-pairings. In particular, it characterizes $gRP_X$ types as those which have a pairing made of internal edges. 

\begin{proposition}\label{pairing}
For any $A\in \CC L_{\multimap}$ and $n\in \BB N$, 
\begin{itemize}
\item[$i.$] $A\in  gRP_X^n$ iff $A$ has a $(n,-)$-pairing with no jump out;
\item[$ii.$] $A\in co-gRP_X^n$ iff $A$ has a $(n,+)$-pairing with no jump in.
\end{itemize}
\end{proposition}
\begin{proof}
For one direction it suffices to verify that, if $A\in  gRP_X^n$ (resp. $A\in co-gRP_X^n$), then any $p\in P_{n}(A)$ has no jump out (resp. jump in). This can be done by induction on definition \ref{$gRP_X$}.

For the converse direction we argue by induction on $A$; if $A$ has no occurrence of $X$, then the claims are trivially true, so we will suppose that $A$ contains at least one occurrence of $X$. If $A= X$ claim $i.$ is trivially true (as $A$ has no $(n,-)$-pairing) and claim $ii.$ is true, as $A\in co-gRP_X^n$.

Let then $A=A_{1}\multimap \dots \multimap A_{k}\multimap Z$:
\begin{enumerate}
\item[a.] if $A$ admits a $(n,-)$-pairing $(\F p, \F a)$ with no jump out, then two cases arise:
	\begin{enumerate}
	\item[a1.] if $Z\neq X$, then $\F p$ splits into $\F p_{1}\cup \dots \cup \F p_{k}$, where the variables in $\F p_{i}$ occur in $A_{i}$, since any edge $(X^{+}_{\alpha},X^{-}_{\beta})$ with $X^{+}_{\alpha}\in A_{i}$ and $X^{-}_{\beta}\in A_{j}$ and $j\neq i$ is a jump out. 
	As the $\alpha_{1},\dots, \alpha_{n}\in \F a$ are distributed among the $A_{i}$, we deduce that 
	there exist integers $n_{1},\dots, n_{k}$ such that $\sum_{i}^{k}{n_{i}}=n$, a partition of $\F a_{1},\dots, \F a_{k}$ of $\F a$, where $\F a_{i}$ has cardinality $n_{i}$, and that 
	$(\F p_{i}, \F a_{i})$ is a $(n_{i},+)$-pairing of $A_{i}$ 
	 with no jump in. By induction hypothesis we get $A_{i}\in co-gRP_X^{n_i}$ and we conclude $A\in  gRP_X^n$ by clause 4.
	
	\item[a2.] if $Z=X_{\alpha}^{+}$, then $X_{\alpha}^{+}$ is paired in $p$ with a positive occurrence $X_{\beta}$ in some $A_{i}$, say $A_{c}$. 
	By reasoning as above, we obtain that $\F p$ splits into $\F p_{1}\cup \dots \cup \F p_{c-1}\cup \F p_{c}\cup \F p_{c+1}\cup \dots \cup \F p_{n}\cup\{ (X_{\alpha}, X_{\beta})\}$ where the variables in $\F p_{i}$ occur in $A_{c}$. As $\beta$ and the $\alpha_{1},\dots, \alpha_{n}$ in $\F a$ are distributed among the $A_{i}$, we deduce that 
	there exist integers $n_{1},\dots, n_{k}$ such that $\sum_{i}^{k}{n_{i}}=n+1$, a partition $\F a_{1},\dots, \F a_{n}$ of $\F a$, where $\F a_{i}$ has cardinality $n_{i}$, and that $(\F p_{i}, \F a_{i})$, for $i\neq c$ is a $(n_{i},+)$-pairing of $A_{i}^{-}$ with no jump in, while $(\F p_{c}, \F a_{c}\cup \{\beta\})$ is a $(n_{c},+)$-pairing of $A_{c}$ with no jump in.
	By applying the induction hypothesis we get $A_{i}\in co-gRP_X^{n_i}$ and we conclude $A\in  gRP_X^n$ by clause 5.

	\end{enumerate}

\item[b.] if $A$ admits a $(n,+)$-pairing $(\F p,\F a)$ with no jump in, then if $Z=X_{\alpha}^{+}$, $X_{\alpha}^{+}$ cannot occur in any edge in $\F p$, since any edge containing $X_{\alpha}^{+}$ and any $X_{\beta}^{-}$ in some $A_{i}$ is a jump in. Hence $\F p$ splits into $\F p_{1}\cup \dots \cup \F p_{k}$, where the variables in $\F p_{i}$ occur in $A_{i}$, since any edge $(X^{+}_{\alpha},X^{-}_{\beta})$ with $X^{+}_{\alpha}\in A_{i}$ and $X^{-}_{\beta}\in A_{j}$ and $j\neq i$ is a jump out. 
	As the $\alpha_{1},\dots, \alpha_{n}$ in $\F a$ are distributed among the $A_{i}$ and possibly $Z$, we deduce that there exist integers $n_{1},\dots, n_{k}$ such that either $\sum_{i}^{k}{n_{i}}=n$ (if $Z\neq X_{\alpha}$) or $\sum_{i}^{k}(n_{i})=n-1$ (if $Z=X$), a partition $\F a_{1},\dots, \F a_{n}$ of $\F a-\{X_{\alpha}^{+}\}$, where $\F a_{i}$ has cardinality $n_{i}$ and that $(\F p_{i},\F a_{i})$ is a $(n_{i},-)$-pairing of $A_{i}$ with no jump out. By the induction hypothesis we get $A_{i}\in  gRP_X^{n_i}$ and we conclude $A\in co-gRP_X^n$ by clause 3.

\end{enumerate}
\end{proof}

Proposition \ref{pairing} provides a decidable criterion to test whether a linear type is $gRP_{X}$: it suffices to check among its $X$-pairings whether there is one made of internal edges. For instance, the type $((X^{-}_{\alpha}\multimap X^{+}_{\alpha'})\multimap X^{-}_{\beta})\multimap X^{+}_{\beta'}$ has the $X$-pairing $\{(X^{+}_{\alpha'},X^{-}_{\alpha}), (X^{+}_{\beta'}, X^{-}_{\beta})\}$ made of internal edges, hence it is $gRP_{X}$. The type
$(X^{+}_{\alpha}\multimap X^{-}_{\alpha'})\multimap (X^{-}_{\beta}\multimap X^{+}_{\beta'})$ has no $X$-pairing made of internal edges (as any pair $(X^{+}_{\alpha},X^{-})$ is a jump out) so it is not $gRP_{X}$.

Moreover, proposition \ref{pairing} allows to prove the equivalence of the definitions \ref{grpx} and \ref{$gRP_X$} of $gRP_X$ types:

\begin{proposition}\label{equivalences}
Definitions \ref{grpx} and \ref{$gRP_X$} of linear $gRP_{X}$ types are equivalent.
\end{proposition}
\begin{proof}
Suppose $A=A'[X/X_{1},\dots, X/X_{p}]$, where $A'$ is $qRP$ in $\CC X=\{X_{1},\dots, X_{p}\}$. Then, for any $X_{i}\in \CC X$, $A'$ has exactly two occurrences of $X_{i}$, one positive and one negative, which form an internal pair.  All such edges, for $1\leq i\leq p$, induce a $X$-pairing of $A$ made of internal edges. We conclude, by proposition \ref{pairing}, that $A\in  gRP_X$.

For the converse direction, it can be verified by induction on definition \ref{$gRP_X$} that, if $A\in  gRP_X^n$ (resp. $A\in co-gRP_X^n$), then given a $(n, -)$-pairing (resp. a $(n,+)$-pairing) $(\F p, \F a)$ of $A$, by renaming the edges in $\F p$ with distinct variables $X_{1},\dots, X_{p}$, we obtain a type $A'$ which is $qRP$ (resp. co-$qRP$) in $\CC X=\{X_{1},\dots, X_{p}\}$. Hence, if $A\in  gRP_X$, we obtain a type $A'$ with no occurrence of $X$ which is $qRP$ in some finite set $\CC X$.
\end{proof}

If $f:\Gamma\to A$ is a graph in $\CC A^{\multimap}$, then its correction graph being a directed acyclic graph, it induces a partial order relation $\prec_{f}$ over all variable occurrences in $\Gamma$ and $A$: $Y_{\alpha}\prec_{f} Z_{\beta}$ holds when the unique path in the correction graph of $f$ from the conclusion to $Z_{\beta}$ passes through $Y_{\alpha}$. 

The following lemma relates the edges in $\Gamma\to A$ and the order $\prec_{f}$ of any correct graph $f:\Gamma\to A$. 

\begin{lemma}\label{jumpout}
Let $f:A_{1},\dots, A_{n}\to A$ in $\CC A^{\multimap}$. Then, by letting $A'=A_{1}\multimap \dots \multimap A_{n}\multimap A$, for any edge $e=(X^{+}_{\alpha}, X^{-}_{\beta})$ in $A'$, the following hold:
\begin{itemize}
\item[$i.$] if $e$ is an internal edge, then $X^{+}_{\alpha}\prec_{f}X^{-}_{\beta}$;

\item[$ii.$] if $e$ is a jump out edge, then either $e\in f$ or $X^{+}_{\alpha}\not\prec_{f} X^{-}_{\beta}$.

\end{itemize}

\end{lemma}

\begin{proof}
To prove claim $i.$, for any positive occurrence of variable $Y_{\alpha}^{+}$ in $A'$, let $A_{Y}^{\alpha}$ be the type occurring positively in $A'$ whose rightmost variable is $Y_{\alpha}^{+}$. 
Let us define a distance $d(Y_{\alpha}^{+},Z_{\beta}^{-})$ between a variable $Y^{+}_{\alpha}$ positively occurring in $A'$ and a variable $Z_{\beta}^{-}$ negatively occurring in $A^{\alpha}_{Y}$: $d(Y_{\alpha}^{+},Z_{\beta}^{-})=0$ if $A_{Y}^{\alpha}=C_{1}\multimap \dots \multimap C_{i-1}\multimap Z_{\beta}\multimap C_{i+1}\multimap \dots \multimap  C_{p}\multimap Y_{\alpha}$; $d(Y_{\alpha}^{+},Z_{\beta}^{-})=k+1$ if 
$A_{Y}^{\alpha}=C_{1}\multimap \dots \multimap C_{i-1}\multimap D\multimap C_{i+1}\multimap \dots \multimap  C_{p}\multimap Y_{\alpha}$ , $D= D_{1}\multimap \dots \multimap D_{j-1}\multimap E\multimap D_{j+1}\multimap \dots \multimap D_{q}\multimap Y'_{\alpha'}$ and $E=E_{1}\multimap \dots \multimap E_{r}\multimap Z'_{\gamma}$, where 
$d((Z')_{\gamma}^{+}, Z_{\beta}^{-})=k$. 

We can argue then by induction on $d=d(Y_{\alpha}^{+},Z_{\beta}^{-})$ that if $Z_{\beta}^{-}$ occurs in the maximal tree of the $\multimap^{+}$-node at the root of $A_{Y}^{\alpha}$, then $Y_{\alpha}^{+}\prec_{f}Z_{\beta}^{-}$. If $d=0$, then this follows by functionality. If $d=k+1$, then by functionality, inspection of the correction graph (which contains a path from $(Y')_{\alpha'}^{-}$ to $(Z')_{\gamma}^{+}$) and the induction hypothesis we have $Y_{\alpha}^{+}\prec_{f} (Y')_{\alpha'}^{-}\prec_{f}(Z')_{\gamma}^{+}\prec_{f} Z_{\beta}^{-}$.

For point $ii.$, since any correct graph is sequentializable, we will argue by induction on a normal $\lambda$-term $u$ such that  $\Gamma\vdash u:A$, by relying on theorem \ref{seque}. We will suppose w.l.o.g. that the sequent $\Gamma\vdash A$ contains at least two occurrences of $X$ (the claim being trivially valid otherwise). 


\begin{enumerate}

\item if $u=x$, then $x:A\vdash x:A$, $\CC G(u)=id_{A}$ and we argue by a sub-induction on $A$:
	\begin{enumerate}
	\item[$1a.$] if $A=X$, then there is exactly one edge $e=(X^{-}_{\alpha},X^{+}_{\beta})\in f$;	
	\item[$1b.$] if $A=B\multimap C$, let $e=(X^{+}_{\alpha},X^{-}_{\beta})$ be a jump out edge in $A'=(B^{+}\multimap C^{-})\multimap (B^{-}\multimap C^{+})$. $e$ can be of four kinds: (1) a jump out edge in either $C^{+}$ or $B^{+}$, (2) a jump in  edge in either $B^{-}$ or $C^{-}$ or (3) a jump out edge with one occurrence in either $B^{+}$ or $B^{-}$ and one in either $C^{+}$ or $C^{-}$ or (4) a jump out edge between either $B^{+}$ and $B^{-}$ or between $C^{+}$ and $C^{-}$. Moreover, as all edges in $id_{A}$ are type I, it cannot be $e\in f$. The correction graph of $id_{A}$ is as illustrated in figure \ref{shape0}.
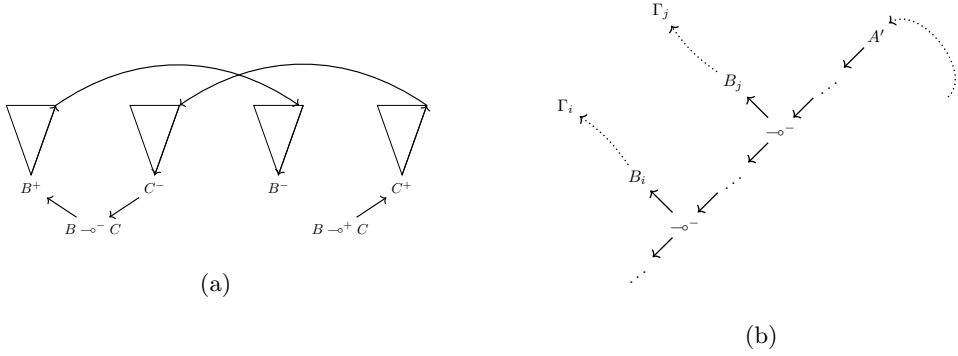
\begin{figure}
\begin{subfigure}{0.48\textwidth}
\begin{center}
\resizebox{0.8\textwidth}{!}{
\begin{tikzpicture}
\node(a) at (3,0) {$B\multimap^{+} C$};

\node(b) at (1.5,1) {$B^{-}$};
\node(c) at (4.5,1) {$C^{+}$};

\node(a') at (-3,0) {$B\multimap^{-} C$};

\node(b') at (-1.5,1) {$C^{-}$};
\node(c') at (-4.5,1) {$B^{+}$};

\draw (4.5,1.3) to (5.1,3) to (3.9,3) -- cycle; 

\draw (1.5,1.3) to (2.1,3) to (0.9,3) -- cycle;
\draw (-1.5,1.3) to (-2.1,3) to (-0.9,3) -- cycle; 
 \draw (-4.5,1.3) to (-5.1,3) to (-3.9,3) -- cycle;

\draw[->, thick] (a) to (c);

\draw[->, thick] (b') to (a');

\draw[->, thick] (a') to (c');

\draw[->, thick] (4.5,1.3) to (5.1,3);
\draw[->, thick] (-0.9,3) to (-1.5,1.3);

\draw[->, thick] (5.1,3) to [bend right=35] (-0.9,3);
\draw[->, thick] (-3.9,3) to [bend left=35] (2,3);
\draw[->, thick] (-4.5,1.3) to (-3.9,3);
\draw[->, thick] (2.1,3) to (1.5,1.3);

\end{tikzpicture}}
\end{center}
\caption{}
\label{shape0}
\end{subfigure}
\begin{subfigure}{0.48\textwidth}
\begin{center}
\resizebox{0.8\textwidth}{!}{
\begin{tikzpicture}

\node(a) at (0,0) {$\iddots$};
\node(b) at (-1,-1) {$\multimap^{-}$};

\node(bb) at (-2,0) {$B_{j}$};
\draw[->, thick] (b) to (bb);

\node(g) at (-3.5,1.5) {$\Gamma_{j}$};
\draw[->, thick, dotted] (bb) to [bend left=10] (g);

\draw[->, thick] (a) to (b);
\node(aa) at (1,1) {$A'$};
\draw[->, thick] (aa) to (a);

\node(dd) at (-2,-2) {$\iddots$};
\draw[->, thick] (b) to (dd);
\node(cc) at (-3,-3) {$\multimap^{-}$};
\draw[->, thick] (dd) to (cc);

\node(bb2) at (-4,-2) {$B_{i}$};
\draw[->, thick] (cc) to (bb2);

\node(gg) at (-5.5,-0.5) {$\Gamma_{i}$};
\draw[->, thick, dotted] (bb2) to [bend right=10] (gg);

\node(ddd) at (-4,-4) {$\iddots$};
\draw[->, thick] (cc) to (ddd);

\draw[->, thick, dotted] (2.5,-0.3) to [bend right=85] (aa);

\end{tikzpicture}}
\end{center}
\caption{}
\label{shape1}
\end{subfigure}
\caption{Correction graphs}
\end{figure}

We claim that there exists a path from the conclusion $A\multimap A$ to $X^{-}_{\beta}$ not passing through $X^{+}_{\alpha}$. In all four cases we can conclude by inspection of the correction graph and by the induction hypothesis.

	\end{enumerate}

\item if $u=\lambda x^{B}.u'$, then $ \AXC{$\Gamma, x:B\vdash u':C$}\UIC{$\Gamma\vdash \lambda x^{B}.u':B\multimap C$}\DP$, and then the claim immediately follows by the induction hypothesis, as the partial order
induced by the correction graph of $\CC G(u)$ is the same as the one induced by the correction graph of $\CC G(u')$.

\item if $u=xu_{1}\dots u_{p}$, then $x:B_{1}\multimap\dots \multimap B_{p}\multimap A$, $\Gamma$ is partitioned into $\Gamma_{1},\dots, \Gamma_{p},x:B_{1}\multimap \dots \multimap B_{p}\multimap A$ and, for $1\leq i\leq p$, $\Gamma_{i}\vdash u_{i}: B_{i}$. We let $B'_{i}$ be the type $C_{i_{1}}\multimap \dots \multimap C_{p_{i}}\multimap B_{i}$, where $\Gamma_{i}=x_{i_{1}}:C_{i_{1}},\dots, x_{i_{p_{i}}}:C_{i_{p_{i}}}$. To avoid confusion we call $A'$ the occurrence of $A$ in the declaration of $x$. Let then $e$ be a jump out edge in $A'$ and suppose $e\notin f$. Two subcases must be considered. First, if $e$ is a jump out edge in $A$ or an edge $(X^{\epsilon}_{\alpha},X_{\beta}^{\epsilon'})$, with $X_{\alpha}$ occurring in $A$ and $X_{\beta}$ occurring in $A'$, then the claim follows by remarking that the correction graph restricted to $A$ and $A'$ is the same as the correction graph of the identity graph $id_{A}:A\to A$, so we can argue as in case $1$. Second, if $e$ is a jump out edge $(X_{\alpha}^{+},X_{\beta}^{-})$, with $X_{\alpha}^{+}\in B'_{i}$ and $X_{\beta}^{-}\in  B'_{j}$, then, if $i=j$, we can apply the induction hypothesis to $u_{i}$; if $i\neq j$, then the claim follows by 
inspection of the correction graph (figure \ref{shape1}), by remarking that, for any variable $X_{\alpha}$ in $B'_{i}$ and $X_{\beta}$ in $B'_{j}$, $X_{\alpha}\not \prec_{f} X_{\beta}$.

%
%
%
%
%
%
%

\end{enumerate}

\end{proof}

Let $f:\Gamma\to A$ be a graph in $\CC A^{\multimap}$ and let $e=(X_{\alpha}^{+},X_{\beta}^{-})$ be an edge over $X$ in $A$. 
We say that $e$ is an \emph{expansible edge} if for any $B_{1},\dots, B_{n}\in \CC L_{\multimap}$, the graph $f':\Gamma\to A'$, where $A'$ is obtained from $A$ by replacing the two occurrences of $X$ in $e$ by $B=B_{1}\multimap\dots\multimap B_{n}\multimap X$ and $f'$ is obtained from $f$ by adding type III edges over the variables of $B_{1},\dots, B_{n}$ as in figure \ref{interfig2}, is correct.

By the functionality condition, if an edge $e=(X_{\alpha}^{+},X_{\beta}^{-})$ is expansible then it must be $X^{+}_{\alpha}\prec_{f} X_{\beta}^{-}$ (the dotted path in figure \ref{interfig2}) in the correction graph of $f$. It can easily be verified that the converse also holds. 
This leads to the following:

%
%

\begin{lemma}\label{inte2}
Let $f:\Gamma\to A$ be a graph in $\CC A^{\multimap}$ and let $e$ be an edge over $X$ in $A$.
Then $e$ is an expansible edge iff $e\in f$ or $e$ is internal.

\end{lemma}
\begin{proof}
By the remark above and lemma \ref{jumpout}.

%

\end{proof}

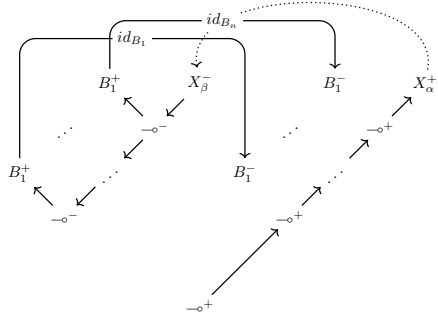
\begin{figure}
\begin{center}
%
%
%
%
\resizebox{0.4\textwidth}{!}{\begin{tikzpicture}[baseline=9ex, every node/.style={fill=white}]
\node(a) at (0,0) {$\multimap^{+}$};

\node(b) at (2,2) {$\multimap^{+}$};
\node(b1) at (1,3) {$B_{1}^-$};
\node(c) at (3,3) {$\iddots$};
\node(z) at (2,4) {$\iddots$};

\node(d) at (4,4) {$\multimap^{+}$};
\node(d1) at (3,5) {$B_{1}^{-}$};
\node(e) at (5,5) {$X^{+}_{\alpha}$};

\node(b') at (-3,2) {$\multimap^{-}$};
\node(b'1) at (-4,3) {$B_{1}^{+}$};
\node(c') at (-2,3) {$\iddots$};
\node(zz) at (-3,4) {$\iddots$};

\node(d') at (-1,4) {$\multimap^{-}$};
\node(d'1) at (-2,5) {$B_{1}^{+}$};
\node(e') at (0,5) {$X^{-}_{\beta}$};

\draw[->, thick] (d) to (e);
\draw[->, thick] (e') to (d');
\draw[->, thick] (b') to (b'1);
\draw[->, thick] (d') to (d'1);

\draw[->, thick] (b) to (c);
\draw[->, thick] (c) to (d);

\draw[->, thick] (c') to (b');
\draw[->, thick] (d') to (c');

\draw[->, thick] (a) to (b);
\draw[->, dotted, thick] (e) to [bend right=100] (e');

\draw[->, thick, rounded corners=10pt] (b'1) to (-4,6) to node {$id_{B_{1}}$} (1,6) to (b1);
\draw[->, thick, rounded corners=10pt] (d'1) to (-2,6.4) to node {$id_{B_{n}}$} (3,6.4) to (d1);

\end{tikzpicture}
}\end{center}
\caption{Expansion of the edge $(X_{\alpha}^{+},X_{\beta}^{-}$.}
\label{interfig2}
\end{figure}

We are now in a position to prove that the simple expansions of $gRP_X$ types are correct:
\begin{proposition}\label{drave}
If $A\in gRP_{X}$, then $A$ is linearly expansible.\end{proposition}
\begin{proof}
If $A\in gRP_{X}$, then $A$ admits a $X$-pairing made of internal edges. Hence, by lemma \ref{inte2}, all such edges induce a correct simple $B$-expansion of $id_{A}$, for any $B\in \CC L_{\multimap}$. 

\end{proof}

We conclude this presentation of $gRP_X$ types by relating them to $qRP_{\CC X}$ types: any closed proof of a $gRP_X$ type is actually (up to variable renaming) a proof of a $qRP_{\CC X}$ type in some $\CC X\subset \CC V$. In other words, one can ``separate the variables'' following a pairing made of internal edges.

\begin{proposition}\label{renaming}
Let $f:\Gamma\to A$ in $\CC A^{\multimap}$, where $\Gamma=\{A_{1},\dots, A_{n}\}$ is a context made of $co-gRP_{X}$ types and $A$ is $gRP_X$. Then there exists a finite set $\CC X=\{X_{1},\dots, X_{n}\}\subset \CC V$, a context $\Gamma'=A'_{1},\dots, A'_{n}$ made of $co-qRP_{\CC X}$ types such that $A'_{i}[X/X_{1},\dots, X/X_{p}]=A_{i}$ and a type $A$ $qRP_{\CC X}$ such that $A'[X/X_{1},\dots, X/X_{p}]=A$ and $f:\Gamma'\to A'$. 

\end{proposition}
\begin{proof}
We will argue by induction on the number of applications in a normal linear $\lambda$-term $u$ in $\eta$-long normal form such that $\Gamma\vdash u:A$. If $u$ has no applications, then either $u=x_{i}$, where $[x_{i}]=A_{i}$, or $u=\lambda y^{Z}.y$. If $u=x_{i}$, since $[x_{i}]=A_{i}=A\in  gRP^X\cap co-gRP_X$, $X\notin FV(A)$, hence the claim trivially holds; if $u=\lambda y^{Z}.y$, then $\Gamma=\emptyset$ and $A=Z\multimap Z\in  gRP_X$. Then, either $Z\neq X$, so the claim trivially holds, or $X=Z$, then $A=X\multimap X\in qRP_X$.

Suppose now $u$ has $k+1$ applications, i.e. $u=\lambda x_{1}^{A_{1}}.\dots.\lambda x_{p}^{A_{p}}.yu_{1}\dots u_{q}$. Then $A=A_{1}\multimap \dots \multimap A_{p}\multimap Z$ and two cases arise:
\begin{description}

\item[($Z\neq X$)] it must be $A_{i}\in co-gRP_X$ for all $1\leq i\leq p$ and from $[y]\in co-gRP_X$ it follows that 
$[u_{j}]\in  gRP_X$ for all $1\leq j \leq q$. Suppose $y\neq x_{i}$, for all $1\leq i\leq p$. Then there exists a partition $\Gamma_{1},\dots, \Gamma_{q}$ of $\Gamma- \{[u_{j}]\}$ and a partition $(s_{1},\dots, s_{q})$ of $\{1,\dots,p\}$ such that, for each $1\leq j\leq q$, $\Gamma_{j}, \Delta_{j}\vdash u_{j}:[u_{j}]$, where $\Delta_{j}=\{x_{i_{1}}:A_{i_{1}},\dots, x_{i_{p_{j}}}:A_{i_{p_{j}}}\}$, $s_{j}=\{i_{1},\dots, i_{p_{j}}\}$; by the induction hypothesis, then, there exists sets $\CC X_{j}\subset \CC V$ contexts $\Gamma'_{j}$, $\Delta'_{j}$ of types $co-qRP_{\CC X_{j}}$ and a type $C_{j}$ $qRP_{\CC X_{j}}$ such that $\Gamma'_{j}, \Delta'_{j}\vdash u_{j}:C_{j}$. W.l.o.g. we can suppose all $\CC X_{j}$ disjoint. Let $\CC X=(\bigcup_{j} \CC X_{j})$. We have then
$\Gamma'_{1},\dots, \Gamma'_{q}, \Delta'_{1},\dots, \Delta'_{q}, y:C_{1}\multimap \dots \multimap C_{q}\multimap Z\vdash yu_{1}\dots u_{q}: Z$, so we can conclude. One can argue similarly if $y=x_{i}$ for some $1\leq i\leq p$.

\item[($Z=X$)] it must be $A_{i}\in co-gRP_X$ for all $1\leq i\leq p$ but for one $A_{c}\in co-gRP^X_{1}$. Observe that this forces $y=x_{c}$. We have that $A_{c}=[u_{1}]\multimap \dots \multimap [u_{q}]\multimap X$. Now, let $A'_{c}=[u_1]\multimap \dots \multimap [u_{q}]\multimap X'$, for some $X'$ not occurring elsewhere. Then $A'_{c}\in co-gRP_X$, and we deduce that $[u_{j}]\in  gRP_X$, for all $1\leq j \leq q$. Now, there exists a partition $\Gamma_{1},\dots, \Gamma_{q}$ of $\Gamma$ and a partition $(s_{1},\dots, s_{q})$ of $\{1,\dots,p\}- \{c\}$ such that, for each $1\leq j\leq q$, $\Gamma_{j}, \Delta_{j}\vdash u_{j}:[u_{j}]$, where $\Delta_{j}=\{x_{i_{1}}:A_{i_{1}},\dots, x_{i_{p_{j}}}:A_{i_{p_{j}}}\}$, $s_{j}=\{i_{1},\dots, i_{p_{j}}\}$; by the induction hypothesis, then, there exists sets $\CC X_{j}\subset \CC V$ contexts $\Gamma'_{j}$, $\Delta'_{j}$ of types $qRP_{\CC X_{j}}$ and a type $C_{j}$ $qRP_{\CC X_{j}}$ such that $\Gamma'_{j}, \Delta'_{j}\vdash u_{j}:C_{j}$. W.l.o.g. we can suppose all $\CC X_{j}$ disjoint and not containing $X'$. Let $\CC X=(\bigcup_{j} \CC X_{j})\cup \{X'\}$. We have then
$\Gamma'_{1},\dots, \Gamma'_{q}, \Delta'_{1},\dots, \Delta'_{q}, x_{c}:C_{1}\multimap \dots \multimap C_{q}\multimap X'\vdash x_{c}u_{1}\dots u_{q}: X'$, so we can conclude.

\end{description}

\end{proof}

\subsection{Characterization of linearly expansible types}\label{sec53}

We show that $gRP_{X}$ types are ``dense'' in the class of linearly expansible types in the following sense: any linearly expansible type is either a $gRP_X$ or it linearly collapses into a set of $gRP_X$ types. We also show that $gRP_X$ types are exactly those types for which all simple expansion graphs are correct.

Let $f:\Gamma\to A[B/X]$ in $\CC A^{\multimap}$, where $B=Y\multimap X$ and $Y\notin FV(\Gamma), FV(A)$. The edges over $Y$ of $f$ form then a pairing $f_{Y}\subset f$ made of type III edges. Moreover, $f-f_{Y}:\Gamma\to A$ is allowable: its correction graph is obtained from the correction graph of $f$ by the transformation illustrated in figure \ref{trans} (in which the dotted path is forced by functionality), which preserves correctness.

 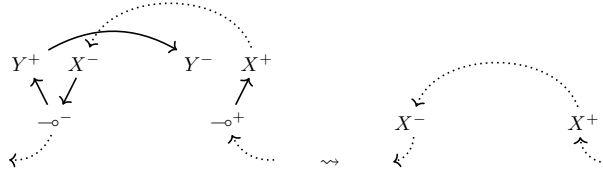
\begin{figure}
 \begin{center}
 \resizebox{0.55\textwidth}{!}{
 \begin{tikzpicture}
 \node(y1) at (-2,0) {$Y^{+}$};
 \node(x1) at (-1,0) {$X^{-}$};
 
 \node(y2) at (1,0) {$Y^{-}$};
 \node(x2) at (2,0) {$X^{+}$};
 
 \node(m1) at (-1.5,-1) {$\multimap^{-}$};
 \node(m2) at (1.5,-1) {$\multimap^{+}$};

 \draw[->, thick] (y1) to [bend left=30] (y2);
 
 \draw[->, thick, dotted] (x2) to [bend right=70] (x1);
 
 \draw[->, thick] (x1) to (m1);
 
 \draw[->, thick] (m1) to (y1); 
 
 \draw[->, thick, dotted] (2.3,-1.7) to [bend left=35] (m2);
  \draw[->, thick, dotted] (m1) to [bend left=35] (-2.3,-1.7);

 \draw[->, thick] (m2) to (x2);
 
 \end{tikzpicture}$ \qquad \leadsto \qquad $
 \begin{tikzpicture}
 \node(x1) at (-1,0) {$X^{-}$};
 
 \node(x2) at (2,0) {$X^{+}$};
 

 
 \draw[->, thick, dotted] (x2) to [bend right=70] (x1);
 
%
 
 \draw[->, thick, dotted] (2.3,-0.7) to [bend left=35] (x2);
  \draw[->, thick, dotted] (x1) to [bend left=35] (-1.3,-0.7);
 
 
 \end{tikzpicture}
 }
 \end{center}
 \caption{Transformation from $f:\Gamma\to A[Y\multimap X/X]$ to $f':\Gamma\to A$.}
 \label{trans}
 \end{figure}

\begin{lemma}\label{regressum}
Let $f:\Gamma\to A[B/X]$ in $\CC A^{\multimap}$, where $B=Y\multimap X$ and $Y\notin FV(\Gamma), FV(A)$. Then the correction graph of $f$ contains no configuration of the form shown in figure \ref{confi}.

\end{lemma}
\begin{proof}
If such a configuration exists, then, by functionality, there must be a path from the leftmost $\multimap^{-}$-link to $Y^{-}_{u'}$. As this path cannot pass through $Y^{+}_{v'}$, its existence implies a configuration as illustrated in figure \ref{confi2}, which contains two isomorphic copies of configuration \ref{confi}. There must be then a new path from the leftmost $\multimap^{-}$-link in \ref{confi2} to $Y^{-}_{u''}$. As all paths are finite (by acyclicity), we must conclude that this is impossible.   

\end{proof}

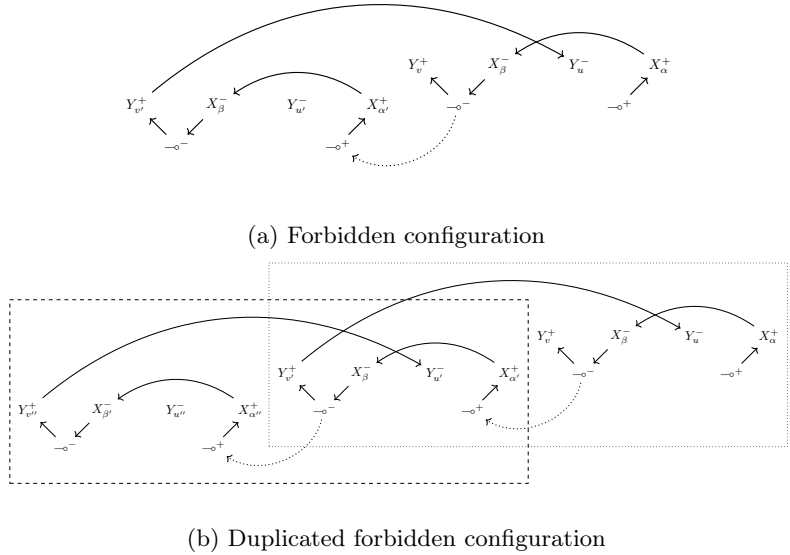
\begin{figure}[h]
\begin{subfigure}{\textwidth}
\begin{center}
\resizebox{0.5\textwidth}{!}{
\begin{tikzpicture}
\node(p1) at (0,0) {$\multimap^{+}$};
\node(x+1) at (1,1) {$X^{+}_{\alpha}$};
\node(y-1) at (-1,1) {$Y^{-}_{u}$};

\draw[->, thick] (p1) to (x+1);

\node(m1) at (-4,0) {$\multimap^{-}$};
\node(x-1) at (-3,1) {$X^{-}_{\beta}$};
\node(y+1) at (-5,1) {$Y^{+}_{v}$};

\draw[->, thick] (x-1) to (m1);
\draw[->, thick] (m1) to (y+1);

\draw[->, thick] (x+1) to [bend right=35] (x-1);

\node(p2) at (-7,-1) {$\multimap^{+}$};
\node(x+2) at (-6,0) {$X^{+}_{\alpha'}$};
\node(y-2) at (-8,0) {$Y^{-}_{u'}$};

\draw[->, thick] (p2) to (x+2);

\node(m2) at (-11,-1) {$\multimap^{-}$};
\node(x-2) at (-10,0) {$X^{-}_{\beta}$};
\node(y+2) at (-12,0) {$Y^{+}_{v'}$};

\draw[->, thick] (x+2) to [bend right=35] (x-2);

\draw[->, thick] (x-2) to (m2);
\draw[->, thick] (m2) to (y+2);

\draw[->, thick] (y+2) to [bend left=35] (y-1);
\draw[->, thick, dotted] (m1) to [bend left=55] (p2);

\end{tikzpicture}}
\end{center}
\caption{Forbidden configuration}
\label{confi}
\end{subfigure}
\begin{subfigure}{\textwidth}
\begin{center}
\resizebox{0.7\textwidth}{!}{
\begin{tikzpicture}
\node(p1) at (0,0) {$\multimap^{+}$};
\node(x+1) at (1,1) {$X^{+}_{\alpha}$};
\node(y-1) at (-1,1) {$Y^{-}_{u}$};

\draw[->, thick] (p1) to (x+1);

\node(m1) at (-4,0) {$\multimap^{-}$};
\node(x-1) at (-3,1) {$X^{-}_{\beta}$};
\node(y+1) at (-5,1) {$Y^{+}_{v}$};

\draw[->, thick] (x-1) to (m1);
\draw[->, thick] (m1) to (y+1);

\draw[->, thick] (x+1) to [bend right=35] (x-1);

\node(p2) at (-7,-1) {$\multimap^{+}$};
\node(x+2) at (-6,0) {$X^{+}_{\alpha'}$};
\node(y-2) at (-8,0) {$Y^{-}_{u'}$};

\draw[->, thick] (p2) to (x+2);

\node(m2) at (-11,-1) {$\multimap^{-}$};
\node(x-2) at (-10,0) {$X^{-}_{\beta}$};
\node(y+2) at (-12,0) {$Y^{+}_{v'}$};

\draw[->, thick] (x+2) to [bend right=35] (x-2);

\draw[->, thick] (x-2) to (m2);
\draw[->, thick] (m2) to (y+2);

\draw[->, thick] (y+2) to [bend left=35] (y-1);
\draw[->, thick, dotted] (m1) to [bend left=55] (p2);

\node(p3) at (-14,-2) {$\multimap^{+}$};
\node(x+3) at (-13,-1) {$X^{+}_{\alpha''}$};
\node(y-3) at (-15,-1) {$Y^{-}_{u''}$};

\draw[->, thick] (p3) to (x+3);

\node(m3) at (-18,-2) {$\multimap^{-}$};
\node(x-3) at (-17,-1) {$X^{-}_{\beta'}$};
\node(y+3) at (-19,-1) {$Y^{+}_{v''}$};

\draw[->, thick] (x+3) to [bend right=35] (x-3);

\draw[->, thick] (x-3) to (m3);
\draw[->, thick] (m3) to (y+3);

\draw[->, thick] (y+3) to [bend left=35] (y-2);
\draw[->, thick, dotted] (m2) to [bend left=55] (p3);

\draw[dotted] (1.5,3) to (1.5,-2) to (-12.5,-2) to (-12.5,3) -- cycle;
\draw[dashed] (-5.5,2) to (-5.5,-3) to (-19.5,-3) to (-19.5,2) -- cycle;

\end{tikzpicture}}
\end{center}
\caption{Duplicated forbidden configuration}
\label{confi2}
\end{subfigure}
\caption{Forbidden configurations}
\end{figure}

\begin{lemma}\label{regressum2}
Let $f:\Gamma\to A[B/X]$ in $\CC A^{\multimap}$, where $B=Y\multimap X$ and $Y\notin FV(\Gamma), FV(A)$. Then, for any type III edge $e_{X}$ over $X$ of $f$, there exists a type III edge $e_{Y}$ over $Y$ of $f$ ``next to $e_{X}$'', i.e. as in  figure \ref{confi3}.
\end{lemma}
\begin{proof}
Suppose the correction graph of $f$ has a configuration as the one in figure \ref{confi4}. Then, by functionality for the rightmost $\multimap^{+}$-link, either there exists a path from the rightmost $\multimap^{-}$-link to the leftmost $\multimap^{+}$-link or there is a path from $Y^{+}_{v}$ to the leftmost $\multimap^{+}$-link. The first case is impossible by lemma \ref{regressum}; the second case forces $u=u'$ and $v=v'$, as from $Y^{+}_{v}$ one must get to the minor premiss $Y^{-}_{v''}$ of a $\multimap^{+}$-link where the path ends (since all negative occurrences of $Y$ are within positive occurrences of $Y\multimap X$). So it must be $v'=v$ and $u=u'$, hence the configuration is as in \ref{confi3}.

\end{proof}

\begin{figure}
\begin{subfigure}{0.4\textwidth}
\begin{center}
\resizebox{0.75\textwidth}{!}{
\begin{tikzpicture}
\node(p1) at (0,0) {$\multimap^{+}$};
\node(x+1) at (1,1) {$X^{+}_{\alpha}$};
\node(y-1) at (-1,1) {$Y^{-}_{u}$};

\draw[->, thick] (p1) to (x+1);

\node(m1) at (-4,0) {$\multimap^{-}$};
\node(x-1) at (-3,1) {$X^{-}_{\beta}$};
\node(y+1) at (-5,1) {$Y^{+}_{v}$};

\draw[->, thick] (x-1) to (m1);
\draw[->, thick] (m1) to (y+1);

\draw[->, thick] (x+1) to [bend right=35] node[above] {$e_{X}$} (x-1);
\draw[->, thick] (y+1) to [bend left=35] node[above] {$e_{Y}$} (y-1);

\end{tikzpicture}}
\end{center}
\caption{Coupled edges $e_{X}$ and $e_{Y}$}
\label{confi3}
\end{subfigure}
\begin{subfigure}{0.58\textwidth}
\begin{center}
\resizebox{\textwidth}{!}{
\begin{tikzpicture}
\node(p1) at (0,0) {$\multimap^{+}$};
\node(x+1) at (1,1) {$X^{+}_{\alpha}$};
\node(y-1) at (-1,1) {$Y^{-}_{u}$};

\draw[->, thick] (p1) to (x+1);

\node(m1) at (-4,0) {$\multimap^{-}$};
\node(x-1) at (-3,1) {$X^{-}_{\beta}$};
\node(y+1) at (-5,1) {$Y^{+}_{v}$};

\draw[->, thick] (x-1) to (m1);
\draw[->, thick] (m1) to (y+1);

\draw[->, thick] (x+1) to [bend right=35] node[above] {$e_{X}$} (x-1);

\node(p2) at (-7,-1) {$\multimap^{+}$};
\node(x+2) at (-6,0) {$X^{+}_{\alpha'}$};
\node(y-2) at (-8,0) {$Y^{-}_{u'}$};

\draw[->, thick] (p2) to (x+2);

\node(m2) at (-11,-1) {$\multimap^{-}$};
\node(x-2) at (-10,0) {$X^{-}_{\beta}$};
\node(y+2) at (-12,0) {$Y^{+}_{v'}$};

\draw[->, thick] (x+2) to [bend right=35] (x-2);

\draw[->, thick] (x-2) to (m2);
\draw[->, thick] (m2) to (y+2);

\draw[->, thick] (y+2) to [bend left=35] (y-1);

\end{tikzpicture}}
\end{center}
\caption{Forbidden cofiguration}
\label{confi4}
\end{subfigure}
\caption{}
\end{figure}

By lemma \ref{regressum2}, any type III edge $e_{X}$ over $X$ in $f:\Gamma\to A[B/X]$ is coupled with a type III edge $e_{Y}$ over $Y$ as in figure \ref{confi3}. We can deduce then that interpolation over $f-f_{Y}:\Gamma\to A$ induces interpolation over $f $ in the following sense:

\begin{lemma}\label{regressum3}
Let $f:\Gamma\to A[B/X]$ in $\CC A^{\multimap}$, where $B=Y\multimap X$ and $Y\notin FV(\Gamma), FV(A)$. Suppose moreover $f$ has a type III edge over $X$. 
Then, the splitting of $f'=f-f_{Y}:\Gamma\to A$ into $p$ graphs $f_{i}:\Gamma_{i}\to A_{i}$ obtained by weak interpolation induces a splitting of $f$ into $p$ graphs $g_{i}:\Gamma_{i}\to A_{i}[B/X]$, where $f_{i}=g_{i}-f_{Y}$.
\end{lemma}
\begin{proof}
By lemma \ref{regressum2}, the $Y$-pairing of $f$ can be partitioned in two sets $p_{1},p_{2}$, where $p_{1}$ contains those edges which are coupled with type III edges over $X$ and $p_{2}$ contains those edges which occur close to type I edges over $X$. The edges in $p_{2}$ induce then a pairing $p$ of the type I edges over $X$, with edges $(e,e')\in p$ when $e$ contains $X^{+}_{\alpha}$, $e'$ contains $X^{-}_{\beta}$ and the two occurrences $Y^{-}_{\alpha'}$ and $Y^{+}_{\beta'}$ occurring next to $X^{+}_{\alpha}$ and $X^{-}_{\beta}$, respectively, form a type III edge $e_{Y}$  in $p_{2}$. After weak interpolation is performed, the correction graph is as in figure \ref{regre}. Indeed, by functionality for the rightmost $\multimap^{+}$ link, there must be a path from $X^{-}_{\alpha'}$ to $X^{+}_{\alpha'}$, which can only be a path through a negative branch in the shape of $A_{i}^{\OV \epsilon}$. By duality, there is then a positive branch in the shape of $A_{i}^{\epsilon}$, hence, by functionality for the leftmost $\multimap^{+}$ link, there must be a path from $X^{-}_{\gamma}$ to $X^{+}_{\gamma}$.

Under these conditions, the edge $(X^{-}_{\beta'},^{+}_{\beta'})$ is expansible in $A_{i}$ (by lemma \ref{inte2}), so we can ``transport'' the edge $(Y^{+}_{v},Y^{-}_{u})$ onto $A_{i}$ as shown in figure \ref{regg}, preserving correctness. By arguing in this way for all edges in $p_{2}$ we obtain correct graphs $g_{i}:\Gamma_{i}\to A_{i}[Y\multimap X/X]$.
\end{proof}

\begin{figure}[h]
\begin{center}\resizebox{0.7\textwidth}{!}{
\begin{tikzpicture}[scale=1.3]

\draw[] (0,0) to (1,0) to (2,1) to (-1,1) -- cycle;
\draw[] (6.8,-1) to (8.3,1) to (5.5,1) -- cycle;

\node(g) at (0.5,-0.3) {$\Gamma_{i}$};
\node(a) at (6.8,-1.3) {$A$};

\node(y1) at (5.7,1.2) {\scriptsize$Y^{+}_{\alpha}$};
\node(x1) at (6.5,1.2) {\scriptsize$X^{-}_{\beta}$};

\node(y2) at (7.2,1.2) {\scriptsize$Y^{-}_{\beta}$};
\node(x2) at (8,1.2) {\scriptsize$X^{+}_{\alpha}$};

\node(x3) at (4.2,1.2) {\scriptsize$X^{+}_{\alpha'}$};
\node(x4) at (4.8,1.2) {\scriptsize$X^{-}_{\alpha'}$};

\node(x5) at (2.7,1.2) {\scriptsize$X^{-}_{\beta'}$};
\node(x6) at (3.3,1.2) {\scriptsize$X^{+}_{\beta'}$};

\node(x7) at (0.5,1.2) {\scriptsize$X^{-}_{\gamma}$};
\node(x8) at (1.5,1.2) {\scriptsize$X^{+}_{\gamma}$};

\node(t) at (6.1,0.5) {\scriptsize$\multimap^{-}$};
\node(p) at (7.6,0.5) {\scriptsize$\multimap^{+}$};

\draw[->] (x1) to (t);
\draw[->] (p) to (x2);
\draw[->] (t) to (y1);

\draw[->] (y1) to [bend left=45] node[above] {$e_{Y}$} (y2);
\draw[->] (x2) to [bend right=45] node[above] {$e$} (x4); 
\draw[->] (x3) to [bend left=45] node[above] {$e'$} (x1); 
\draw[<-] (x8) to [bend left=45] (x6); 
\draw[<-] (x5) to [bend right=45] (x7);

\draw[] (2.5,1) to (3.5,1) to (3,0) -- cycle;
\draw[] (4,1) to (5,1) to (4.5,0) -- cycle;
\node(i1) at (3,-0.3) {$A_{i}^{\epsilon}$};
\node(i2) at (4.5,-0.3) {$A_{i}^{\OV\epsilon}$};

\draw[] (i1) to [bend right=55] node[below] {$cut$} (i2);

\node(tt) at (3,0.5) {\scriptsize$\multimap^{+}$};
\draw[->] (tt) to (3.3,1);

\node(pp) at (4.5,0.5) {\scriptsize$\multimap^{-}$};
\draw[->] (4.8,1) to (pp);
\draw[->] (pp) to (4.2,1);
\draw[<-, thick, dotted] (0.5,0.9) to [bend right=55] (1.5,0.9);

\end{tikzpicture}}
\end{center}
\caption{Correction graph after weak interpolation}
\label{regre}
\end{figure}
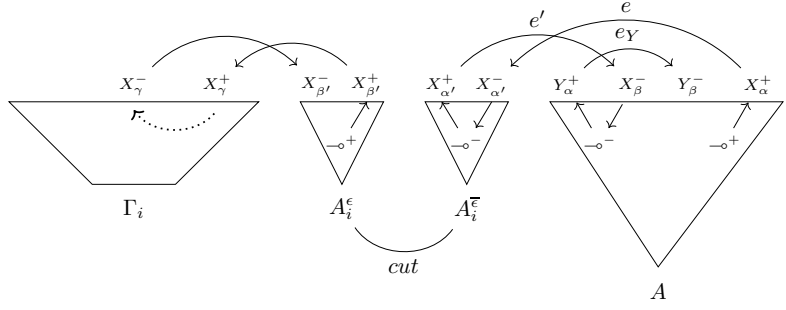
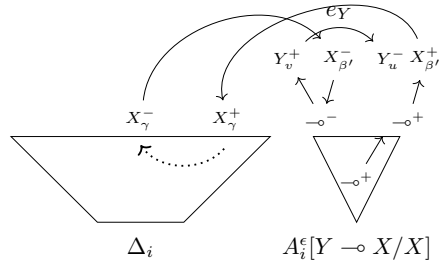
\begin{figure}
\begin{center}\resizebox{0.4\textwidth}{!}{
\begin{tikzpicture}[scale=1.3]

\draw[] (0,0) to (1,0) to (2,1) to (-1,1) -- cycle;

\node(g) at (0.5,-0.3) {$\Delta_{i}$};

\node(x5) at (2.8,1.9) {\scriptsize$X^{-}_{\beta'}$};
\node(x6) at (3.8,1.9) {\scriptsize$X^{+}_{\beta'}$};

\node(x7) at (0.5,1.2) {\scriptsize$X^{-}_{\gamma}$};
\node(x8) at (1.5,1.2) {\scriptsize$X^{+}_{\gamma}$};

\node(p) at (2.6,1.2) {\scriptsize$\multimap^{-}$};

\node(t) at (3.6,1.2) {\scriptsize$\multimap^{+}$};

\node(y1) at (2.2,1.9) {\scriptsize$Y^{+}_{v}$};
\node(y2) at (3.4,1.9) {\scriptsize$Y^{-}_{u}$};

\draw[->] (y1) to [bend left=45] node[above] {$e_{Y}$} (y2);

\draw[<-] (x8) to [bend left=85] (x6); 
\draw[<-] (x5) to [bend right=65] (x7); 

\draw[->] (x5) to (p);
\draw[->] (p) to (y1);
\draw[->] (t) to (x6);

\draw[] (2.5,1) to (3.5,1) to (3,0) -- cycle;
\node(i1) at (3,-0.3) {$A_{i}^{\epsilon}[Y\multimap X/X]$};

\node(tt) at (3,0.5) {\scriptsize$\multimap^{+}$};
\draw[->] (tt) to (3.3,1);

\draw[<-, thick, dotted] (0.5,0.9) to [bend right=55] (1.5,0.9);

\end{tikzpicture}}
\end{center}
\caption{Transport of edge $e_{Y}$}
\label{regg}
\end{figure}

By exploiting the previous lemmas, proposition \ref{pairing} as well as lemma \ref{pairing} we can prove the proposition below, which is the fundamental step to characterize linearly expansible types.

\begin{lemma}\label{trave0}
Let $f:\Gamma\to A[B/X]$ in $\CC A^{\multimap}$, where $B=Y\multimap X$ and $Y\notin FV(\Gamma), FV(A)$. Then either $A\in  gRP_X$ or $f-f_{Y}:\Gamma\to A$ has $gRP_{X}$ interpolants. 
\end{lemma}
\begin{proof}




Suppose $A\notin gRP_{X}$. By lemma \ref{paired} $A$ is paired and by lemma \ref{pairing}, any  $X$-pairing of $A$ contains a jump out edge. We deduce that for any $Y$-pairing of $A[B/X]$, one can find occurrences of $(Y_{\alpha}\multimap X_{\alpha})^{+},(Y_{\beta}\multimap X_{\beta})^{-}$ such that the edge $(X^{+}_{\alpha},X^{-}_{\beta})$ is a jump out edge. 

Let then $e=(X^{+}_{\alpha}, X^{-}_{\beta})$ be such a jump out edge and let us consider the correction graph of $f':\Gamma\to A$. 
By lemma \ref{jumpout} it follows that either $X_{\alpha}^{+}\not \prec_{f'} X_{\beta}^{-}$ or $e=(X_{\alpha}^{+},X_{\beta}^{-})\in f'$. 
The first case is impossible: as there is a path in the correction graph of $f'$ going from the conclusion to $X^{-}_{\beta}$ without passing through $X_{\alpha}^{+}$, the functionality condition must fail, as illustrated in figure \ref{shape2}, contradicting the hypothesis that $f'$ is correct.  

 \begin{figure}[h]
\begin{center}
\resizebox{0.4\textwidth}{!}{
\begin{tikzpicture}
\node(a) at (0,0) {$Y^{+}\multimap^{-}X^{-}_{\beta}$};
\node(b) at (-1,1) {$Y^{+}$};
\node(c) at (1,1) {$X^{-}_{\beta}$};

\draw[->, thick] (c) to (a);

\draw[->, thick] (a) to (b);

\node(a') at (5,0) {$Y^{-}\multimap^{+}X^{+}_{\alpha}$};
\node(b') at (4,1) {$Y^{-}$};
\node(c') at (6,1) {$X^{+}_{\alpha}$};

\draw[->, thick] (a') to (c');


\node(cc) at (4,-2) {$A\multimap A[B/X]$};

\draw[->, thick] (b) to [bend left=35] (b');

\draw[->, thick, dotted] (cc) to [bend right=55] (c);

\end{tikzpicture}
}
\end{center}
\caption{Failure of functionality}
\label{shape2}
 \end{figure}
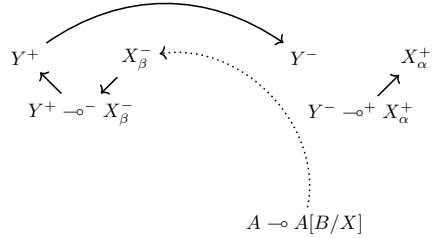

Hence $e=(X_{\alpha}^{+},X_{\beta}^{-})\in f'$. 
We can now apply weak interpolation to eliminate all type III edges from $f'$, included $e$. 
By lemma \ref{regressum3}, we conclude that there exists a partition $\Gamma_{1},\dots, \Gamma_{p}$ of $\Gamma$, types $A_{1},\dots, A_{p}\hookrightarrow A$, with $\sum_{i}n(A_{i})< n(A)$ and correct graphs $g_{1},\dots, g_{p}$, where $g_{i}:\Gamma_{i}\to A_{i}[B/X]$ and $g_{i}-f_{Y}:\Gamma_{i}\to A_{i}$, where the latter have no type III edge. 

Now it must be $A_{i}\in gRP_{X}$ as, by the same argument as above, if $A_{i}\notin gRP_{X}$, then $g_{i}-f_{Y}$ has a type III edge over $X$, which is impossible.

\end{proof}



By applying lemma \ref{trave0} with $\Gamma=\{A\}$, as well as remark \ref{unique}, we get:

\begin{theorem}\label{expa1}
$A\in \CC L_{\multimap}$ is linearly expansible iff either $A$ is $gRP_X$ or $A$ linearly collapses into a $gRP_X$ type.

\end{theorem}

\begin{figure}
\begin{center}
\resizebox{0.9\textwidth}{!}{
\begin{tikzpicture}
\node(a) at (0,0) {$(X\multimap X)\multimap (X\multimap X)$};

\draw[<-, thick] (1.4,-0.2) to [bend right=18] (3.6,-2.8);

\draw[->, thick] (-2.7,-2.8) to [bend right=18] (-1.4,-0.2);
\draw[->, thick] (0.6,-0.2) to [bend right=9] (1.6,-2.8);
\draw[->, thick] (-0.6,-0.2) to [bend left=9] (-0.6,-2.8);

\draw[->, thick] (-1.6,-2.8) to [bend left=35] (2.7,-2.8);

\draw[->, thick] (0.6,-2.8) to [bend right=35] (-3.6,-2.8);

\node(d) at (0,-3) {$((Y\multimap X)\multimap (Y\multimap X))\multimap ((Y\multimap X)\multimap (Y\multimap X))$};

\end{tikzpicture}
\begin{tikzpicture}
\node(a) at (0,0) {$(X\multimap X)\multimap (X\multimap X)$};

\draw[<-, thick] (1.4,-0.2) to [bend right=18] (3.6,-2.8);

\draw[->, thick] (-2.7,-2.8) to [bend right=18] (-1.4,-0.2);
\draw[->, thick] (0.6,-0.2) to [bend right=9] (1.6,-2.8);
\draw[->, thick] (-0.6,-0.2) to [bend left=9](-0.6,-2.8);

\draw[->, thick] (0.6,-2.8) to [bend left=55] (2.7,-2.8);

\draw[->, thick] (-1.6,-2.8) to [bend right=55] (-3.6,-2.8);

\node(d) at (0,-3) {$((Y\multimap X)\multimap (Y\multimap X))\multimap ((Y\multimap X)\multimap (Y\multimap X))$};

\end{tikzpicture}
}
\end{center}
\caption{Simple expansions of $(X\multimap X)\multimap (X\multimap X)$}
\label{collapsec1}
\end{figure}
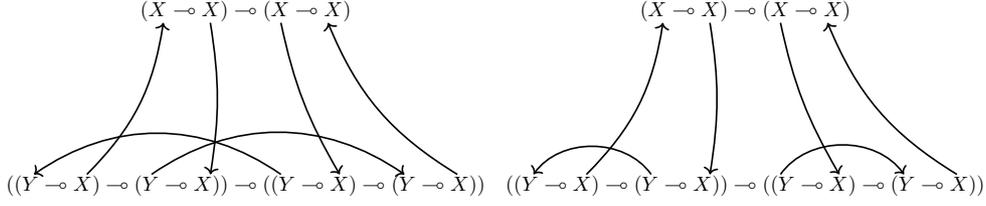

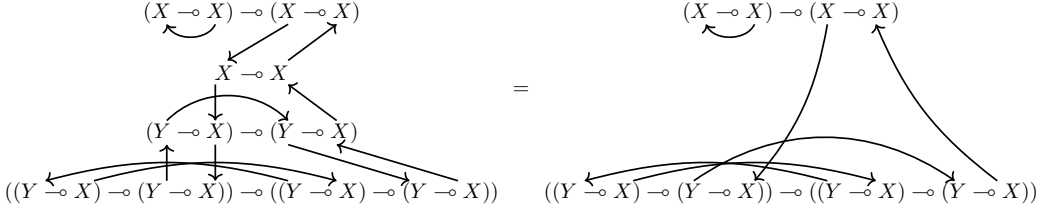
\begin{figure}
\begin{center}
\resizebox{0.95\textwidth}{!}{
\begin{tikzpicture}[baseline=-9ex]
\node(a) at (0,0) {$(X\multimap X)\multimap (X\multimap X)$};

\node(b) at (0,-1) {$X\multimap X$};

\draw[->, thick] (0.6,-0.2) to (-0.4,-0.8);
\draw[<-, thick] (1.4,-0.2) to (0.6,-0.8);
\draw[->, thick] (-0.6,-0.2) to [bend left=65] (-1.4,-0.2);

\node(c) at (0,-2) {$(Y\multimap X)\multimap (Y\multimap X)$};

\draw[<-, thick] (0.6,-1.2) to (1.4,-1.8);
\draw[->, thick] (-0.6,-1.2) to (-0.6,-1.8);
\draw[->, thick] (-1.4,-1.8) to [bend left=45] (0.6,-1.8);

\node(d) at (0,-3) {$((Y\multimap X)\multimap (Y\multimap X))\multimap ((Y\multimap X)\multimap (Y\multimap X))$};

\draw[<-, thick] (1.4,-2.2) to (3.4,-2.8);
\draw[->, thick] (0.6,-2.2) to (2.6,-2.8);

\draw[->, thick] (-0.6,-2.2) to (-0.6,-2.8);
\draw[<-, thick] (-1.4,-2.2) to (-1.4,-2.8);

\draw[<-, thick] (1.4,-2.8) to [bend right=15] (-2.6,-2.8);
\draw[<-, thick] (-3.4,-2.8) to [bend left=15]  (0.6,-2.8);

\end{tikzpicture}
$=$
\begin{tikzpicture}[baseline=-9ex]
\node(a) at (0,0) {$(X\multimap X)\multimap (X\multimap X)$};

\draw[->, thick] (-0.6,-0.2) to [bend left=65] (-1.4,-0.2);

\draw[->, thick] (-1.6,-2.8) to [bend left=35] (2.7,-2.8);

\node(d) at (0,-3) {$((Y\multimap X)\multimap (Y\multimap X))\multimap ((Y\multimap X)\multimap (Y\multimap X))$};

\draw[<-, thick] (1.4,-0.2) to [bend right=15] (3.4,-2.8);
\draw[->, thick] (0.6,-0.2) to [bend left=15] (-0.6,-2.8);

\draw[<-, thick] (1.4,-2.8) to [bend right=15] (-2.6,-2.8);
\draw[<-, thick] (-3.4,-2.8) to [bend left=15]  (0.6,-2.8);

\end{tikzpicture}
}
\end{center}
\caption{Correct expansion of $(X\multimap X)\multimap(X\multimap X)$}
\label{collapsec}
\end{figure}

By putting together lemma \ref{inte2}, proposition \ref{pairing} and lemma \ref{trave0}, we obtain a nice characterization of $gRP_X$ types:
\begin{proposition}\label{$gRP_X$2}
$A\in  gRP_X$ iff for any type $B$ there exists a correct simple $B$-expansion of $A$.
\end{proposition}
\begin{proof}
By proposition \ref{pairing}, $A\in  gRP_X$ iff it admits a $X$-pairing $p$ with no jump out. Hence, if $A\in  gRP_X$, then, by lemma \ref{inte2}, for any $B$ the $p_{B}$-expansion of $id_{A}$ is correct. If $A\notin  gRP_X$, by reasoning as in the proof of lemma \ref{trave0} we can conclude that, by letting $B=Y\multimap X$, any correct graph $f:A\to A[B/X]$ contains a type III edge over $X$, hence it is not a simple $B$-expansion.

\end{proof}

Proposition \ref{$gRP_X$2} refines proposition \ref{drave}, as it shows that if $A\in  gRP_X$, then it is not only linearly expansible, but, for any type $B=B_{1}\multimap\dots\multimap B_{n}\multimap X$, the graph $EXP_{A}(B):A\to A[B/X]$ is a simple $B$-expansion graph. If a type $A$ is linearly expansible but not $gRP_{X}$, then, by theorem \ref{expa1}, $A$ linearly collapses into a $gRP_{X}$ type $A'$. By proposition \ref{$gRP_X$2}, the simple $B$-expansion graphs of $A$ are not correct. However, $A$ can be expanded by composing arrows $A \to A' \to A'[B/X] \to A[B/X]$. The resulting graph is not a simple expansion graph. Indeed, as the last arrow in the chain comes from interpolation, the graph has a type III edge over $X$.

For example, the type $C=(X\multimap X)\multimap (X\multimap X)$, which is not $gRP_X$, has two simple $Y\multimap X$-expansion graphs, shown in figure \ref{collapsec1}, both not correct. However, as $C$ linearly collapses into the $gRP_{X}$ type $X\multimap X$, $C$ is linearly expansible: a correct expansion of $D$ is obtained by collapsing it on $X\multimap X$, as shown in figure \ref{collapsec}.


\section{From the expansion property to instantiation overflow}\label{sec8}

In this section we establish our main result, that is, that a simple type $A$ has instantiation overflow if and only if $\forall XA$ is either derivable or logically equivalent to a product of $gRP$ types.

First, we consider generalized Russell-Prawitz types in $\lambda_{\To}$ and we prove that a simple type $A$ is expansible iff it is either derivable or logically equivalent to a product of $gRP_{X}$ types. The characterization is slightly different from the one given for $\lambda_{\multimap}$, as one must consider that derivable types are expansible in $\lambda_{\To}$, though not in $\lambda_{\multimap}$, and that weak interpolation in $\lambda_{\To}$ is significantly more complex than weak interpolation in $\lambda_{\multimap}$.

Then we consider the instantiation overflow problem for the types $\forall XA$, with $A\in \CC L_{\To}$. We suitably extend the expansion property and the notion of collapse to $F_{at}$ and we prove (1) that a type $A$ is $F_{at}$-expansible iff $\forall XA$ is either derivable or logically equivalent to a product of $gRP$ types and (2) that $A$ has instantiation overflow iff it is $F_{at}$-expansible.

\subsection{Expansible and generalized Russell-Prawitz types}

Similarly to the previous section, we fix a variable $X\in \CC V$. A type $A\in \CC L_{\To}$ will be called \emph{expansible} when for every simple types $C_{1},\dots, C_{p}\in \CC L_{\To}$, there exists an arrow $EXP_{A}(B): A \to A[B/X]$ in $\CC T$, where $B$ is 
$C_{1}\To \dots \To C_{p}\To X$.

Similarly to the last section, we let a simple type be $co-gRP_{X}$ when it is of the form $A_{1}\To \dots \To A_{n}\To Z$, where $Z\neq X$ and the $A_{i}$ are $gRP_{X}$.
We provide an equivalent inductive definition of $gRP_{X}$ and $co-gRP_{X}$ types in $\CC L_{\To}$: 

%

\begin{definition}\label{expo$gRP_X$}
We define by mutual induction the classes $ gRP_X^n, co-gRP_X^n$, for $n\in \BB N$:
\begin{enumerate}

\item if $A$ has no occurrence of $X$, then $A\in  gRP_X^0, co-gRP_X^0$;

\item $X\in  gRP_{X}^0,co-gRP_X^{1}$; 


%
\item if $B\in  gRP_X^p$ and $C\in co-gRP_X^{q}$, then $B\To C\in co-gRP_X^{p+q}$;

\item if $B\in co-gRP_X^p$ and $C\in  gRP_X^{q}$, then $B\To C\in  gRP_X^{p+q}$;



\item if $A_{1}\in co-gRP_X^{n_1},\dots, A_{p}\in co-gRP_X^{n_p}$, then the type
$
A_{1}\To \dots  \To A_{p}\To X
$
is in $ gRP_X^{q}$ for all $q< \sum_{i}^{p}{}n_{i}$.


\end{enumerate}

We let $ gRP_X:= gRP_X^0$ (resp. $co  gRP_X:=co-gRP_X^0$).
\end{definition}

We highlight some differences between definition \ref{$gRP_X$} and definition \ref{expo$gRP_X$}. The type $X$ is  $gRP_X$ but not linear $gRP_{X}$. This corresponds to the fact that $X$ is expansible in $\lambda_{\To}$ but not in $\lambda_{\multimap}$. 
As we already observed, the type $D=(A\To X)\To (B\To X) \To X$, where $X\notin FV(A),FV(B)$ is $gRP_X$, while the corresponding linear type $D^{\multimap}$ is not linear $gRP_X$.
%
%


$gRP_{X}$ types in $\CC L_{\multimap}$ and $\CC L_{\To}$ are related by the following facts, which are easily proved by induction on a type $B\in \CC L_{\To}$:

\begin{lemma}\label{enest}
For all types $A\in \CC L_{\multimap}, B\in \CC L_{\To}$ and $k\in \BB N$, the following hold:
\begin{itemize}
\item if $A\in \CC E^{+}(B)$ and $A\in  gRP_X^{k}$, then there exists $h\leq k$ such that $B\in gRP_{X^{h}}$;

\item if $A\in \CC E^{-}(B)$ and $A\in co-gRP_X^{k}$, then there exists $h\leq k$ such that $B\in co-gRP_{X}^{h}$.
\end{itemize}
\end{lemma}

From lemma \ref{enest} we deduce that for all $A\in \CC L_{\multimap}, B\in \CC L_{\To}$:
 \begin{itemize}
 \item[-] if $A$ is $gRP_{X}$, then $A^{\To}$ is $gRP_{X}$;

 \item[-] if $A\in \CC E^{+}(B)$ and $A$ is $gRP_{X}$, then $B$ is $gRP_{X}$.
 \end{itemize}


The following proposition shows that $gRP_X$  types are expansible.

\begin{proposition}[simple expansions]\label{drave}
For all simple types $A, B=B_{1}\To \dots \To B_{k}\To X\in \CC L_{\To}$ and $k\in \BB N$, if $A\in gRP_{X}^{k}$ (resp. $A\in co-gRP_{X}^{k}$), there exist an arrow $EXP_{A}(B):\Delta, A\to A[B/X]$ (resp. $coEXP_{A}(B):\Delta,A[B/X]\to A$) in $m\CC T$, where 
 $\Delta=x_{1}:B_{1},\dots, x_{n}:B_{n}$ if $k\geq 1$ and $\Delta=\emptyset$ if $k=0$. 

In case $k=0$ we call $EXP_{A}(B)$ (resp. $coEXP_{A}(B)$) the \emph{simple $B$-expansion} (resp. \emph{simp $B$-coexpansion} of $A$. 
\end{proposition}
\begin{proof}
Induction on clauses 1.-5.:
\begin{enumerate}
\item if $X\notin FV(A)$, then $EXP_{A}(B)= coEXP_{A}(B)= x$;
\item if $A=X$, then $EXP_{A}(B)= Intro_{B}x$ and $coEXP_{A}(B)=Elim_{B}x$;
\item if $A=A_{1}\To A_{2}$, then $coEXP_{A}(B)=
EXP_{A_{1}}(B)\To coEXP_{A_{2}}(B)$;

\item if $A=A_{1}\To A_{2}$, then $EXP_{A}(B)=
coEXP_{A_{1}}(B)\To EXP_{A_{2}}(B)$;

\item if $A=A_{1}\To \dots \to A_{p}\To X$, where $A_{i}\in co  gRP^X_{p_{i}}$ and for at least one $i$, say $i=c$, $p_{c}=k+1$, then 
$$EXP_{A}(B)=   \lambda x_{1}^{A_{1}[B/X]}.\dots.\lambda x_{p}^{A_{p}[B/X]}.  Intro_{B}\big ( x ( coEXP_{A_{1}}(B)[x_{1}/x])\dots (coEXP_{A_{p}}(B)[x_{p}/x])$$

\end{enumerate}

\end{proof}

We extend the notion of simple $B$-expansion in accordance with theorem \ref{linear}: given $A\in \CC L_{\multimap}, B\in \CC L_{\To}$, by a \emph{generalized simple $B$-expansion of $A$} we indicate a graph $f:A\to C$ such that $id_{A}\subseteq f$, where $C$ is obtained from $A$ by replacing positive (resp. negative) occurrences of $X$ by some $B'\in \CC E^{+}(B)$ (resp. $B'\in \CC E^{-}(B)$);
dually, a \emph{generalized simple $B$-expansion} is a graph $f:C\to A$ such that $id_{A}\subseteq f$, where $C$ is obtained from $A$ by replacing positive (resp. negative) occurrences of $X$ by some $B'\in \CC E^{-}(B)$ (resp. $B'\in \CC E^{+}(B)$).

The fact that $EXP_{A}(B), coEXP_{A}(B)$ are called simple expansions and coexpansions respectively comes from the following property:

\begin{proposition}\label{simplexpaTo}
Let $A\in \CC L_{\To}$ be $gRP_{X}$ (resp. $co-gRP_{X}$). Then, for any $B\in \CC L_{\To}$, the graph of $\partial EXP_{A}(B)$ (resp. $\partial coEXP_{A}(B)$) is a generalized simple expansion graph (resp. a generalized simple coexpansion graph).

\end{proposition}
\begin{proof}
Given $A\in  gRP^X_{k}$ (resp. $A\in co-gRP^X_{k}$), it suffices to show that the graph $\partial EXP_{A}(B)$ (resp. $\partial coEXP_{A}(B)$) contains $id_{A}$. We argue by induction on clauses 1.-5.. For clauses 1.-2. the claim is immediate. For clauses 3.-4. we argue as follows: 
first, given graphs $g:C_{1}\to D_{1}$, $h:D_{2}\to C_{2}$ corresponding to $\lambda$-terms $u_{g}$ and $u_{h}$, the graph of the $\lambda$-term $u_{g}\To  u_{h}$ is simply the graph $g\multimap h:=g+h: (D_{1}\multimap D_{2}) \to (C_{1}\multimap C_{2})$. Now the claim follows from the fact that, if $g$ is a simple co-expansion (resp. expansion) and $h$ a simple expansion (resp. co-expansion), then $g\multimap h$ is a simple expansion (resp. co-expansion). Indeed, if $id_{D_{1}}\subseteq g$ (resp $id_{C_{1}}\subseteq g$) and $id_{D_{2}}\subseteq h$ (resp. $id_{C_{2}}\subseteq h$), then $id_{D_{1}\multimap D_{2}}=id_{D_{1}}+id_{D_{1}}\subseteq g\multimap h$ (resp. $id_{C_{1}\multimap C_{2}}=id_{C_{1}}+id_{C_{2}}\subseteq g\multimap h$).

Finally, in case 5., that $id_{A}$ is contained in the graph $f$ of $\partial EXP_{A}(B)$ can be seen from the fact that the  graphs of the terms $\partial coEXP_{A_{i}}(B)$ are contained in $f$, where $A_{1}\multimap\dots\multimap A_{p}\multimap X$, as well as the graph $id_{X}$, as illustrated in figure \ref{expamega} (where we only drew the edges over $A$). We conclude that $id_{A}=id_{X}+\sum_{i}^{p}id_{A_{i}}\subseteq f$.

\end{proof}

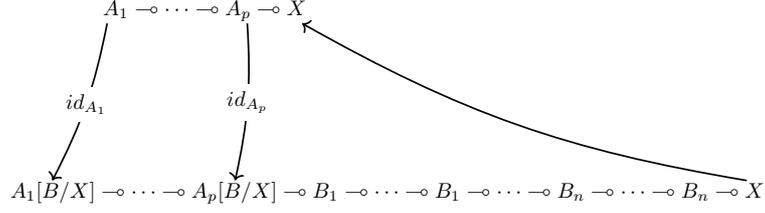
\begin{figure}
\begin{center}
\resizebox{0.7\textwidth}{!}{
\begin{tikzpicture}[every node/.style={fill=white}]
\node(a) at (-3,0) {$A_{1}\multimap \dots \multimap A_{p}\multimap X$};

\node(b) at (0,-3) {$A_{1}[B/X]\multimap \dots \multimap A_{p}[B/X]\multimap B_{1}\multimap \dots \multimap B_{1}\multimap \dots \multimap B_{n}\multimap \dots \multimap B_{n}\multimap X$};

\draw[->, thick] (5.9,-2.8) to [bend left=10] (-1.4,-0.2);

\draw[->, thick] (-4.6,-0.2) to [bend left=8] node {$id_{A_{1}}$} (-5.5,-2.8);

\draw[->, thick] (-2.3,-0.2) to [bend left=8] node {$id_{A_{p}}$} (-2.5,-2.8);

\end{tikzpicture}
}
\end{center}
\caption{Generalized simple expansion graph}
\label{expamega}
\end{figure}

As in the previous section, we label distinct occurrences of $X$ in a simple type with integers $\alpha\in \BB N$.
The notion of $(n,\epsilon)$-pairing is replaced, in this context, by the notion of $(n,\epsilon)$-tiling:
\begin{definition}
Let $A\in \CC L_{\To}$. A list $L=(X_{\alpha_{0}}^{\epsilon}, X_{\alpha_{1}}^{\OV\epsilon},\dots, X_{\alpha_{n}}^{\OV\epsilon})$ is a \emph{tile in $A$} (resp. a \emph{co-tile in $A$}) if $X_{\alpha_{0}}$ is a positive (resp. negative) occurrence of $X$ in $A$ and the $X_{\alpha_{i}}$, for $1\leq i\leq n$ are distinct negative (resp. positive) occurrences of $X$ in $A$. 

%



A \emph{$(n,\epsilon)$-tiling} (resp. a \emph{$(n,\epsilon)$-co-tiling}) of $A$ is a
pair $(\F p,\F a)$, where $\F a$ is a set containing labels of distinct positive (resp. negative) occurrences of $X$ in $A$ and $\F p$ is a tiling (resp. a co-tiling) of all remaining occurrences of $X$ in $A$.

%
%

We will call a $(0,\epsilon)$-tiling (resp. a $(0,\epsilon)$-co-tiling) simply a tiling (resp. a co-tiling).

\end{definition}

The notion of jump out is extended to the case of tiles: a tile $(X_{\alpha_{0}}^{+},X_{\alpha_{1}}^{-},\dots,X_{\alpha_{n}}^{-})$ is a jump out tile in $A$ if for some $1\leq i\leq n$, $(X_{\alpha_{0}}^{+},X_{\alpha_{i}}^{-})$ is a jump out edge in $A^{\multimap}$. Similarly, a co-tile $(X_{\alpha_{0}}^{-},X_{\alpha_{1}}^{+},\dots,X_{\alpha_{n}}^{+})$ is a jump in co-tile in $A$ if for some $1\leq i\leq n$, $(X_{\alpha_{0}}^{-},X_{\alpha_{i}}^{+})$ is a jump in edge in $A^{\multimap}$. A tile which is not a jump out will be called an internal tile.

The proposition below is the analogous, in this frame, of proposition \ref{pairing}. We omit the proof as the argument is similar to that of proposition \ref{pairing}.

\begin{proposition}\label{tiling}
For any type $A$ and $n\in \BB N$,
\begin{itemize}
\item $A\in  gRP_X^n$ iff $A$ has a $(n,-)$-tiling with no jump out;
\item $A\in co-gRP_X^n$ iff $A$ has a $(n,+)$-co-tiling with no jump in. 

\end{itemize}
\end{proposition}
\begin{proof}
Similar to the proof of proposition \ref{pairing}. 
\end{proof}

Proposition \ref{tiling} allows to prove the equivalence of the two definitions of $gRP_X$ types:

\begin{proposition}\label{equivalences2}
$gRP_X$ types following definitions \ref{grpx} and \ref{expo$gRP_X$} coincide.
\end{proposition}
\begin{proof}
Similar to the proof of proposition \ref{equivalences}, with tilings in place of pairings.

\end{proof}

Also proposition \ref{renaming} can be straightforwardly extended to $\lambda_{\To}$:

\begin{proposition}\label{renaming2}
Suppose $u: \Gamma\to A$ in $\CC T$ is normal and $\eta$-long, where $\Gamma=x_{1}:A_{1},\dots,x_{n}: A_{n}$ is a context made of $co-gRP_{X}$ types and $A$ is $gRP_X$. Then there exists a finite set $\CC X=\{X_{1},\dots, X_{n}\}\subset \CC V$, a context $\Gamma'=A'_{1},\dots, A'_{n}$ made of $qRP_{\CC X}$ types such that $A'_{i}[X/X_{1},\dots, X/X_{p}]=A_{i}$ and a $qRP_{\CC X}$ type $A$ such that $A'[X/X_{1},\dots, X/X_{p}]=A$ and $u':\Gamma'\to A'$, where $u'$ is an appropriate renaming of the types appearing in $u$. 

\end{proposition}
\begin{proof}
Similar to the proof of proposition \ref{renaming}.

\end{proof}

\subsection{Characterization of expansible types}

We adapt to $\lambda_{\To}$ the argument in subsection \ref{sec53} that generalized Russell-Prawitz types are ``dense'' in the class of expansible types.

Let $u:\Gamma\to A[B/X]$ in $\CC T$ be normal and $\eta$-long, where $B=Y\To X$ and $Y\notin FV(A), FV(\Gamma)$. Let $v_{C}\in Subt(u)$ be a term such that $[v_{C}]=C[B/X]$, where $C$ is a type of the form $C_{1}\To \dots \To C_{p}\To X$ occurring positively in $A$. Then $v_{C}$ is of the form
$\lambda x_{1}^{C_{1}[B/X]}.\dots.\lambda x_{p}^{C_{p}[B/X]}.\lambda y^{Y}.zu_{1}\dots u_{q} $. 
We say that $u$ is \emph{good} when for all $v_{C}\in Subt(u)$ as above, whenever $q\geq 1$ and $[u_{q}]=Y$, then $u_{q}=y$.


If $u$ is good then $u$ is linear in all variables $y$ such that $[y]=Y$: if $y$ is one of such variables and $u$ has a subterm $v$ of the form $\lambda x_{1}^{C_{1}}.\dots.\lambda x_{p}^{C_{p}}.zu_{1}\dots u_{i-1}yu_{i+1}\dots u_{q}$, then $[z]= [u_{1}]\To \dots \To [u_{i-1}]\To [Y]\To [u_{i+1}]\To \dots \To [u_{q}]\To Z$ forces $q=i$ and $Z=X$, hence $v=v_{C}$ for some positive subtype of $A$ of the form $D_{1}\To \dots \To D_{p-1}\To X$, $C_{j}=D_{j}[B/X]$ for $1\leq j\leq p-1$, $C_{p}=Y$ and $x_{p}=y$.

Observe that we can always transform a normal $\eta$-long term $u:\Gamma\to A[B/X]$ into a good one by replacing $u_{q}$ by $y$ in any subterm $v_{C}$ of the appropriate form.

We can now prove the analogous of lemma \ref{trave0}:

\begin{lemma}\label{travexp}
If $A\in \CC L_{\To}$ and there is a term $u:A\to A[B/X]$, where $B=Y\To X$, with $Y\notin FV(A),FV(\Gamma)$, then either $A$ is derivable or $A$ is $gRP_{X}$ or $u$ has $gRP_{X}$ interpolants.
\end{lemma}
\begin{proof}
We can suppose w.l.o.g. that $u$ is good.
Then, for some $d\in \BB N$, 
$\partial u:A_{1}^{-}, \dots, A^{-}_{d}\to  C $ in $\CC A^{\multimap}$, where $C\in \CC E^{+}(A[B/X])$. Suppose $A$ is not derivable nor $gRP_{X}$; then $d\neq 0$ and, from the fact that $u$ is linear in all $y$ such that $[y]=Y$ we deduce that $C$ is of the form $ A^{+}[Y\multimap X/X]$, for some $A^{+}\in \CC E^{+}(A)$. Indeed, $C$ is obtained by replacing, in some type $A^{+}\in \CC L^{+}(A)$, all negative occurrences of $X$ by $Y\multimap X$ and all positive occurrences of $X$ by $\underbrace{Y\multimap \dots \multimap Y}_{p\text{ times}}\multimap X$, for some $p\in \BB N$ corresponding to the number of occurrences in $u$ of some $y$ such that $[y]=Y$.

Hence $\partial u: A^{-}_{1},\dots, A^{-}_{d}\to A^{+}[Y\multimap X/X]$, where $Y\notin FV(A_{j}^{-})$ and $Y\notin FV(A^{+})$ and we can apply lemma \ref{trave0}: 
either $A^{+}\in gRP_{X}$ or $\partial u$ has $gRP_{X}$ interpolants.

By proposition \ref{enest} $b.$, $A^{+}\notin gRP_{X}$. We conclude then that $\partial u$ has $gRP_{X}$ interpolants $C_{1},\dots, C_{p}$, whence, by proposition \ref{enest} $a$., $u$ has $gRP_{X}$ interpolants $C_{1}^{\To},\dots, C_{p}^{\To}$.

\end{proof}

We finally get:

\begin{proposition}\label{expoexpa2}
$A\in \CC L_{\To}$ is expansible iff either $A$ is derivable, $A$ is $gRP_{X}$ or $A$ collapses into a finite family of $gRP_{X}$ types.
\end{proposition}

From theorem \ref{expoexpa2} we obtain immediately the following:
\begin{theorem}\label{expoexpa}
A simple type is expansible iff it is either derivable or logically equivalent to a product of $gRP_{X}$ simple types.
\end{theorem}

Observe that, differently from the case of linear expansible terms, derivable types are always expansible: if there is a closed term $u$ of type $A$, then one can form a closed term $u'$ of type $A[B/X]$, for all $B\in \CC L_{\To}$, hence $\lambda x^{A}.u:A\to A[B/X] $, for all $x\notin FV(A)$. That this does not hold for linear expansible term is shown by the linear type $C=(X\multimap Y)\multimap (X\multimap Y)$ which is derivable but has a unique edge over $X$, which is a jump out. As $C$ is clearly $\prec_{I}$-minimal, by proposition \ref{jumpout}, $C$ is not expansible.

\subsection{Characterization of instantation overflow for simple types}\label{9}

%
%
%
%

In this section we introduce the $F_{at}$-expansion property and we prove (1) that a $F_{at}$-expansible simple type $A$ is either derivable or such that $\forall XA$ is equivalent in $F_{at}$ to the conjunction of a family of $gRP$ types and (2) that for a simple type $A$, instantiation overflow for $\forall XA$ coincides with the $F_{at}$-expansibility of $A$. 

Let $A\in \CC L_{\To}$ and $u:\forall XA\to \forall XA$ in $\CC F_{at}$. By deleting second order constructs, we obtain an arrow $u_{0}:A[Z_{1}/X],\dots, A[Z_{n}/X]\to A$ in $\CC T$, for some variables $Z_{1},\dots, Z_{n}$. We first show that we can suppose $Z_{1},\dots, Z_{n}\in FV(A)$:

\begin{lemma}\label{interpofat}
Let $A\in \CC L_{\To}$ and $u:\forall XA\to \forall XA$ in $\CC F_{at}$. Then there exists 
$u^{*}:A[Y_{1}/X],\dots, A[Y_{n}/X]\to A$ in $\CC T$, where $FV(A)=\{Y_{1},\dots, Y_{n}\}$.

\end{lemma}
\begin{proof}
We can suppose w.l.o.g. $u$ normal. By deleting second order constructs, we obtain a term $u_{0}$ such that $y_{1}:A[Y_{1}/X],\dots, y_{n}:A[Y_{n}/X], z_{1}:A[Z_{1}/X],\dots, z_{p}:A[Z_{p}/X]\vdash u_{0}: A$ is derivable in $\lambda_{\To}$, where the variables $Z_{1},\dots, Z_{p}$ do not occur in $A$. If we consider the graph $\CC G(\partial u)$, it is clear that all edges over $Z_{j}$, for $1\leq j\leq p$, are of type II. By renaming all such edges as $Y_{1}$, we obtain the graph $\CC G(\partial u')$ of a term $u'$ such that $y_{1}:A[Y_{1}/X],\dots, y_{n}:A[Y_{n}/X], z_{1}:A[Y_{1}/X],\dots, z_{p}:A[Y_{1}/X]\vdash u_{0}: A$. We can thus put $u^{*}=u'[y_{1}/z_{1},\dots, y_{1}/z_{p}]$.

\end{proof}

From lemma \ref{interpofat} we deduce that the interpolants of an arrow $u:\forall XA\to \forall XA$ in $\CC F_{at}$ are formulas whose free variables are included in those of $A$. This leads to the following definition:

\begin{definition}[$F_{at}$-collapse]
Let $A,B_{1},\dots, B_{p}\in \CC L_{\To}$. $A$ \emph{$F_{at}$-collapses into} $B_{1},\dots, B_{p}$ if $FV(B_{j})\subseteq FV(A)$ for all $1\leq j\leq p$ and there exist arrows $u_{j}:A[Y_{1}/X],\dots, A[Y_{n}/X]\to B_{j}$ and 
$u:B_{1},\dots, B_{p}\to A$, where $FV(A)=\{Y_{1},\dots, Y_{n}\}$.


\end{definition}
If $A$ $F_{at}$-collapses into $B_{1},\dots, B_{p}$, then $\forall XA$ is logically equivalent to the product of the $\forall X B_{1},\dots, \forall X B_{p}$. 
$F_{at}$-collapse is not equivalent to collapse: for instance the type $(Y\To X)\To X$ $F_{at}$-collapses into $Y$, but $(Y\To X)\To X$ does not collapse into $Y$.

We introduce $F_{at}$-expansible types:

\begin{definition}[$F_{at}$-expansible type]
A type $A\in \CC L_{\To}$ is \emph{$F_{at}$-expansible} if for every types $C_{1},\dots, C_{p}\in \CC L_{\To}$ there exists an arrow $EXP_{A}(B): A[Y_{1}/X],\dots, A[Y_{n}/x]\to A[B/X]$, where $B=C_{1}\To \dots \To C_{p}\To X$ and $FV(A)=\{Y_{1},\dots, Y_{n}\}$.

\end{definition}

Similarly to lemma \ref{travexp}, we can characterize $F_{at}$-expansible types as those which $F_{at}$-collapse into $gRP_{X}$ types:

\begin{proposition}\label{travefor}
$A\in \CC L_{\To}$ is $F_{at}$-expansible iff either $A$ is derivable or $A$ is $gRP_{X}$ or $A$ $F_{at}$-collapses into a family of $gRP_{X}$ types.
\end{proposition}
\begin{proof}

We can argue similarly to lemma \ref{travexp}. Suppose $A$ is neither derivable nor $gRP_{X}$, let $Y\notin FV(A)$ and suppose there is an arrow $u:A[Y_{1}/X], \dots, A[Y_{n}/X]\to A[Y\To X/X]$, where $FV(A)=\{Y_{1},\dots, Y_{n}\}$.  Then we obtain $gRP_{X}$ interpolants $A_{1},\dots, A_{p}$ of $u$, into which $A$ $F_{at}$-collapses.

%
%
\end{proof}


Similarly to the previous subsection, if $\Gamma,\Delta,\Sigma$ are contexts, $Y$ a variable not occurring in any of them nor in a simple type $A$ and $u:\Gamma,\Delta[Y\To X/X], \Sigma[Y/X] \to A[B/X]$ in $\CC T$ is normal and $\eta$-long, where $B=Y\To X$, we can consider all  $v_{C}\in Subt(u)$ such that $[v_{C}]=C[B/X]$, where $C$ is a type of the form $C_{1}\To \dots \To C_{p}\To X$ occurring positively in $A$. Then $v_{C}$ is of the form
$\lambda x_{1}^{C_{1}[B/X]}.\dots.\lambda x_{p}^{C_{p}[B/X]}.\lambda y^{Y}.zu_{1}\dots u_{q} $. 
We say that $u$ is \emph{good} when for all $v_{C}\in Subt(u)$ as above, whenever $q\geq 1$ and $[u_{q}]=Y$, then $u_{q}=y$. We will suppose that all such arrows are good.

To prove our final result, i.e. the equivalence of instantiation overflow and the $F_{at}$-expansion property, we need the lemma below.

\begin{lemma}\label{svorta}
Let $\Gamma, \Delta, \Sigma$ be contexts, $A\in \CC L_{\To}$ and $Y$ be a variable not occurring in $\Gamma,\Delta,\Sigma,A$.
\begin{description}
\item[$a.$] for any arrow $u: \Gamma, \Delta[Y\To X/X], \Sigma[Y/X]\to A$ there exists an arrow
$u^{*}: \Gamma, \Delta[Y\To X/X], \Sigma\to A$;
\item[$b.$] for any arrow $u: \Gamma, \Delta[Y\To X/X], \Sigma[Y/X]\to A[Y\To X/X]$ there exists a arrow
$u^{*}: \Gamma, \Delta[Y\To X/X], \Sigma\to A[Y\To X/X]$;
\item[$c.$] for any arrow $u: \Gamma, \Delta[Y\To X/X], \Sigma[Y/X]\to A[Y/X]$ there exists an arrow
$u^{*}: \Gamma, \Delta[Y\To X/X], \Sigma\to A$.

\end{description}
\end{lemma}
\begin{proof}
For any variable $y$ and context $\Sigma$, we let $y\in \Sigma$ indicate that the variable $y$ is declared in $\Sigma$.
We argue by induction on the number of applications of $u$. If $u=\lambda x_{1}^{A_{1}}.\dots.\lambda x_{p}^{A_{p}}.y$, then
\begin{description}
\item[$a.$] if $y\in \Gamma$, $u^{*}=u$; if $y\in \Delta[Y\To X/X]$ or $y\in \Sigma[Y/X]$, then $X\notin FV([y])$, so $u^{*}=u$;

\item[$b.$] if $y\in \Delta[Y\To X/X] $, $u^{*}=u$; if $y\in \Gamma$ or $y\in \Sigma[Y/X]$, then $X\notin FV([y])$, so $u^{*}=u$;

\item[$c.$] if $y\in \Sigma[Y/X] $, $u^{*}=u$; if $y\in \Gamma$ or $y\in \Delta[Y\To X/X]$, then $X\notin FV([y])$, so $u^{*}=u$.

\end{description}

If $u=\lambda x_{1}^{A_{1}}.\dots.\lambda x_{p}^{A_{p}}.yu_{1}\dots u_{q}$, for some $q\geq 1$, then we must consider nine cases, three for each case $a,b,c$. We only consider case $a1-a3$. and case $b2$. as all other cases can be treated similalry:

 \begin{itemize}
\item[$a1.$] if $y\in \Gamma$ or $y=x_{i}$, for some $1\leq i\leq p$, then $[u_{j}]= C_{j}$ has no occurrence of $Y$, for $1\leq j\leq q$; by the induction hypothesis $a.$, we deduce $u_{j}^{*}:\Gamma, A_{1},\dots, A_{p},\Delta[Y\To X/X], \Sigma\to C_{j}$, so we put
$u^{*} =\lambda x_{1}^{A_{1}}.\dots.\lambda x_{p}^{A_{p}}.yu_{1}^{*}\dots u_{p}^{*}$;

\item[$a2.$] if $y\in \Delta[Y\To X/X]$, then $[y]=C[Y\To X/X]$ for some $C=C_{1}\To \dots \To C_{q'}\To Z$. If $Z\neq X$, then 
 $[u_{j}]=C_{j}[Y\To X/X]$, for $1\leq j\leq q$; by the induction hypothesis $b.$, we deduce $u_{j}^{*}:\Gamma,A_{1},\dots, A_{p},\Delta[Y\To X/X], \Sigma\to C_{j}[Y\To X/X]$, so we put
$u^{*} =\lambda x_{1}^{A_{1}}.\dots.\lambda x_{p}^{A_{p}}.yu_{1}^{*}\dots u_{p}^{*}$.
If $Z= X$, then $q=q'+1$, 
 $[u_{j}]=C_{j}[Y\To X/X]$, for $1\leq j\leq q'$ and $[u_{q}]=Y$; by the induction hypothesis $c.$ there exists then an arrow $u_{q}^{*}:\Gamma,A_{1},\dots, A_{p}, \Delta[Y\To X/X],\Sigma\to X$, so we put $u^{*} =\lambda x_{1}^{A_{1}}.\dots.\lambda x_{p}^{A_{p}}.u_{q}^{*}$;
 
%

\item[$a3.$] if $y\in \Sigma[Y/X]$, then $[u_{j}]=C_{j}[Y/X]$, for $1\leq j\leq q$; by the induction hypothesis $c.$, we deduce $u_{j}^{*}:\Gamma,A_{1},\dots, A_{p},\Delta[Y\To X/X], \Sigma\to C_{j}$, so we put
$u^{*} =\lambda x_{1}^{A_{1}}.\dots.\lambda x_{p}^{A_{p}}.yu_{1}^{*}\dots u_{p}^{*}$;

\item[$b2.$] if $y\in \Delta[Y\To X/X]$, then $[y]=C[Y\To X/X]$ for some $C=C_{1}\To \dots \To C_{q'}\To Z$. If $Z\neq X$, then 
 $[u_{j}]=C_{j}[Y\To X/X]$, for $1\leq j\leq q$; by the induction hypothesis $b.$, we deduce $u_{j}^{*}:\Gamma,\Delta[Y\To X/X],A_{1},\dots, A_{p}, \Sigma\to C_{j}[Y\To X/X]$, so we put
$u^{*} =\lambda x_{1}^{A_{1}}.\dots.\lambda x_{p}^{A_{p}}.yu_{1}^{*}\dots u_{p}^{*}$.
If $Z= X$, then $q=q'+1$, $A_{p}=Y$, $u_{q}=x_{p}$ (as we supposed $u$ is good) and for $1\leq j\leq 1'$,  
 $[u_{j}]=C_{j}[Y\To X/X]$, and $x_{p}\notin FV(u_{j})$;  by the induction hypothesis $b.$ we deduce then $u_{j}^{*}:\Gamma,\Delta[Y\To X/X], A_{1},\dots, A_{p-1},\Sigma\to C_{j}[Y\To X/X]$, so we can put
$u^{*} =\lambda x_{1}^{A_{1}}.\dots.\lambda x_{p}^{A_{p}}.yu_{1}^{*}\dots u_{q}^{*}x_{p}$.

\end{itemize}

\end{proof}

From lemma \ref{svorta} we deduce that $IO$ is equivalent to the $F_{at}$-expansion property:

\begin{proposition}\label{IO2}
For all $A\in \CC L_{\To}$, $\forall XA$ has $IO$ iff $A$ is $F_{at}$-expansible.

\end{proposition}
\begin{proof}

For one direction, for all $B\in \CC L_{\To,\forall}$, we can define $IO_{A}(B)[x]=EXP_{A}(B)[xZ/x]$, where $Z$ is the rightmost variable of $B$.

%

For the converse direction, let $Y\notin FV(A)$ and $IO_{A}(Y\To X): \forall XA\to A[Y\To X/X]$. By deleting second order constructs in $IO_{A}(Y\To X)$ we obtain an arrow $u: A[Z_{1}/X_{1}],\dots, A[Z_{p}/X_{p}]$, $A[Y_{1}/X_{1}],\dots, A[Y_{q}/X_{q}]\to A[Y\To X/X]$, where $Z_{1},\dots, Z_{p}\in FV(A)$ and $Y_{1},\dots, Y_{q}\notin FV(A)$. By reasoning as in the proof of lemma \ref{interpofat} we can suppose w.l.o.g. $q=1$ and $Y_{1}=Y$, i.e.  $u: A[Z_{1}/X], \dots,  A[Z_{r}/X], A[Y/X]\to A[Y\To X/X]$. By applying lemma \ref{svorta} $b.$ and possibly identifying some variables we obtain an arrow $u^{*}:A[Z_{1}/X],\dots, A[Z_{n}/X], A \to A[Y\To X/X]$. Now we can argue as in proposition \ref{travefor}: $A$ is either derivable, or $gRP_{X}$ or $F_{at}$-collapses to some $gRP_{X}$ types. In all such cases $A$ is $F_{at}$-expansible.

\end{proof}

We then finally get:
\begin{theorem}\label{IO}
Let $A$ be a simple type. Then $\forall XA$ has instantiation overflow iff it is either derivable or logically equivalent to a product of $gRP$ types.
\end{theorem}

\section{Some open questions}\label{sec9}

By exploiting ideas coming from linear logic and functorial polymorphism, we provided the first general investigation of the instantiation overflow phenomenon, providing a characterization of the types $\forall XA$, with $A$ a simple type, satisfying this property, as well as a characterization of the linear and simple types satisfying the related expansion property. As it can be expected, there are many questions which naturally arise and are left open by the results contained in this paper. In the following lines we mention some of them.

First, it is not clear how to extend our characterization of instantiation overflow in terms of Russell-Prawitz types to all System $F$ types, as the following example shows: let $E$ be the second order type 
$\forall X(X\To \forall ZZ\To D)$, where $X\notin FV(D)$; then we can define, for $B=\forall \OV Y_{1} (B_{1}\To \dots \To \forall Y_{p}(B_{p}\To \forall Y_{p+1}Z)\dots )$, an expansion term $IO_{D}(B):\forall XE\to E[B/X]$ in $\CC F_{at}$ as follows
$$IO_{D}(B)= \lambda y^{B}. \lambda z^{\forall ZZ}.
xZ Elim_{B}\big [ EXP_{Z}(B_{1})[z/x]/x_{1},\dots, EXP_{Z}(B_{p})[z/x]/x_{p}\big ] z$$
Hence $E$ has instantiation overflow, though it does not seem to be logically equivalent to any $gRP_{X}$ type. 

A second important question is about the decidability of the instantiation overflow property. For a simple type $A$, to decide whether $\forall XA$ enjoys $IO$ one must check (1) whether $A$ is derivable, (2) whether $A$ is $gRP_{X}$ and (3) whether $\forall XA$ is logically equivalent to (better, $F_{at}$-collapses into) the product of a family of $gRP_{X}$ types. Problems (1) and (2) are surely decidable for the restricted case considered in this paper (the property of being $gRP_X$ for a simple type can be decided by checking its tilings). Problem (3) is surely decidable in the case of linear types, as one must only consider types with smaller logical complexity than $A$ and variables included in those of $A$. The decidability of (3) remains open in the case of simple types and $F_{at}$.

In the case of full System $F$ types, the instantiation overflow property is most likely to be undecidable, as (1) is undecidable in System $F$. Worse, (1) remains undecidable even if one restricts derivability to $F_{at}$, as this system, in spite of its weak expressive power, is undecidable. This follows from the fact that $F_{at}$ is equivalent to $\mathit{IFOL_{mon}^{1}}$ the $\To,\forall$-fragment of \emph{first-order} intuitionistic logic
over a language with one monadic predicate $p(x)$ and no function symbol, which is known to be undecidable (see \cite{Gabbay1981}, Th. 1, p. 234). The equivalence results from the bijective translation below
$$
X_{i}^{*} \ = \  p(x_{i}) \qquad (A\To B)^{*} \ = \ A^{*}\To B^{*} \qquad (\forall X_{i}A)^{*} \ = \ \forall x_{i}A^{*}
$$  
which can be extended into a bijective translation of the rules of $F_{at}$ and $\mathit{IFOL_{mon}^{1}}$.
%


Another natural question concerns the relationship between $gRP$ types and logical connectives: if $RP$ types correspond to the translation of logical connectives, what about $gRP$ types? 
Proposition \ref{renaming2} shows that $gRP_X$ types can be seen as $qRP_{\CC X}$ types whose variables have been identified. This suggests that $gRP$ types can be seen as types translating connectives obtained by composing different basic connectives. For instance, the $gRP_{X}$ type $((((A\To B\To X)\To X)\To X)\To (C\To X))\To X$ can be renamed as a $((((A\To B\To X_{2})\To X_{2})\To (C\To X_{1}))\To X_{1}$, which is $qRP_{\{X_{1},X_{2}\}}$ and yields a translation of the composed connective $(A\land B)\lor C$.

Finally, the problem of extending proposition \ref{dinatural}, i.e. the equivalence of instantiation overflow and full extraction modulo dinaturality, to general Russell-Prawitz types should be considered, as the result does not seem to scale in a striaghtforward way to $gRP$ types.

\bibliographystyle{alpha}
\bibliography{Victims}

\appendix

\section{Weak interpolation for $\lambda_{\multimap}$ nets}\label{appA}

We show that every weak interpolation problem (definition \ref{interpo2}) admits a solution in $\lambda_{\multimap}$. The arguments adapts the proof in \cite{deGroote1996} of interpolation for $IMLL^{-}$ nets.

We first reformulate the splitting lemma for essential nets (see \cite{Ong}):
\begin{lemma}[splitting]\label{splitting}
Let $f:\Gamma\to A$ in $\CC A^{\multimap}$, where $\Gamma=A_{1},\dots, A_{m}$ and $A=A_{m+1}\multimap \dots \multimap A_{m+n}\multimap Z$. 
Then, for some $1\leq i\leq m+n$, where
$A_{i}=B_{1}\multimap\dots \multimap B_{p}\to Z$,
 there exists a partition $\D P=\{ \F d_{1},\dots, \F d_{p}\}$  of $\{1,\dots, i-1,i+1,\dots,m+n\}$ in $p$ sets $\F d_{1}=\{ i_{11},\dots, i_{1k_{1}}\},\dots, \F d_{p}=\{i_{p1},\dots, i_{pk_{p}}\}$, where $k_{j}$ denotes the cardinality of the $j$-th element of the partition, such that $f=g_{1}\cup \dots \cup g_{p}\cup \{( Z^{-},Z^{+}\}$, where 
 $g_{j}: \Gamma_{j}\to B_{j}$ is a correct net, for $\Gamma_{j}=A_{i_{j1}},\dots, A_{i_{jk_{j}}}$, for $1\leq j\leq p$.

\end{lemma}

A graph $f:\Gamma\to A$ with $n$ splitting cuts (see subsection \ref{sec32}) is \emph{reduced} when all cuts are positive and any $\Gamma_{i}$ is of the form $\Delta_{i},cut$, i.e. with a unique cut link.  

Let $f:\Gamma\to A$ in $\CC A$.
Given a directed path $\gamma$ in the correction graph and a type $B$ occurring in either $A$ or $\Gamma$, we say the $\gamma$ \emph{crosses} $B$ if $\gamma$ passes through some node in $vS_{B}$.

We prove the following \emph{splitting property} of correct graphs:
\begin{proposition}[splitting property]\label{splitprop}
Let $f:\Gamma\to A$ in $\CC A^{\multimap}$, where $\Gamma=A_{1},\dots, A_{m}$ and $A=A_{m+1}\multimap \dots \multimap A_{m+n}\multimap Z$. 
Let $n_{1},n_{2}$ be two nodes in the correction graph of $f$ such that there exist two paths from the conclusion to $n_{1}$ and to $n_{2}$, respectively. Then, for no path $\gamma_{1}$ starting from $n_{1}$ and path $\gamma_{2}$ starting from $n_{2}$, $\gamma_{1}$ and $\gamma_{2}$ cross the same type $A_{i}$, for all $1\leq i\leq m+n$.\end{proposition}
\begin{proof}
By induction on the number $k$ of nodes in the correction graph. If $k=2$, then the claim is trivially true. 
Otherwise, by the splitting lemma \ref{splitting}, $f$ splits into $f_{1}:\Gamma_{1}\to B_{1},\dots, f_{p}:\Gamma_{p}\to B_{p}$, where, for some $1\leq i\leq m+n$, $A_{i}=B_{1}\multimap \dots \multimap B_{p}\multimap Z$. Hence, the path starting from $A$ reaches $Z^{+}$, then passes through a negative occurrence of $A_{i}$ and splits at all the $\multimap^{-}$-nodes in it. If $n_{1}$ and $n_{2}$ belong to the correction graphs of $f_{i},f_{j}$ for $i\neq j$, then we are done. If $n_{1},n_{2}$ belong to the correction graphs of the same $f_{i}$, then we can apply the induction hypothesis as all paths $\gamma_{1},\gamma_{2}$ belong to the correction graph of $f_{i}$ and $f_{j}$, respectively.
\end{proof}

The following lemma is the analogous of lemmas 7.1 and 7.2 in \cite{deGroote1996} and provides an algorithm to construct the weak interpolants of an allowable graph.

\begin{lemma}\label{algo}
Any graph with splitting cuts can be transformed into a reduced one.
\end{lemma}
\begin{proof}
We define an algorithm to transform a graph with splitting cuts into a reduced one. 
\begin{enumerate}
\item For any $cut^{-}$, check if there is some $cut^{+}$ which is connected to $cut^{-}$ by a path which never gets out of the subtree of $A$:
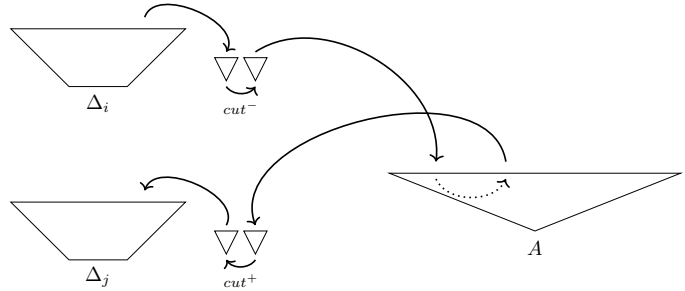
$$
\resizebox{0.6\textwidth}{!}{\begin{tikzpicture}

\draw[] (0,0) to (1,0) to (2,1) to (-1,1) -- cycle;
\draw[] (8,-2.5) to (10.5,-1.5) to (5.5,-1.5) -- cycle;

\node(g) at (0.5,-0.3) {$\Delta_{i}$};
\node(a) at (8,-2.8) {$A$};

\draw (2.5,0.5) to (2.9,0.5) to (2.7,0.1) -- cycle;
\draw (3,0.5) to (3.4,0.5) to (3.2,0.1) -- cycle;
\draw[->, thick] (2.7,0) to [bend right=55] node[below] {\scriptsize$cut^{-}$} (3.2,0);

\draw[->, thick] (1.3,1.2) to [bend left=95] (2.7,0.6);

\draw[->, thick] (3.2,0.6) to [bend left=65] (6.3,-1.3);
\draw[->, thick, dotted] (6.3,-1.6) to [bend right=65] (7.5,-1.6);

\draw[] (0,-3) to (1,-3) to (2,-2) to (-1,-2) -- cycle;

\node(g1) at (0.5,-3.3) {$\Delta_{j}$};

\draw (2.5,-2.5) to (2.9,-2.5) to (2.7,-2.9) -- cycle;
\draw (3,-2.5) to (3.4,-2.5) to (3.2,-2.9) -- cycle;
\draw[<-,thick] (2.7,-3) to [bend right=55] node[below] {\scriptsize$cut^{+}$} (3.2,-3);

\draw[<-, thick] (1.3,-1.8) to [bend left=95] (2.7,-2.4);

\draw[<-, thick] (3.2,-2.4) to [bend left=95] (7.5,-1.3);

\end{tikzpicture}}
$$

In this case then transform the graph as follows:

$$
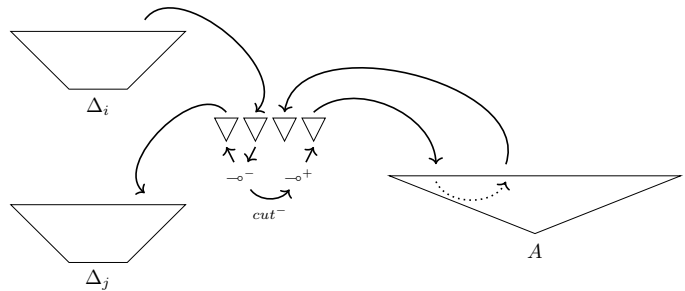
\resizebox{0.6\textwidth}{!}{\begin{tikzpicture}

\draw[] (0,0) to (1,0) to (2,1) to (-1,1) -- cycle;
\draw[] (8,-2.5) to (10.5,-1.5) to (5.5,-1.5) -- cycle;

\node(g) at (0.5,-0.3) {$\Delta_{i}$};
\node(a) at (8,-2.8) {$A$};

\draw (2.5,-0.5) to (2.9,-0.5) to (2.7,-0.9) -- cycle;
\draw (3,-0.5) to (3.4,-0.5) to (3.2,-0.9) -- cycle;
\node(t) at (2.95,-1.5) {\scriptsize$\multimap^{-}$};

\draw (3.5,-0.5) to (3.9,-0.5) to (3.7,-0.9) -- cycle;
\draw (4,-0.5) to (4.4,-0.5) to (4.2,-0.9) -- cycle;
\node(t1) at (3.95,-1.5) {\scriptsize$\multimap^{+}$};

\draw[<-, thick] (t1) to [bend left=55] node[below] {\scriptsize$cut^{-}$} (t);
\draw[->, thick] (t1) to (4.2,-1);
\draw[->, thick] (3.2,-1) to (t);
\draw[->, thick] (t) to (2.7,-1);

\draw[->, thick] (1.3,1.2) to [bend left=95] (3.2,-0.4);

\draw[->, thick] (4.2,-0.4) to [bend left=65] (6.3,-1.3);
\draw[->, thick, dotted] (6.3,-1.6) to [bend right=65] (7.5,-1.6);

\draw[] (0,-3) to (1,-3) to (2,-2) to (-1,-2) -- cycle;

\node(g1) at (0.5,-3.3) {$\Delta_{j}$};

\draw[<-, thick] (1.3,-1.8) to [bend left=95] (2.7,-0.4);

\draw[<-, thick] (3.7,-0.4) to [bend left=95] (7.5,-1.3);

\end{tikzpicture}}
$$
It can be easily verified that the new correction graph is still acyclic and functional. 
Once all such transformations are done or if there is no negative cut, go to step 2.

\item Choose one $cut^{-}$. If there is none go to step 3. Then there is a $cut^{+}$ which is connected to $cut^{-}$ by a path which never gets out of the part of the graph over $\Delta_{i}$: 
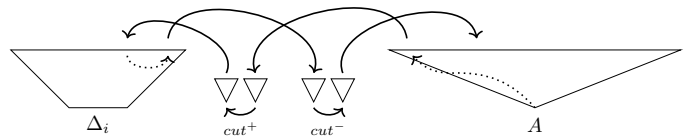
$$
\resizebox{0.6\textwidth}{!}{\begin{tikzpicture}

\draw[] (0,0) to (1,0) to (2,1) to (-1,1) -- cycle;
\draw[] (8,0) to (10.5,1) to (5.5,1) -- cycle;

\node(g) at (0.5,-0.3) {$\Delta_{i}$};
\node(a) at (8,-0.3) {$A$};

\draw (2.5,0.5) to (2.9,0.5) to (2.7,0.1) -- cycle;
\draw (3,0.5) to (3.4,0.5) to (3.2,0.1) -- cycle;
\draw[<-, thick] (2.7,0) to [bend right=55] node[below] {\scriptsize$cut^{+}$} (3.2,0);

\draw (4,0.5) to (4.4,0.5) to (4.2,0.1) -- cycle;
\draw (4.5,0.5) to (4.9,0.5) to (4.7,0.1) -- cycle;
\draw[->, thick] (4.2,0) to [bend right=55] node[below] {\scriptsize$cut^{-}$} (4.7,0);

\draw[<-,thick] (1,1.2) to [bend left=95] (2.7,0.6);
\draw[->, thick] (1.7,1.2) to [bend left=95] (4.2,0.6);

\draw[<-, thick] (3.2,0.6) to [bend left=95] (5.8,1.2);
\draw[->, thick] (4.7,0.6) to [bend left=95] (7,1.2);

\draw[->, thick, dotted] (1,0.9) to [bend right=85] (1.7,0.9);
\draw[->, thick, dotted] (8,0) to [bend right=25] (6.6,0.6) to [bend left=25] (5.9,0.9);

\end{tikzpicture}}
$$
This follows from the existence of a unique path from $A$ to $cut^{-}$. In this case then transform the graph as follows:
$$
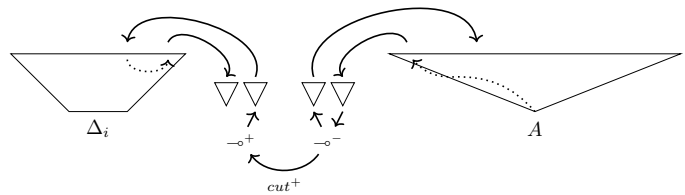
\resizebox{0.6\textwidth}{!}{\begin{tikzpicture}

\draw[] (0,0) to (1,0) to (2,1) to (-1,1) -- cycle;
\draw[] (8,0) to (10.5,1) to (5.5,1) -- cycle;

\node(g) at (0.5,-0.3) {$\Delta_{i}$};
\node(a) at (8,-0.3) {$A$};

\draw (2.5,0.5) to (2.9,0.5) to (2.7,0.1) -- cycle;
\draw (3,0.5) to (3.4,0.5) to (3.2,0.1) -- cycle;

\node(p) at (2.95,-0.5) {\scriptsize$\multimap^{+}$};

\draw (4,0.5) to (4.4,0.5) to (4.2,0.1) -- cycle;
\draw (4.5,0.5) to (4.9,0.5) to (4.7,0.1) -- cycle;

\node(p1) at (4.45,-0.5) {\scriptsize$\multimap^{-}$};

\draw[<-, thick] (p) to [bend right=55] node[below] {\scriptsize$cut^{+}$} (p1);

\draw[->, thick] (p) to (3.2,0); 
\draw[->, thick] (p1) to (4.2,0); 
\draw[->, thick] (4.7,0) to (p1);

\draw[<-,thick] (1,1.2) to [bend left=95] (3.2,0.6); 
\draw[->, thick] (1.7,1.2) to [bend left=95] (2.7,0.6); 

\draw[<-, thick] (4.7,0.6) to [bend left=95] (5.8,1.2); 
\draw[->, thick] (4.2,0.6) to [bend left=95] (7,1.2); 

\draw[->, thick, dotted] (1,0.9) to [bend right=85] (1.7,0.9);
\draw[->, thick, dotted] (8,0) to [bend right=25] (6.6,0.6) to [bend left=25] (5.9,0.9);

\end{tikzpicture}}
$$
It can be easily verified that the new correction graph is still acyclic and functional. 
Observe that the number of negative cuts decreases by one. If there is still one negative cut go back to step 1. Otherwise go to step 3.

\item Now all cuts are positive. If the graph is not yet reduced, then there exists two positive cuts as below:
$$
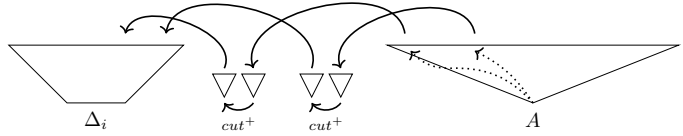
\resizebox{0.6\textwidth}{!}{\begin{tikzpicture}

\draw[] (0,0) to (1,0) to (2,1) to (-1,1) -- cycle;
\draw[] (8,0) to (10.5,1) to (5.5,1) -- cycle;

\node(g) at (0.5,-0.3) {$\Delta_{i}$};
\node(a) at (8,-0.3) {$A$};

\draw (2.5,0.5) to (2.9,0.5) to (2.7,0.1) -- cycle;
\draw (3,0.5) to (3.4,0.5) to (3.2,0.1) -- cycle;
\draw[<-, thick] (2.7,0) to [bend right=55] node[below] {\scriptsize$cut^{+}$} (3.2,0);

\draw (4,0.5) to (4.4,0.5) to (4.2,0.1) -- cycle;
\draw (4.5,0.5) to (4.9,0.5) to (4.7,0.1) -- cycle;
\draw[<-, thick] (4.2,0) to [bend right=55] node[below] {\scriptsize$cut^{+}$} (4.7,0);

\draw[<-,thick] (1,1.2) to [bend left=95] (2.7,0.6);
\draw[<-, thick] (1.7,1.2) to [bend left=95] (4.2,0.6);

\draw[<-, thick] (3.2,0.6) to [bend left=95] (5.8,1.2);
\draw[<-, thick] (4.7,0.6) to [bend left=95] (7,1.2);

\draw[->, thick, dotted] (8,0) to [bend right=15] (7,0.9);
\draw[->, thick, dotted] (8,0) to [bend right=25] (6.6,0.6) to [bend left=25] (5.9,0.9);

\end{tikzpicture}}
$$
By the splitting property, $\Delta_{i}$ splits then into two contexts $\Delta_{i}^{1}$ and $\Delta_{i}^{2}$:
$$
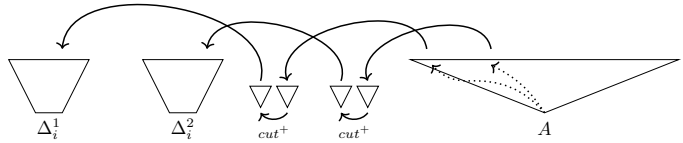
\resizebox{0.6\textwidth}{!}{\begin{tikzpicture}

\draw[] (1,0) to (1.5,0) to (2,1) to (0.5,1) -- cycle;
\draw[] (8,0) to (10.5,1) to (5.5,1) -- cycle;
\draw[] (-1.5,0) to (-1,0) to (-0.5,1) to (-2,1) -- cycle;

\node(g) at (1.25,-0.3) {$\Delta_{i}^{2}$};
\node(a) at (8,-0.3) {$A$};
\node(g) at (-1.25,-0.3) {$\Delta_{i}^{1}$};

\draw (2.5,0.5) to (2.9,0.5) to (2.7,0.1) -- cycle;
\draw (3,0.5) to (3.4,0.5) to (3.2,0.1) -- cycle;
\draw[<-, thick] (2.7,0) to [bend right=55] node[below] {\scriptsize$cut^{+}$} (3.2,0);

\draw (4,0.5) to (4.4,0.5) to (4.2,0.1) -- cycle;
\draw (4.5,0.5) to (4.9,0.5) to (4.7,0.1) -- cycle;
\draw[<-, thick] (4.2,0) to [bend right=55] node[below] {\scriptsize$cut^{+}$} (4.7,0);

\draw[<-,thick] (-1,1.2) to [bend left=95] (2.7,0.6);
\draw[<-, thick] (1.7,1.2) to [bend left=95] (4.2,0.6);

\draw[<-, thick] (3.2,0.6) to [bend left=95] (5.8,1.2);
\draw[<-, thick] (4.7,0.6) to [bend left=95] (7,1.2);

\draw[->, thick, dotted] (8,0) to [bend right=15] (7,0.9);
\draw[->, thick, dotted] (8,0) to [bend right=25] (6.6,0.6) to [bend left=25] (5.9,0.9);

\end{tikzpicture}}
$$
We can conclude that the graph is now reduced.
\end{enumerate}

\end{proof}

We can now prove the weak interpolation theorem.

\begin{proof}[Proof of theorem \ref{interpolation2}]
Given $f:\Gamma\to A$, transform the graph in \ref{fig1} into a graph with cuts, by replacing type I edges by either positive or negative cut links:
$$
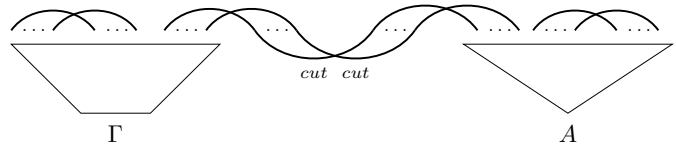
\resizebox{0.6\textwidth}{!}{
\begin{tikzpicture}

\draw[] (0,0) to (1,0) to (2,1) to (-1,1) -- cycle;
\draw[] (7,0) to (8.5,1) to (5.5,1) -- cycle;

\node(g) at (0.5,-0.3) {$\Gamma$};
\node(a) at (7,-0.3) {$A$};

\draw[thick] (-1,1.2) to [bend left=55] (0.2,1.2);
\draw[thick] (-0.4,1.2) to [bend left=55] (0.8,1.2);
\node(d1) at (-0.65,1.2) {\scriptsize$\dots$};
\node(d2) at (0.55,1.2) {\scriptsize$\dots$};

\draw[thick] (6.5,1.2) to [bend left=55] (7.7,1.2);
\draw[thick] (7.1,1.2) to [bend left=55] (8.3,1.2);
\node(d1) at (6.85,1.2) {\scriptsize$\dots$};
\node(d2) at (8.05,1.2) {\scriptsize$\dots$};

\draw[thick] (1.2,1.2) to [bend left=55] (2.5,1.2);
\draw[thick] (1.8,1.2) to [bend left=55] (3.1,1.2);
\node(d1) at (1.55,1.2) {\scriptsize$\dots$};
\node(d2) at (2.85,1.2) {\scriptsize$\dots$};

\draw[thick] (4.2,1.2) to [bend left=55] (5.7,1.2);
\draw[thick] (4.8,1.2) to [bend left=55] (6.3,1.2);
\node(d1) at (4.55,1.2) {\scriptsize$\dots$};
\node(d2) at (6.05,1.2) {\scriptsize$\dots$};

\draw[thick] (2.5,1.2) to [bend right=55] node[below] {\scriptsize$cut$} (4.2,1.2);
\draw[thick] (3.1,1.2) to [bend right=55] node[below] {\scriptsize$cut$} (4.8,1.2);

\end{tikzpicture}}
$$
Now apply lemma \ref{algo} to obtain a reduced graph with splitting cuts. The cut types $I_{1},\dots, I_{p}$ are then the interpolants of $f$.

\end{proof}

\end{document}